\documentclass[aps,prx,10pt,twocolumn, superscriptaddress, nofootinbib,longbibliography,floatfix,amsmath,amssymb]{revtex4-2}
\usepackage{graphicx} 
\AtBeginDocument{\providecommand{\selectlanguage}[1]{}\renewcommand{\selectlanguage}[1]{}} 

\usepackage{multirow}
\usepackage[utf8]{inputenc}
\usepackage[dvipsnames,svgnames]{xcolor}
\usepackage[T1]{fontenc}
\usepackage{amsmath,amssymb,amsthm}
\usepackage{algorithm,algpseudocode}
\usepackage[colorlinks]{hyperref}
\hypersetup{
    pdfstartview={FitH},
    colorlinks=true,
    linkcolor=blue,
    citecolor=blue,
    filecolor=magenta,
    urlcolor=blue
}
\usepackage{enumitem}
\usepackage{booktabs,array}
\usepackage{mathtools}
\usepackage{cleveref}
\usepackage{tikz-cd}
\usepackage{xspace}
\crefname{figure}{Fig.}{Figs.}
\Crefname{figure}{Fig.}{Figs.}
\crefname{section}{Section}{Sections}
\Crefname{section}{Section}{Sections}
\crefname{table}{Table}{Tables}

\Crefname{table}{Table}{Tables}
\crefname{appendix}{Appendix}{Appendices}
\Crefname{appendix}{Appendix}{Appendices}
\crefname{equation}{Eq.}{Eqs.}
\Crefname{equation}{Eq.}{Eqs.}
\crefname{theorem}{Theorem}{Theorems}
\Crefname{theorem}{Theorem}{Theorems}
\crefname{definition}{Definition}{Definitions}
\crefname{prop}{Proposition}{Propositions}
\crefname{conjecture}{Conjecture}{Conjectures}
\crefname{corollary}{Corollary}{Corollaries}
\usepackage{makecell}
\usepackage{braket}
\DeclareMathOperator{\im}{\text{im}}

\newcommand{\R}{R}
\usepackage{pifont}
\newcommand{\cmark}{\ding{51}}%
\newcommand{\xmark}{\ding{55}}%
\newcommand{\wt}{\operatorname{wt}}
\newcommand{\outline}[1]{\section*{#1}}
\AtBeginDocument{%
  \heavyrulewidth=.08em
  \lightrulewidth=.05em
  \cmidrulewidth=.03em
  \belowrulesep=.65ex
  \belowbottomsep=0pt
  \aboverulesep=.4ex
  \abovetopsep=0pt
  \cmidrulesep=\doublerulesep
  \cmidrulekern=.5em
  \defaultaddspace=.5em
}

\newtheorem{theorem}{Theorem}[section]
\newtheorem{prop}{Proposition}

\newtheorem*{conjecture*}{Conjecture}

\newtheorem{lemma}[theorem]{Lemma}

\newcommand{\F}{\mathbb{F}_2}
\newcommand{\FG}{\mathbb{F}_2[G]}

\DeclareMathOperator{\Ann}{Ann}

\newcommand{\code}[1]{[\![#1]\!]}\newcommand{\bZ}{\bar Z}\newcommand{\bX}{\bar X}
\newcommand{\T}{\mathsf T}
\DeclareRobustCommand{\tstate}{%
  \ifmmode
    \ket{\bar{T}}%
  \else
    \ensuremath{\ket{\bar{T}}}\xspace
  \fi
}

\DeclareRobustCommand{\cczstate}{%
  \ifmmode
    \ket{\overline{\mathrm{CCZ}}}%
  \else
    \ensuremath{\ket{\overline{\mathrm{CCZ}}}}\xspace
  \fi
}

\DeclareRobustCommand{\detstate}{%
  \ifmmode
    \ket{\overline{\textnormal{det}}}%
  \else
    \ensuremath{\ket{\overline{\textnormal{det}}}}\xspace
  \fi
}

\begin{document}

\title{Single-shot magic state factories for 
quantum LDPC codes}

\author{
Varun~Menon}
\thanks{These authors contributed equally; varunmenon@g.harvard.edu}
\affiliation{Department of Physics, Harvard University, Cambridge, MA 02138, USA}
\author{
Rohan~Mehta}
\thanks{These authors contributed equally; rohanmehta@g.harvard.edu}
\affiliation{Department of Physics, Harvard University, Cambridge, MA 02138, USA}
\author{
Andi~Gu}
\affiliation{Department of Physics, Harvard University, Cambridge, MA 02138, USA}
\author{
Andrei~C.~Diaconu}
\affiliation{Department of Physics, Harvard University, Cambridge, MA 02138, USA}
\author{
Daniel~Bochen~Tan}
\affiliation{Department of Physics, Harvard University, Cambridge, MA 02138, USA}
\author{Xiao~Xiao}
\affiliation{QuEra Computing Inc., 1380 Soldiers Field Road, Boston, MA 02135, USA}
\affiliation{Joint Center for Quantum Information and Computer Science,
University of Maryland, College Park, MD 20742, USA}
\author{Michael~J.~Gullans}
\affiliation{QuEra Computing Inc., 1380 Soldiers Field Road, Boston, MA 02135, USA}
\affiliation{Joint Center for Quantum Information and Computer Science,
University of Maryland, College Park, MD 20742, USA}
\author{
Qian~Xu}
\affiliation{\mbox{Institute for Quantum Information and Matter, California Institute of Technology, Pasadena, CA 91125, USA}}
\affiliation{Walter Burke Institute for Theoretical Physics, California Institute of Technology, Pasadena, CA 91125, USA}
\author{
Hengyun~Zhou}
\affiliation{Department of Electrical Engineering and Computer Science, Massachusetts Institute of Technology, Cambridge, MA 02138, USA}
\author{
Mikhail~D.~Lukin}
\affiliation{Department of Physics, Harvard University, Cambridge, MA 02138, USA}
\author{
J.~Pablo~Bonilla~Ataides}
\thanks{jbonillaataides@g.harvard.edu}
\affiliation{Department of Physics, Harvard University, Cambridge, MA 02138, USA}

\begin{abstract}
Quantum low-density parity-check (qLDPC) codes can substantially reduce the qubit overhead of fault-tolerant quantum computation, but any useful architecture based on such codes also requires an efficient mechanism to implement non-Clifford gates in a fault-tolerant manner. We introduce a space-time-efficient magic state factory for high-rate qLDPC codes by combining the transversal non-Clifford gates of recently developed \textit{tricycle} codes with single-shot four-dimensional balanced product \textit{quadcycle} codes. The key ingredient is an efficient ``escape'' mechanism that teleports postselected logical magic states from tricycle codes to larger distance quadcycle codes. Under a standard circuit-level noise model at physical error rate $10^{-3}$, our factories produce \cczstate states at logical error rates as low as $10^{-9}$ with 10--40 times lower space-time cost than state-of-the-art magic state generation methods. The resulting states are directly encoded in high-rate blocks suitable for computation. We identify the quadcycle codes by an AI-assisted search and show that they can serve as excellent memories capable of single-shot error correction. Finally, we present a concrete implementation of the entire factory with highly efficient atom transport on reconfigurable neutral-atom quantum computers. Together, these results provide an end-to-end, practical protocol for producing and utilizing high-fidelity magic states within a high-rate qLDPC architecture, with the lowest space-time cost to date.
\end{abstract}

\maketitle

\outline{Introduction}\label{sec:intro}

Quantum error correction is essential for scalable quantum computation, protecting information from decoherence and from the errors introduced by noisy operations. The paradigmatic approach pairs the surface code \cite{fowler_surface_2012} with efficient Clifford logic, realized either transversally \cite{zhou2025resource} or by lattice surgery \cite{horsman_surface_2012, litinski_game_2019, gidney2025factor}. The required non-Clifford operations are performed by utilizing logical resource states, prepared through procedures such as magic state distillation \cite{bravyi2005universal, bravyi2012magic, sales2025experimental}. Such procedures require substantial overhead in both space, through the number of physical qubits used, and time, through multiple syndrome-extraction rounds or repeated attempts in repeat-until-success schemes. Recently, much effort has been devoted to reducing these costs. One of the most promising approaches involves magic state \textit{cultivation}, which hierarchically suppresses errors on an initially noisy magic state \cite{sahay2026fold, gidney_magic_2024, vaknin2026cultivation, claes2025cultivating, chen2025efficient}. Despite these efforts, scaling surface code architectures to thousands of logical qubits and the error rates that large-scale algorithms demand remains resource intensive, and magic state production continues to dominate the overhead.

High-rate quantum low-density parity-check (qLDPC) codes offer a complementary route to fault tolerance. By encoding many logical qubits in a single block while maintaining distances that grow favorably with block size, they reach low logical error rates with substantially fewer physical qubits than comparable surface code memories. Advances in circuit-level simulation and decoding have identified finite-size qLDPC codes with high encoding rates and thresholds, and low logical error rates under experimentally relevant noise levels~\cite{gidney_stim_2021, roffe2005decoding, xu2024constant, panteleev2021degenerate, bravyi_high-threshold_2024, kasai_breaking_2026, steffan2025tile}. The nonlocal checks required for realization of these codes can be implemented naturally on platforms with reconfigurable connectivity such as neutral-atom arrays \cite{xu2024constant, cain_shors_2026, menssen2026strategic, menon_magic_2025, zhao2026towards}.

The use of these codes in complete architectures is limited by two critical elements. First, while recently proposed Pauli-product measurements and code-surgery protocols allow for efficient realization of
 Clifford logic operations on qLDPC codes \cite{cohen2022low, williamson2026low, cowtan_parallel_2026, cross_improved_2025, zheng_high-rate_2025, xu2025fast, baspin2025fast, cowtan_fast_2025, zhou2026parallel, maskara2026efficient},   no native, high-throughput source of non-Clifford states is currently known that preserves the rate advantage of the surrounding code. Second, the operations on qLDPC codes typically suffer from additional temporal overhead as
 conventional protocols repeat noisy stabilizer measurements a number of times growing with the code distance. This is in contrast to fully transversal architectures~\cite{zhou_low-overhead_2025} or single-shot error correction that extracts sufficient redundancy from a constant number of rounds \cite{bombin2015single}. Single-shot behavior is known for higher-dimensional topological and expander-code constructions \cite{brown2016fault, breuckmann2016local, leverrier2022quantum}, but high-rate codes combining it with strong finite-size performance have remained elusive.

\begin{figure*}[t!]
    \centering
     \includegraphics[width=\textwidth]{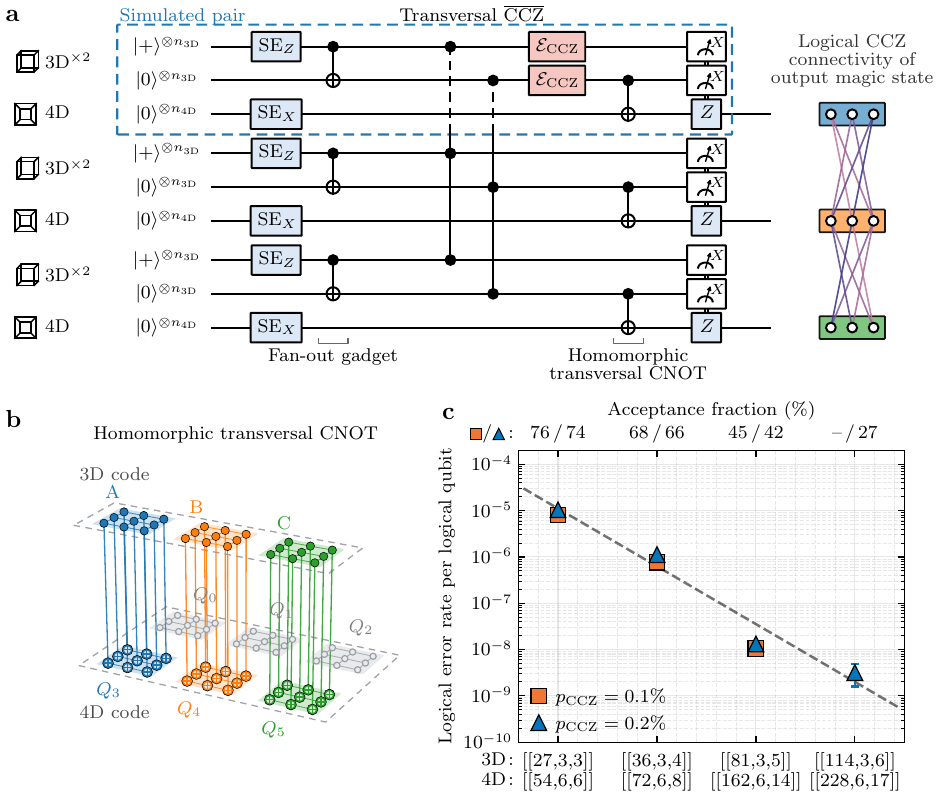}
    \caption{\textbf{Transversal dimension-jump magic state factory.}
(a) Three tricycle blocks are prepared in $\ket{\overline{+}}$ and acted on by a transversal $\overline{\mathrm{CCZ}}$ circuit using a fan-out gadget, producing the nine-qubit hypergraph state $\detstate$ shown on the right. The gadget pairs each block with an ancilla patch of the same code through one transversal layer of CNOTs, splitting the two CCZ layers across the two patches so that every qubit takes part in a single gate. A depth-1 homomorphic CNOT teleports the magic state into three paired quadcycle blocks prepared in $\ket{\overline{0}}$. Transversal $X$ measurements of the tricycle blocks provide the teleportation outcomes and detector information used for postselection. The dashed box marks the pair used in the simulations. (b) Physical geometry of the homomorphic CNOT. The three qubit sectors of the tricycle code are aligned with their embedded counterparts in the quadcycle code, so matched qubits are coupled by a single layer of site-local CNOT gates; the remaining quadcycle qubits are idle. (c) Post-selected logical error rate per output logical qubit for the simulated code pairs, plotted against the tricycle code distance $d_{\textnormal{3D}}$, with the code parameters given below the axis. Reset, measurement, idle and two-qubit gates undergo depolarizing noise at $p=10^{-3}$. The CCZ channel is the single-qubit marginal of a three-qubit depolarizing channel of strength $p_{\mathrm{CCZ}}$. The top axis gives the acceptance fraction of a \emph{single} tricycle--quadcycle pair, for $p_{\mathrm{CCZ}}=0.1\%$ and $0.2\%$ respectively. The factory requires all three pairs to be accepted, so its yield is this value cubed. Error bars are binomial, and the dashed line is a log-linear guide through the $0.2\%$ series.}
    \label{fig:dimjump}
\end{figure*}

In this Article, we construct an end-to-end space-time-efficient magic state factory using single-shot qLDPC codes. Our approach makes use of three-dimensional balanced-product \emph{tricycle} codes that support transversal non-Clifford gates and generate several logical magic states in parallel and in constant depth \cite{menon_magic_2025, jacob2025singleshot}. However, their strongly asymmetric distances limit the fidelity with which those states can be produced and held. We pair these tricycle codes with four-dimensional balanced-product \emph{quadcycle} codes, which share algebraic structure with tricycle codes, enabling an efficient state-teleportation mechanism between them. The shared structure enables a depth-1 homomorphic CNOT that ``escapes'' the magic into the quadcycle block, which supplies larger, balanced distances and single-shot preparation and error correction in both Pauli bases.
Thanks to these single-shot properties, this entire mechanism can be realized in constant depth and is natively compatible with postselection, further refining the magic state. 
With an AI-assisted search, we identify finite-size quadcycle codes with favorable rates, thresholds, and logical error rates under circuit-level noise.

Our factories produce entangled magic states \cite{rossi_quantum_2013} across nine logical qubits, corresponding to a gate that implements the determinant of a $3\times3$ matrix over $\F$, denoted $\detstate$. We demonstrate that with homomorphic measurement \cite{huang_homomorphic_2023} gadgets, these states can be converted into two usable $\ket{\overline{\mathrm{CCZ}}}$ states. We further construct low-overhead protocols for post-processing and injection into qLDPC computation codes, as well as the required Clifford feed-forward operations. Finally, we provide depth-optimal syndrome-extraction circuits and an end-to-end compilation based on the reconfigurable neutral atom array architecture, including explicit planar layouts and atom rearrangements required from initialization through output. Under a standard circuit-level noise model at $p=10^{-3}$, the resulting factories deliver \cczstate{} states at roughly an order of magnitude lower space-time cost than the state-of-the-art magic state cultivation protocols and at a lower logical error rate, already encoded in a high-rate block ready for downstream computation.

\outline{Single-shot magic state factory}\label{sec:single-shot}

Our factory separates magic state generation from storage and processing, assigning each task to a code optimized for it. Tricycle codes support transversal non-Clifford gates and can therefore generate multiple logical magic states in parallel, but their strongly asymmetric distances limit the fidelity with which these states can be retained, typically requiring full postselection on stabilizer measurement outcomes to achieve low logical error rates. Quadcycle codes, on the other hand, provide larger, more balanced distances, together with single-shot state preparation and error correction capabilities for both $X$ and $Z$ basis stabilizers, but they lack transversal non-Clifford gates. Fortunately, these two families can be paired in such a way as to harness their complementary benefits. We demonstrate that every tricycle code can be embedded into a corresponding quadcycle code, enabling a depth-$1$ transversal homomorphic $\mathrm{CNOT}$ between the two \cite{li_transversal_2025} that we use to teleport magic from the tricycle codes to the quadcycle codes.

This teleportation mechanism enjoys two distinct advantages. First, measurement outcomes from the transversal teleportation circuit provide the necessary information for efficient and accurate postselection, suppressing the infidelity of the encoded magic state by several additional orders of magnitude. Second, the postselected magic state is now natively stored in the quadcycle code, which has large enough distance to maintain the fidelity of the magic state in subsequent rounds of error correction. The quadcycle is also better suited for downstream processing.

We consider a class of tricycle codes that support a depth-$2$ code-space-preserving $\mathrm{CCZ}$ circuit whose logical action produces the $\detstate$ hypergraph magic state. To improve the fault-tolerance properties of the $\mathrm{CCZ}$ circuit, we use a fan-out gadget on top of the tricycle code blocks, which effectively amounts to concatenation with a distance $2$ bit-type ($Z$-basis) repetition code (see \cref{fig:dimjump}). The fan-out gadget converts the $\mathrm{CCZ}$ circuit of the tricycle codes into a strictly transversal depth $1$ circuit, which preserves the circuit distance of the tricycle code and leads to significantly greater error suppression. This comes at the cost of doubling the data register, though not the peak qubit count, since the ancilla patch reuses the qubits that served as $Z$-syndrome ancillae -- see Section~\ref{app:concat} of the Supplementary Information~\cite{SI} for details.

The protocol, summarized in Fig.~\ref{fig:dimjump}, proceeds in three stages. First, three tricycle blocks are prepared in $\ket{\overline{+}}$, each with its fan-out gadget applied, and acted on by a depth-$1$ transversal $\mathrm{CCZ}$ circuit producing the logical hypergraph state $\detstate$. In parallel, three paired quadcycle blocks are prepared in $\ket{\overline{0}}$. A depth-1 transversal CNOT then teleports the logical state into the quadcycle blocks, after which transversal measurements of the tricycle blocks provide both the teleportation outcomes and the syndrome information used for postselection. The resulting logical $\detstate$ state can be processed to yield two disjoint $\cczstate$ states for further computation.

\outline{Quadcycle escape codes}\label{sec:quadcycle-codes}

\begin{table*}[ht!]
\centering
\renewcommand{\arraystretch}{1.3}
\resizebox{\textwidth}{!}{%
\begin{tabular}{@{}lccccc@{}}
\toprule
\textbf{Method}
& \textbf{Output state}
& \multicolumn{2}{c}{\textbf{Encoded in}}
& \textbf{\makecell{CCZ-converted \\ error rate}}
& \textbf{\makecell{Expected space-time \\ volume per \cczstate}} \\
\cmidrule(lr){3-4}
&
& \textbf{Code}
& \textbf{Single-shot}
&
& \\
\midrule
\textbf{This work ($f=3$)}
& 2 \cczstate
& Quadcycle
& \cmark
& $1.2 \times 10^{-5}$
& $4.3 \times 10^3$ \\

Surface code cultivation ($f=3$) \cite{sahay2026fold}
& $\ket{\bar T}$
& Surface
& \xmark
& $1.2 \times 10^{-5}$
& $80.4 \times 10^3$ \\

Color code cultivation ($f=3$) \cite{gidney_magic_2024}
& $\ket{\bar T}$
& Grafted
& \xmark
& $1.2 \times 10^{-5}$
& $162.8 \times 10^3$ \\

\midrule
\textbf{This work ($f=4$)}
& 2 \cczstate
& Quadcycle
& \cmark
& $1.2 \times 10^{-6}$
& $8.6 \times 10^3$ \\

Surface code cultivation ($f=5$) \cite{sahay2026fold}
& $\ket{\bar T}$
& Surface
& \xmark
& $1.2 \times 10^{-6}$
& $83.8\times 10^3$ \\

\midrule

\textbf{This work ($f=5$)}
& 2 \cczstate
& Quadcycle
& \cmark
& $1.6 \times 10^{-8}$
& $6.7 \times 10^4$ \\

Color code cultivation ($f=5$) \cite{gidney_magic_2024}
& $\ket{\bar T}$
& Grafted
& \xmark
& $1.6\times 10^{-8}$
& $228.6 \times 10^4$ \\



\bottomrule
\end{tabular}
}
\caption{Comparison of our magic state factories with various state-of-the-art magic state generation protocols at physical error rates of $10^{-3}$, as in Ref.~\cite{sahay2026fold}. $f$ denotes the nominal fault distance of the protocol, which we compute exactly for our factories \cite{chan_diagnosing_2026}. We normalize the expected space-time volume of other protocols to account for the fact that our factories produce two extractable \cczstate states as opposed to a single $\ket{\bar T}$ state. Since $\geq 4$ $\ket{\bar T}$ states are needed to synthesize a $\cczstate$ \cite{beverland2020lower}, we multiply the expected space-time volume per $\ket{\bar T}$ state for the other protocols by $4$ to attain the numbers in the last column. The LER and space-time volume contributions from the $T$-to-CCZ conversion circuits are not included. Expected space-time volume numbers for cultivation-based protocols are estimated from \cite{sahay2026fold} (Figure 2(b)), and logical error rates are obtained from simulations with mid-circuit postselection \cite{mehta_practical_2026}. Additional overheads from decoder latency are not included for any scheme.}
\label{tab:comparison}
\end{table*}

The receiving code in our architecture must satisfy several requirements simultaneously. It must share sufficient algebraic structure with the tricycle code to support a transversal dimension jump~\cite{li_transversal_2025}, encode many logical qubits so that the factory remains high rate, support fault-tolerant state preparation, allow decoding using information from only a constant number of error-correction rounds, and have sufficiently large and balanced distances to preserve the fidelity obtained through postselection.

To meet these requirements, we study a particular family of Abelian multi-cyclic multivariate codes \cite{mian2026multivariate, lin2025abelian}, which for this work we refer to as \emph{quadcycle codes}. They can be understood as a four-dimensional generalization of bicycle (2D) and tricycle (3D) codes \cite{bravyi_high-threshold_2024, viszlai2024matchinggeneralizedbicyclecodesneutral, menon_magic_2025, jacob2025singleshot}, and they belong to the broad categories of balanced product \cite{breuckmann2021balanced} and lifted product codes \cite{panteleev2021quantum}. A quadcycle code is specified by four commuting sparse circulant matrices $A,B,C,D$ acting on a finite abelian group of size $n_G$. It contains $6n_G$ physical qubits arranged in six blocks, together with $4n_G$ stabilizer generators of each Pauli type. The first three factors, $A,B,C$, coincide with those defining a paired tricycle code, while $D$ introduces an additional cyclic direction. This shared structure embeds the tricycle complex within the quadcycle complex and enables the depth-1 homomorphic CNOT used for the escape operation. The complete parity-check matrices and mapping cone construction are given in Methods.

The quadcycle checks are redundant in both Pauli bases. In each basis, $n_G$ explicit relations among the stabilizer outcomes define \emph{metachecks}, which provide consistency conditions on a noisy measured syndrome. This bidirectional syndrome redundancy allows measurement faults to be identified without temporal repetition and supports single-shot preparation and error correction in both the $X$ and $Z$ bases. Consequently, the receiving blocks can be fault-tolerantly initialized and error corrected using only a constant number of syndrome-extraction rounds.

Quadcycle codes also inherit useful symmetries from their group-algebra construction. Group translations act as depth-1 code automorphisms, while an $X$--$Z$ duality supplies transversal Clifford operations and ensures equal $X$ and $Z$ distances (Methods). The finite-size codes identified in our search combine these properties with high encoding rates and substantially larger, more balanced distances than their paired tricycle codes.

The two code families therefore play complementary roles. The shared $A,B,C$ factors enable all logical magic states produced by the tricycle code to be transferred transversally, while the larger distance and single-shot structure of the quadcycle code in both the $X$ and $Z$ bases allow the accepted states to be verified, stored, and processed without sacrificing their fidelity.

Finally, we confirm that the escape codes are sufficiently good quantum memories to hold the magic states they receive. Figure~\ref{fig:4dsims} reports the circuit-level logical error rate per cycle of representative quadcycle codes under single-shot, sliding-window most-likely-error decoding. At the physical error rate $p=10^{-3}$ used in our factory simulations, this per-cycle rate is comparable to or below the logical error rate of the states the factory delivers (Fig.~\ref{fig:dimjump}c). This addresses the natural concern for downstream use: once a low-error magic state has been teleported into the quadcycle code, the code must have enough error-correction tolerance to \emph{maintain} that low error rate through subsequent storage and processing, rather than letting it grow. The comparable per-cycle rate shows that it does---the quadcycle code holds the factory-established fidelity across the additional error-correction cycles---and, because these codes are single-shot in both bases, each cycle of error correction costs only a constant number of syndrome-extraction rounds, independent of the code distance. In Section~\ref{sec:search} of the Supplementary Information~\cite{SI}, we additionally present several quadcycle codes which are not constrained to have a homomorphism to specific tricycle codes, found by an AI-assisted code search to have high rates and relative distances, strong single-shot circuit-level performance, and structured logical operators, demonstrating that these codes are good candidates for single-shot state-preparable LDPC quantum memories in their own right -- in particular we find quadcycle codes with parameters that improve upon those in Refs. \cite{mian2026multivariate, lin2025abelian}.

\begin{figure}[t!]
    \centering
    \includegraphics[width=\linewidth]{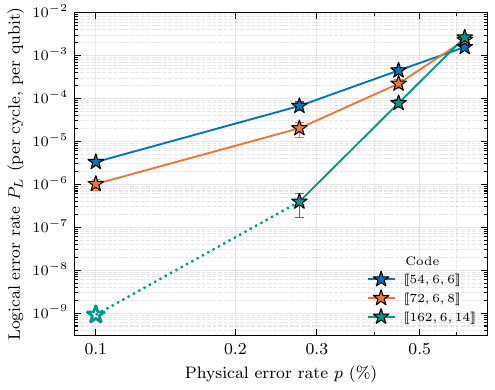}
    \caption{\textbf{Single-shot circuit-level performance of the quadcycle codes.} Logical error rate per cycle per logical qubit versus physical error rate $p$ for the $[[54,6,6]]$, $[[72,6,8]]$, and $[[162,6,14]]$ quadcycle codes, from windowed memory simulations under circuit-level depolarizing noise. Decoding is single-shot: a sliding window spanning three syndrome-extraction rounds and advancing one round per step ($W=3$, $S=1$) solves for the most-likely error within each window, so a constant number of rounds suffices independent of the code distance. The open marker extrapolates the $[[162,6,14]]$ curve to $p=10^{-3}$ using $P_L\propto p^{d/2}$ with $d=12$ (its circuit-level distance upper bound). Error bars are binomial.}
\label{fig:4dsims}
\end{figure}

\outline{Constant-depth preparation and postselection}

We now detail the three stages of the protocol and the mechanisms that make them fault tolerant. For the code pairs considered here, each tricycle block has parameters $[[n_{\textnormal{3D}},3, (d_{\textnormal{3D}}^X,d_{\textnormal{3D}}^Z)]]$ with $d_{\textnormal{3D}}^Z\ll d_{\textnormal{3D}}^X$, and is prepared in $\ket{\overline{+}}$ alongside the ancilla patch of the fan-out gadget in $\ket{\overline{0}}$. A single transversal layer of CNOTs between them fixes the gadget's checks deterministically, requiring no additional syndrome extraction or feed-forward. All $Z$ stabilizers are fixed to $+1$ by one round of $Z$-stabilizer extraction, which is fault tolerant because the tricycle code is single-shot state-preparable in the $Z$ basis. A depth-1 transversal physical CCZ circuit then implements the logical non-Clifford operation that produces the nine-qubit hypergraph state $\detstate$ (Fig.~\ref{fig:dimjump}a). It is the fan-out gadget that reduces this circuit from depth two to depth one and thereby preserves the circuit-level distance of the factory (Section~\ref{app:concat} of the Supplementary Information~\cite{SI}). Meanwhile, in parallel, three paired quadcycle blocks are prepared in $\ket{\overline{0}}$ using one round of $X$-stabilizer extraction. For the pairs studied here, each quadcycle block has the same number of physical qubits as a tricycle block together with its gadget ancilla patch, twice its logical dimension, and approximately balanced distances comparable to $d_{\textnormal{3D}}^X$.

The dimension jump teleportation is implemented by a single transversal layer of physical
CNOT gates whose controls act on qubits of the tricycle code, and the targets act on part of the paired quadcycle block. The exact physical circuit that implements this homomorphic CNOT is described by the \textit{chain map} that identifies the tricycle code with a subset of the quadcycle code. Logically, this introduces a logical CNOT between all $k_{\textnormal{3D}}$ tricycle logical qubits and one of the two
$k_{\textnormal{3D}}$-dimensional sectors of the quadcycle logical qubits (Fig.~\ref{fig:dimjump}b and Methods). We confirm the logical action of this map covers all of the logical qubits of the tricycle code. This ensures every logical magic state generated in
the tricycle block is extracted and teleported to the quadcycle code.

After the dimension jump, the tricycle blocks are measured transversally in
the $X$ basis. We reject any run with a nontrivial detector outcome, while
for accepted runs the decoded logical measurement outcomes determine
Pauli-frame $\overline{Z}$ corrections on the quadcycle blocks. This
postselection substantially increases the effective $Z$ fault distance of
the protocol. A distance-$d$ code can correct at most
$\lfloor(d-1)/2\rfloor$ faults, whereas full error detection rejects every
nontrivial fault pattern of weight below $d$. Conditioning on acceptance
therefore approximately doubles the minimum number of faults required for a
logical failure to survive undetected. We confirm this scaling directly. Rescaling every noise rate by a common factor gives a fitted error exponent of $n_{\mathrm{eff}}=3.11\pm0.04$ for the $d_{\textnormal{3D}}=3$ factory and $3.91\pm0.07$ for $d_{\textnormal{3D}}=4$, matching the code distance rather than the $\sim d/2$ expected of a decoder without error detection (Section~\ref{app:concat} of the Supplementary Information~\cite{SI}).

Full error detection remains practical here because the postselected space-time region is relatively small: the tricycle blocks are modest in size and participate in only a shallow, constant-depth circuit before being measured and discarded. The resulting rejection probability is therefore manageable despite postselecting on all detector outcomes. In our protocol, the resulting suppression compensates for the strongly asymmetric distance of the tricycle code and brings the effective $Z$-error suppression closer to that set by its larger $X$ distance and by the distance of the quadcycle block.

We simulate each tricycle--quadcycle pair under circuit-level depolarizing noise at $p=10^{-3}$, with an effective channel for each physical CCZ gate of strength $p_{\mathrm{CCZ}}=p$ and $2p$. Results for the four pairs are shown in \cref{fig:dimjump}(c).
We note that since this approach uses only a small number of physical CCZ gates, their increased error rates do not significantly affect the resulting logical-level performance (Methods).

Finally, we note that the same shared algebraic structure between the tricycle and quadcycle codes enables a direct code-growth protocol where the tricycle code is grown into a quadcycle code by measuring additional stabilizers, without requiring an additional code block. Although direct growth may use fewer qubits, we find that the 3D--4D state-teleportation scheme presented here has lower space-time cost and leads to slightly better logical error rates of the produced magic states -- see Section~\ref{sec:code_growth} of the Supplementary Information~\cite{SI} for details on the direct code-growth procedure.

\outline{Comparison to state-of-the-art magic state preparation schemes}
We compare the magic state preparation protocol outlined in the preceding section to other state-of-the-art approaches (Table \ref{tab:comparison}). We focus on fold-transversal surface code cultivation, which uses nonlocal connectivity \cite{sahay2026fold}, and the original color code cultivation approach, which imposes more stringent connectivity constraints \cite{gidney_magic_2024}. The logical error rates and associated overheads we tabulate for the cultivation-based protocols are based on Ref. \cite{mehta_practical_2026}, which studies magic cultivation \cite{gidney_magic_2024, sahay2026fold} with mid-circuit postselection using only information that is accessible to a decoder in practice. This treatment places cultivation on equal footing with our simulations in Fig. \ref{fig:dimjump}(c), where we only use information from the measurements involved in logical teleportation, and not the transversal readout of the quadcycle code blocks, to inform our postselection. Including information from the transversal measurement of the quadcycle code, which is analogous to the inclusion of readout information in past simulations of magic state cultivation \cite{gidney_magic_2024, sahay2026fold}, leads to  orders of magnitude lower logical error rates than we observe here \cite{li_transversal_2025, menon_magic_2025, xiao2026efficient}.  

Since cultivation produces $\ket{\bar T}$ states, whereas our factory natively produces $\detstate$, we normalize the resource accounting--in terms of both the logical error rate and the expected resource cost---of these protocols by their CCZ yield, i.e., the number of $\cczstate$ states that can be extracted from the resulting magic state with additional Clifford operations. This is motivated by the fact that most primitives for quantum algorithms of interest are naturally compiled to $\mathrm{CCZ}$ gates, such as arithmetic operations for cryptographic applications \cite{gidney2025factor, zhou2025resource} or multiplexed rotations in quantum simulation applications \cite{gidney2018halving, kivlichan2020improved}. Our normalization reflects that at least $4$ $\ket{\bar T}$ states are necessary to synthesize a CCZ gate \cite{t_to_ccz}, whereas $2$ \cczstate states can be obtained from the \detstate state. More details on the space-time cost accounting resulting in the numbers in \cref{tab:comparison} can be found in Methods.

\begin{figure}[t!]
    \centering
    \includegraphics[width=\linewidth]{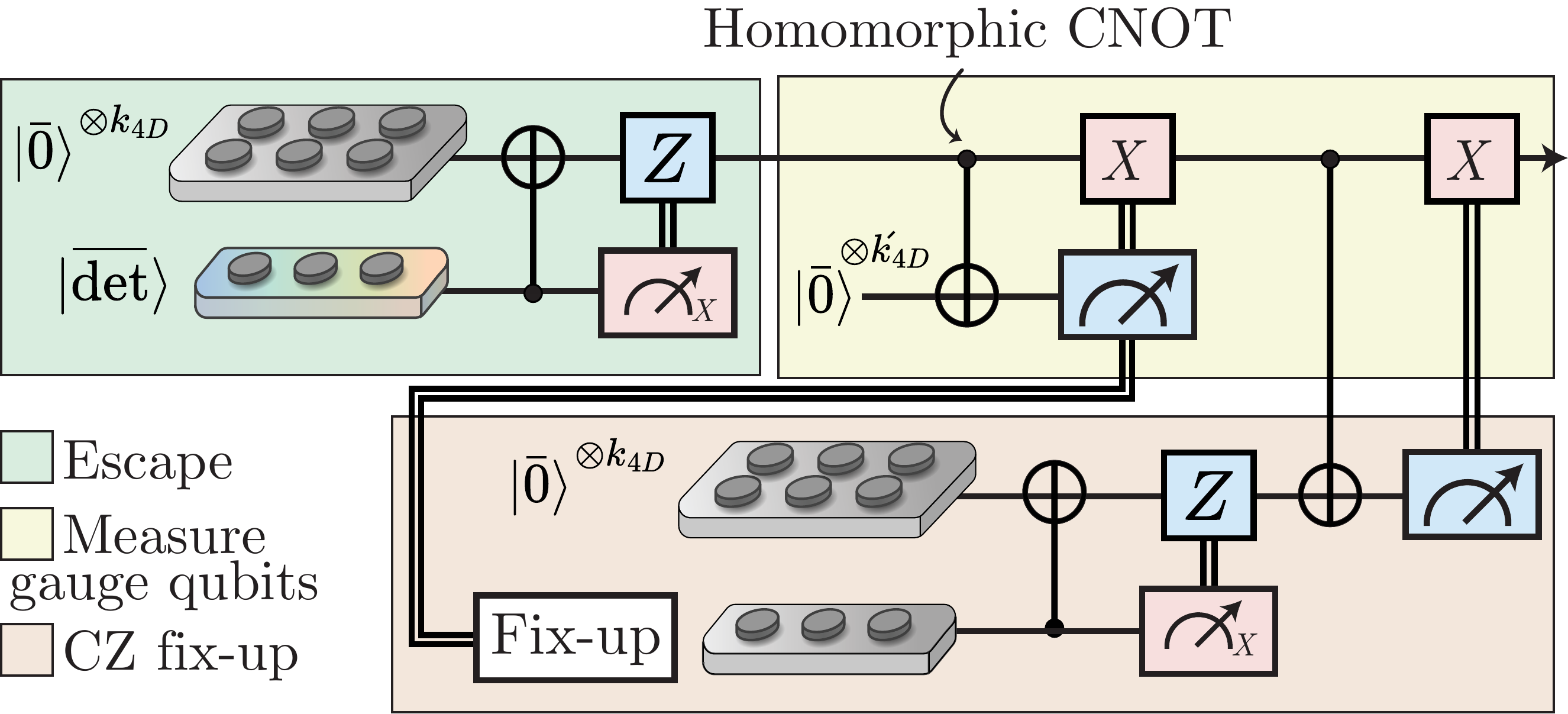}
    \caption{\textbf{Complete factory pipeline with downstream fault-tolerant gadgets.} Only one of the three parallel sets of code blocks is shown. The escape stage abstracts Fig.~\ref{fig:dimjump}, in which the magic state is teleported to three quadcycle blocks and postselected (not shown). Using ancillary quadcycle codes and homomorphic CNOTs, we then perform a transversal measurement that effectively addresses a single logical qubit in the $Z$ basis. As the ancillary quadcycle code is also single-shot state preparable, this measurement can be performed in constant depth. Then, depending on the logical outcome, we modify the factory to transversally prepare a resource state for the required addressable CZ corrections (Methods). Retries of the escape stage are not shown. Finally, the correction resource state is consumed transversally through gate teleportation, removing the undesired CZ gates dressing the disjoint \cczstate states. The resulting magic state has two disjoint \cczstate states and is ready to be consumed, e.g., via surgery with another code. }
    \label{fig:gadgets}
\end{figure}

Note that cultivation natively offers a mechanism to tune the overall logical error rate of the protocol at the expense of the expected resource overhead, by adjusting the retry strategy. We choose the operating point where magic state cultivation with mid-circuit postselection \cite{mehta_practical_2026} results in logical error rates equal to that of our factories at various sizes in order to offer a direct comparison.

\outline{Consuming magic downstream}\label{sec:computation}
\begin{figure*}[t]
\centering
\includegraphics[width=\textwidth]{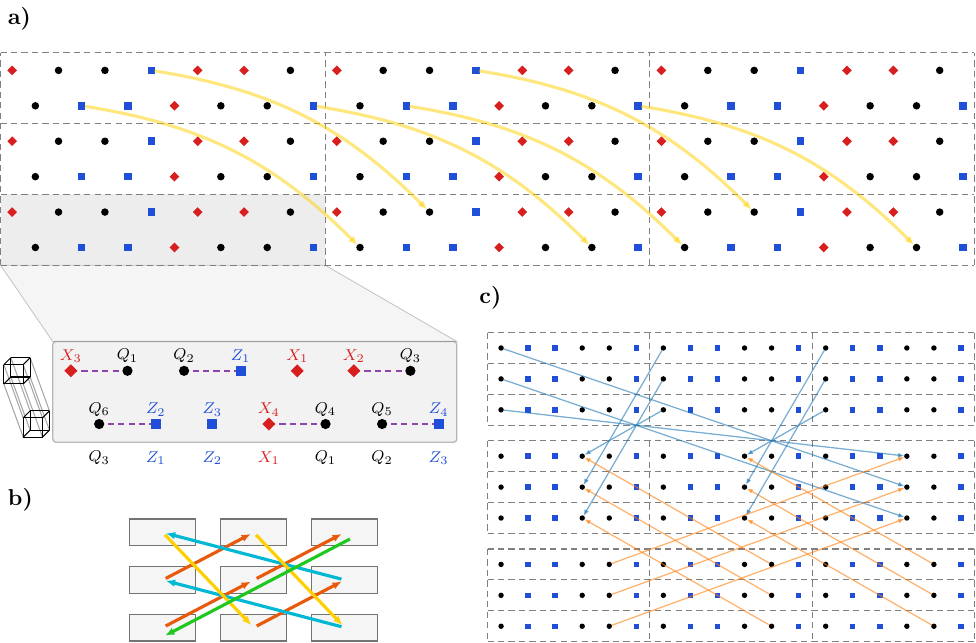}
\caption{\label{fig:compilation}\textbf{Qubit placement and routing for the neutral-atom implementation.}
\textbf{(a)}~Toric layout of a quadcycle code. Qubits are arranged into
    unit cells containing six data qubits, four $X$-check ancillae, and four
    $Z$-check ancillae. The corresponding tricycle cell, shown below, is
    identified with the lower row of the quadcycle cell. This shared layout
    allows the two codes to use the same routing patterns for their common
    stabilizer terms and aligns their data qubits for the homomorphic CNOT.
    The dashes inside the unit cell display which sectors are coupled for an $A$ term.
    \textbf{(b)}~Parallel atom rearrangement implementing the cyclic
    translation associated with the monomial $xy$. The translations
    are divided among at most four AODs so that all required two-qubit
    interactions can be realized in parallel. An interactive animation of the parallel atom rearrangement is included in the ancillary files with this manuscript.
    \textbf{(c)}~Routing for one layer of the transversal
    $\overline{\mathrm{CCZ}}$. Three tricycle blocks are stacked vertically,
    and atoms from the upper and lower blocks are moved into interaction range
    of the middle block. The three upper and three lower sector pairings are
    specified by six monomial translations (two are drawn), which are executed sequentially
    within each CCZ layer.}
\end{figure*}

Efficient magic state generation must be accompanied by a pipeline for consuming these states with low overhead. While a broader qLDPC architecture supporting general Clifford logic can accomplish this task \cite{webster2026pinnacle, cain_shors_2026}, we show that in our case it can be done efficiently with minimal further assumptions, in three phases: post-processing, Clifford fix-ups, and injection. Explicit constant-depth gadgets for each are given in Methods and in Sections~\ref{sec:Extracting CCZs}, \ref{sec:zero_out}, and~\ref{supp:teleporting-cz} of the Supplementary Information~\cite{SI} and are summarized in \cref{fig:gadgets}.

The hypergraph state \detstate is not itself a directly useful primitive for most algorithms. Measuring one logical qubit of each quadcycle block in the $Z$ basis instead yields two disjoint \cczstate states, which is optimal \cite{menon_magic_2025}, potentially dressed by a constant number of $\overline{\mathrm{CZ}}$ gates depending on the logical measurement outcomes. In Section~\ref{sec:zero_out} of the Supplementary Information~\cite{SI}, we show that these selective logical measurements can be performed with an efficient, constant-depth homomorphic measurement \cite{huang_homomorphic_2023} protocol using an auxiliary quadcycle code, exploiting the fact that half the quadcycle code remains in $\ket{\bar 0}$ and can serve as logical ancillae. The output is two disjoint \cczstate states dressed by additional $\overline{\mathrm{CZ}}$ gates, which can be fixed in parallel with feed-forward corrections -- in Section~\ref{supp:teleporting-cz} of the Supplementary Information~\cite{SI} we show that any such $\overline{\mathrm{CZ}}$ correction pattern can be implemented by efficiently preparing and consuming resource states transversally using an ancillary tricycle code. The cost of these logical measurements is not included in the space-time cost accounting in \cref{tab:comparison}, as more optimized schemes using high-rate and single-shot lattice surgery techniques may be possible, which we do not study here. Similarly, we do not include the costs of converting $\ket{T}$ states to $\cczstate$ states for the cultivation protocols in \cref{tab:comparison}.

Injection into another quadcycle block, or into a complementary code designed for Clifford logic \cite{webster2026pinnacle, cain_shors_2026}, proceeds by qLDPC code surgery, and our escape codes are well suited to it. The qubits holding the magic state admit disjoint logical $Z$ operators, which simplifies parallel surgery \cite{cohen2022low, cowtan_parallel_2026, zhang_2025_time, xu2025batched, zheng_high-rate_2025, gu_qgpu_2026}, and these operators can be measured fault-tolerantly---independently or in parallel---by gadgets with minimal qubit overhead that are highly structured and efficient to implement with neutral-atom reconfigurability. The resulting gadgets bridge naturally with those of the target code \cite{swaroop_universal_2025, cross_improved_2025}. Injection requires $d$ rounds of error correction \cite{chamberland2022temporal}, which recent single-shot surgery techniques may reduce further \cite{baspin2025fast, cowtan_fast_2025}.

Finally, consuming the magic state leaves an adaptively determined pattern of logical $\overline{CZ}$ corrections. These Clifford corrections could be delegated to the broader architecture \cite{chamberland2022temporal, cowtan_parallel_2026}. Alternatively, the factory itself can be modified to transversally produce a resource state for whichever correction is required, which is then teleported in the same way (Methods). This ensures that feed-forward corrections do not bottleneck the factory, without requiring additional Clifford-logic capabilities. In practice, however, one might use logical Pauli-product measurement gadgets constructed by generalized lattice surgery techniques to perform this correction efficiently \cite{cohen2022low, zheng_high-rate_2025, baspin2025fast, maskara2026efficient}.

\outline{Efficient implementation on neutral atom arrays}\label{sec:implementation}

We next consider the practical implementation of the proposed magic state factories, focusing on reconfigurable neutral-atom arrays.
The factory compiles into natural atom-rearrangement patterns by making use of three key features of the tricycle--quadcycle construction. Their common product structure gives a regular toric layout in which every monomial of the defining polynomials acts as a uniform cyclic translation of a whole sector, executed as parallel row and column moves (Fig.~\ref{fig:compilation}a,b). Because a matched pair shares the factors $A$, $B$, and $C$, the tricycle cell is identified with the lower row of the quadcycle cell, so the two codes share routing schedules for their common stabilizer terms and the homomorphic CNOT reduces to one layer of site-matched gates between overlaid sectors. The transversal $\overline{\mathrm{CCZ}}$ uses the same primitive: the three blocks are stacked vertically and paired by monomial translations, with the fan-out gadget allowing every gate to be applied in a single Rydberg pulse (Fig.~\ref{fig:compilation}c). Every operation in the factory therefore admits an explicit routing schedule whose duration is a constant number of array traversal times, at optimal circuit depth for both syndrome extraction and the transversal CCZ~\cite{tan26syndrome}. The explicit constructions are provided in the Methods section.

\outline{Discussion and Outlook}\label{sec:outlook}

We have introduced a qLDPC-native architecture in which transversal gates on tricycle codes generate magic in parallel and a depth-1 homomorphic CNOT escapes it into single-shot quadcycle codes, delivering \cczstate{} states at roughly an order of magnitude lower space-time cost than magic state cultivation.

One immediate improvement would be to replace full error detection with sliding-scale
postselection \cite{lee2025efficient, xie_simple_2026, wills_forced_2026}. In the current analysis, every run with a nontrivial detector outcome is rejected, which
suppresses errors strongly but reduces the acceptance fraction and therefore increases the expected space-time cost. Accepting a carefully chosen subset of nontrivial syndromes could
raise acceptance substantially at little cost in postselected fidelity. Quantifying this tradeoff for the tricycle--quadcycle pairs studied here may lead to further reductions in factory overhead. This could be complemented by simulations of our factories that exactly include the effect of the CCZ gates, which can be realized efficiently with recent methods \cite{xiao2026efficient, cercelescu2026polynomial}.

The short depth of the factory is also advantageous for neutral-atom implementations. Because each atom is initialized, used, and measured over a small number of circuit layers, there is limited time for atom loss or leakage to accumulate before the affected block is read out~\cite{baranes_leveraging_2026, bluvstein_architectural_2025}. If the final measurements can distinguish atom loss from ordinary computational-basis outcomes, the locations of lost atoms can be supplied to the decoder as additional information, further improving error-correction performance. Our numerical results use a most-likely error decoder, which provides a strong benchmark for the achievable logical performance but is too computationally expensive for real-time operation at scale. Recent machine learning decoders have achieved near-optimal performance with practical inference times for high-rate qLDPC codes \cite{blue_machine_2026, gu2026scalable}. Hierarchical decoders have likewise demonstrated similar performance with sufficient throughput~\cite{zhao2026towards, bhardwaj2026high, senior2026scalable, gu2026scalable}. These developments suggest that much of the performance reported here may be attainable with practical decoders trained directly on the factory circuits, potentially incorporating loss-resolving readout as an additional decoder input.

While the tricycle codes used in the factories presented in this work all have $k=3$ logical qubits, higher-rate variants of tricycle codes have been introduced in Ref. \cite{menon_magic_2025}, potentially leading to higher $\cczstate$ throughput when paired with appropriately matched quadcycle codes. However, the entangled hypergraph magic states produced have more complex structure than the $\detstate$ state, and extracting as many $\cczstate$ states as possible is a computationally challenging task. Moreover, the transversal $\mathrm{CCZ}$ circuits supported by these higher-rate tricycle codes are deeper. A thorough analysis of the fault-tolerance and resource costs of factories using higher-rate tricycle codes is an interesting direction for future work.

Finally, an important direction is to evaluate the factory within complete fault-tolerant computations rather than only through resource-state-level metrics. End-to-end estimates could use the logical $\ket{\overline{\mathrm{CCZ}}}$ states produced here as resources for large-scale algorithms, including Shor's algorithm \cite{shor_1994_algorithms} and simulations of the Fermi--Hubbard model \cite{hubbard_electron_1963, gutzwiller_1963_effect}. Such studies would include the production, storage, consumption, and scheduling of magic states together with the Clifford portion of the computation, and would determine how the reduced factory overhead translates into the total qubit count and runtime of practically relevant algorithms.

\begin{acknowledgments}
We thank Nazli Ugur Koyluoglu, Masha Elovenkova, Alexandra Geim, Nishad Maskara, Louis Golowich, Christine Li, Zhiyang He, and Kaavya Sahay for helpful discussions. We acknowledge support from Anthropic under the AI for Science program. This work was further supported by the National Science Foundation (CCF 2313084), the CUA Physics Frontiers Center (PHY-2317134) and NVQL (PHY-2410716), IARPA and the Army Research Office, under the Entangled Logical Qubits program (W911NF 23-2-0219), the DARPA MeasQuIT program (HR0011-24-9-0359), and QuEra Computing.

\textit{AI use disclosure---} The authors used models from Anthropic’s Claude family to assist with numerical work, including code search and simulation development. All code was reviewed, tested, and validated by the authors.
\end{acknowledgments}

\bibliography{refs}

\clearpage
\newpage
\section*{Methods}

\noindent\textbf{Code constructions and properties}\\ 

We use notation and terminology as in previous works that study quantum group-algebra codes of various types -- see Refs. \cite{lin2024quantum, bravyi_high-threshold_2024, menon_magic_2025, lin2025abelian, jacob2025singleshot, breuckmann_cups_2025}.

Quadcycle codes are CSS codes constructed from a finite abelian group
$G$ of order $n_G \coloneqq |G|$ and its binary group algebra
$R=\FG$. An element $a=\sum_{g\in G}a_g g\in R$, with $a_g\in\F$,
may be identified with its support, namely the subset of $G$ on which
$a_g=1$. Multiplication in $R$ is given by group convolution. If
$G=\mathbb{Z}_{l_1}\times\cdots\times\mathbb{Z}_{l_r}$, then, after
associating one formal variable with each cyclic factor,
$R\cong\F[x_1,\ldots,x_r]/(x_1^{l_1}-1,\ldots,x_r^{l_r}-1)$.
Elements of $R$ can therefore be represented as multivariate cyclic
polynomials.

For $a\in R$, let $\mathbb{B}(a)$ denote the $n_G\times n_G$ binary
matrix representing left multiplication by $a$ in the regular
representation of $G$. These matrices satisfy
$\mathbb{B}(a)\mathbb{B}(b)=\mathbb{B}(ab)$ and
$\mathbb{B}(a)^T=\mathbb{B}(a^*)$, where
$a^*\coloneqq\sum_{g\in G}a_g g^{-1}$ is the involution induced by
group inversion. Because $G$ is abelian, $R$ is commutative, and all
matrices of the form $\mathbb{B}(a)$ commute pairwise.

A quadcycle code is specified by four elements $a,b,c,d\in R$. It has $n = 6|G|$ physical qubits, grouped into six blocks of size $|G|$, together with $4|G|$ $X$-checks and $4|G|$ $Z$-checks. Their parity-check matrices are the block matrices
\begin{align}\label{eq: quadcycle checks}
    \mathcal{H}_X = \mathbb{B}\!\begin{pmatrix}
        c^* & 0   & b^* & d^* & 0   & 0   \\
        0   & c^* & a^* & 0   & d^* & 0   \\
        a^* & b^* & 0   & 0   & 0   & d^* \\
        0   & 0   & 0   & a^* & b^* & c^*
    \end{pmatrix}
    \\
        \mathcal{H}_Z = \mathbb{B}\!\begin{pmatrix}
        b & a & c & 0 & 0 & 0 \\
        d & 0 & 0 & c & 0 & a \\
        0 & d & 0 & 0 & c & b \\
        0 & 0 & d & b & a & 0
    \end{pmatrix}
\end{align}
where $\mathbb{B}$ is applied entrywise, each entry becoming its $|G| \times |G|$ circulant block and each $0$ a zero block. Each row is one stabilizer generator. Since $R$ is commutative, $\mathcal{H}_X \mathcal{H}_Z^T = 0$ over $\F$: the contributions to every entry of the product cancel in pairs, so the $X$- and $Z$-checks commute and define a valid CSS code. Every generator has weight $|a| + |b| + |c| + |d|$ minus the weight of one element, i.e. the sum of three of $|a|, |b|, |c|, |d|$; for the weight-$(2,2,2)$ choices of $(a, b, c)$ and weight-$4$ choice of $d$ used here, the checks have weights $6$ and $8$. Each of the six qubit blocks --- the six columns of \eqref{eq: quadcycle checks} --- couples through two of the four factors in $\mathcal{H}_X$ and through the complementary two in $\mathcal{H}_Z$, which indexes the blocks by the two-element subsets of $\{a,b,c,d\}$. Labeling each by its $\mathcal{H}_Z$ pair, in column order they are $\{b,d\}$, $\{a,d\}$, $\{c,d\}$, $\{b,c\}$, $\{a,c\}$, $\{a,b\}$.

We
write
$A\coloneqq\mathbb{B}(a)$,
$B\coloneqq\mathbb{B}(b)$,
$C\coloneqq\mathbb{B}(c)$, and
$D\coloneqq\mathbb{B}(d)$. Both matrices have dimensions $4n_G\times6n_G$. After the blocks are
expanded, each binary row specifies one stabilizer generator.

\paragraph{Metachecks and the four-dimensional chain complex.}
The listed stabilizer generators are redundant in both Pauli bases.
The corresponding metacheck matrices, in circulant-expanded form, are
\begin{equation}\label{eq:quadcycle-metachecks}
    \mathcal{M}_X
    =
    \begin{bmatrix}
        A^T & B^T & C^T & D^T
    \end{bmatrix},
    \qquad
    \mathcal{M}_Z
    =
    \begin{bmatrix}
        D & B & A & C
    \end{bmatrix}.
\end{equation}
Each metacheck matrix has dimensions $n_G\times4n_G$ and supplies
$n_G$ explicit, though not necessarily independent, relations among
the listed stabilizers. Direct block multiplication gives
$\mathcal{M}_X\mathcal{H}_X=0$ and
$\mathcal{M}_Z\mathcal{H}_Z=0$. Together with the CSS orthogonality
condition, these maps form the five-term chain complex
\begin{equation}\label{eq:quadcycle-complex}
    \F^{n_G}
    \xrightarrow{\ \mathcal{M}_Z^T\ }
    \F^{4n_G}
    \xrightarrow{\ \mathcal{H}_Z^T\ }
    \F^{6n_G}
    \xrightarrow{\ \mathcal{H}_X\ }
    \F^{4n_G}
    \xrightarrow{\ \mathcal{M}_X\ }
    \F^{n_G}.
\end{equation}
Consecutive maps compose to zero:
$\mathcal{H}_X\mathcal{H}_Z^T=0$,
$\mathcal{M}_X\mathcal{H}_X=0$, and
$\mathcal{H}_Z^T\mathcal{M}_Z^T
=(\mathcal{M}_Z\mathcal{H}_Z)^T=0$.
Reading outward from the center of
\cref{eq:quadcycle-complex}, the physical qubits occupy
$\F^{6n_G}$; the adjacent copies of $\F^{4n_G}$ index the
$Z$- and $X$-checks; and the outer copies of $\F^{n_G}$ index the
$Z$- and $X$-metachecks.

For an ideal measured $X$-syndrome $s_X\in\F^{4n_G}$, physical
consistency requires $\mathcal{M}_X s_X=0$. Similarly, an ideal
$Z$-syndrome $s_Z\in\F^{4n_G}$ satisfies
$\mathcal{M}_Z s_Z=0$. If the measured $X$-syndrome is
$\widetilde{s}_X=s_X+\eta_X$, where $\eta_X$ is a
syndrome-measurement error, its metasyndrome is
$\mathcal{M}_X\widetilde{s}_X=\mathcal{M}_X\eta_X$. An analogous
relation holds in the $Z$ basis. The metachecks therefore provide
spatial consistency conditions that can reveal and help decode
measurement faults using a single syndrome-extraction round. Metachecks, when combined with soundness \cite{campbell2019theory}, suffice for single-shot state preparation of a code block: the ability to fault-tolerantly learn the eigenvalues of the stabilizers with only a single round of syndrome measurement. While the metachecks provide some protection against measurement errors, some low weight measurement errors will still remain undetected by the metachecks. Soundness ensures that these can effectively be treated as low weight errors acting on the data qubits. Ref. \cite{menon_magic_2025} details a proof that several of the quadcycle codes we consider here possess soundness in both bases.

This two-sided redundancy is the algebraic structure exploited by our
single-shot state-preparation and error-correction procedures in both
Pauli bases~\cite{bombin2015single, brown2016fault}. The existence of
metachecks alone, however, does not establish useful single-shot fault
tolerance: the resulting suppression also depends on the metacheck
distances, syndrome-extraction circuits, and decoder. We assess these
properties for the finite-size codes used here through the
fault-distance analysis and circuit-level simulations described
below.

\paragraph{Relation to tricycle codes.}

The elements $a,b,c$ are precisely the three factors defining the
paired tricycle code, while $d$ introduces the fourth cyclic
direction. At the level of chain complexes, the quadcycle code can be
constructed as the \textit{mapping cone} of the chain map induced by
multiplication by $d$ on the tricycle complex -- see Section~\ref{app:tcpairs} of the Supplementary Information~\cite{SI}. Equivalently, it is the
asymmetric balanced product of the tricycle complex with the
two-term chain complex
$R\xrightarrow{\cdot d}R$.

As detailed in Section~\ref{app:tcpairs} of the Supplementary Information~\cite{SI}, there is a canonical embedding of the
tricycle code complex into the quadcycle code complex. At the physical level,
the corresponding transversal homomorphic $\mathrm{CNOT}$ used for the dimension
jump acts on three of the six quadcycle qubit sectors. It can be shown that
$k_{\textnormal{4D}}\leq2k_{\textnormal{3D}}$. All tricycle--quadcycle pairs
used in this work saturate this upper bound,
$k_{\textnormal{4D}}=2k_{\textnormal{3D}}$. The mapping cone diagram, the
induced action on homology, and the logical map implemented by the
dimension jump are detailed in Section~\ref{app:tcpairs} of the Supplementary Information \cite{SI}.

\begin{table*}[t]
    \centering
    \renewcommand{\arraystretch}{1.2}
    \setlength{\tabcolsep}{5pt}

    \resizebox{\textwidth}{!}{%
    \begin{tabular}{@{}c l l l c l c@{}}
        \toprule
        \thead{$(l,m,n)$}
        & $a$ & $b$ & $c$
        & \thead{Tricycle code\\$[[n,k,(d_X,d_Z)]]$}
        & $d$
        & \thead{Quadcycle code\\$[[n,k,d]]$} \\
        \midrule
        $3,3$
        & $1+x$ & $1+y$ & $x+y$
        & $[[27,3,(6,3)]]$
        & $1+xy$
        & $[[54,6,6]]$ \\

        $3,4$
        & $1+xy$ & $1+y$ & $x+y$
        & $[[36,3,(8,4)]]$
        & $1+x+x^2+y^2$
        & $[[72,6,8]]$ \\

        $3,9$
        & $1+y$ & $x+y^2$ & $x^2+y^2$
        & $[[81,3,(15,5)]]$
        & $1+x+y^3+y^6$
        & $[[162,6,14]]$ \\

        $2,19$
        & $x+y^2$ & $x+y^3$ & $x+y^5$
        & $[[114,3,(19,6)]]$
        & $1+xy+y^2+y^9$
        & $[[228,6,17]]$ \\

        $2,5,7$
        & $1+yz$ & $x+z^2$ & $xy+z^3$
        & $[[210,3,(28,7)]]$
        & $1+xy+xy^2+y^3z$
        & $[[420,6,25]]$ \\
        \bottomrule
    \end{tabular}%
    }

    \caption{Polynomials defining the tricycle--quadcycle code pairs and
    their parameters. Codes are constructed over the group algebra
    $\mathbb{F}_2[\mathbb{Z}_l\times\mathbb{Z}_m\times\mathbb{Z}_n]$,
    with the cyclic group orders given in the first column. The three
    polynomials $a,b,c$ define the 3D tricycle code; the fourth
    polynomial $d$ specifies the paired 4D quadcycle code, for which
    $d_X=d_Z$. Quadcycle distances are determined exactly for
    $n\leq72$ using MaxSAT and estimated for larger blocklengths using
    Stern's algorithm, cross-checked with Monte Carlo sampling.}
    \label{tab:codes}
\end{table*}

\paragraph{Automorphisms and transversal Clifford operations.}

For every $g\in G$, simultaneous translation by $g$ permutes the
qubits within each block through the regular representation
$\mathbb{B}(g)$. Because these translations commute with $A$, $B$,
$C$, and $D$, they preserve the stabilizer group and define
depth-1 code automorphisms. Their induced logical actions are used
in the post-processing of the escaped hypergraph state.

The quadcycle complex also possesses an $X$--$Z$ duality. A transversal Hadamard on all $n$ physical qubits followed by an order-two permutation gives a Hadamard-SWAP, where the permutation swaps the column blocks in the first half of $\mathcal{H}_X/\mathcal{H}_Z$ with those in the second. Writing a logical $Z$ operator by its support on each half as a nontrivial cycle $\mathbb{B}(z_1, z_2) \in \ker \mathcal{H}_X$,
\begin{equation}
\mathbb{B}\begin{pmatrix}
        z_1 \\z_2
    \end{pmatrix} \xmapsto{H\text{-}\mathrm{SWAP}} \mathbb{B}\begin{pmatrix}
        (z_2^*)^T \\
        (z_1^*)^T
    \end{pmatrix} \in \ker \mathcal{H}_Z \setminus \im \mathcal{H}_X^T,
\end{equation}
the map carries low-weight $Z$ operators to low-weight $X$ operators, and therefore $d_X = d_Z$ \cite{bravyi_high-threshold_2024}. The same structure produces a transversal $\overline{CZ}$ gate between two distinct blocks. The duality also yields a fold-transversal $\overline{CZ}$--$\bar S$ gate \cite{breuckmann_fold-transversal_2024}: the $\overline{CZ}$ acts between interchanged pairs of logical qubits, and $\bar S$ on the fixed points of the permutation. Quadcycle codes, like bicycle codes, have no such fixed points, so the $\bar S$ component is absent. Recovering a $\bar S$ gate on bicycle codes requires imposing additional self-dualities~\cite{liang_self-dual_2026}, and we expect similar conditions to result in self-dual quadcycle codes.

\noindent\textbf{Circuit-level noise model}\\ We simulate every tricycle--quadcycle pair in Stim \cite{gidney_stim_2021} under the standard circuit-level depolarizing model (SD6): idling and single-qubit gates are followed by a single-qubit depolarizing channel, two-qubit gates by a two-qubit depolarizing channel, and resets and measurements by a single-qubit flip channel, all at strength $p=10^{-3}$. This is the model used in the cultivation studies we compare against \cite{gidney_magic_2024, sahay2026fold}, so the comparison in \cref{tab:comparison} is made on equal footing. Each physical CCZ triple carries a three-qubit depolarizing channel, in the uniform non-identity-Pauli convention, of strength $p_{\mathrm{CCZ}}=2p$. Because we simulate only a single tricycle--quadcycle pair, the other two legs of every triple belong to the remaining factory copies and are traced out. Since the transversal $\overline{\mathrm{CCZ}}$ is depth one, each simulated qubit takes part in exactly one triple, so the reduced channel is exactly the single-qubit marginal of the three-qubit channel, namely depolarizing noise of strength $\tfrac{16}{21}\,p_{\mathrm{CCZ}} = \tfrac{32}{21}\,p$. Tracing out the other legs erases cross-block correlations in the CCZ channel, which we expect to make both the logical error rate and the acceptance fraction conservative. Acceptance fractions are reported per block. The full factory requires all three pairs to be accepted, so its yield is the cube of the quoted value, which is what enters the space-time volumes of \cref{tab:comparison}.

\begin{figure}[ht!]
    \centering
    \includegraphics[width=\linewidth]{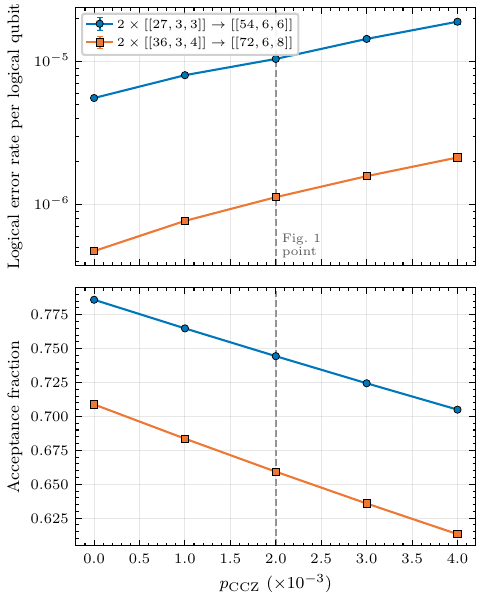}
    \caption{\textbf{Dependence of the factory on the CCZ error rate.} Post-selected logical error rate per output logical qubit (top) and single-block acceptance fraction (bottom) for the $2\times[[27,3,3]]\rightarrow[[54,6,6]]$ and $2\times[[36,3,4]]\rightarrow[[72,6,8]]$ factories, as a function of the CCZ error rate $p_{\mathrm{CCZ}}$, with all other operations held at depolarizing noise $p=10^{-3}$. The simulated CCZ noise is the single-qubit marginal of a three-qubit depolarizing channel; this marginal is exact because each qubit participates in only one gate in the depth-1 transversal $\overline{\mathrm{CCZ}}$ circuit. Error bars are binomial.}
\label{fig:ccz_sweep}
\end{figure}

\begin{figure*}[t!]
    \centering
    \includegraphics[width=\linewidth]{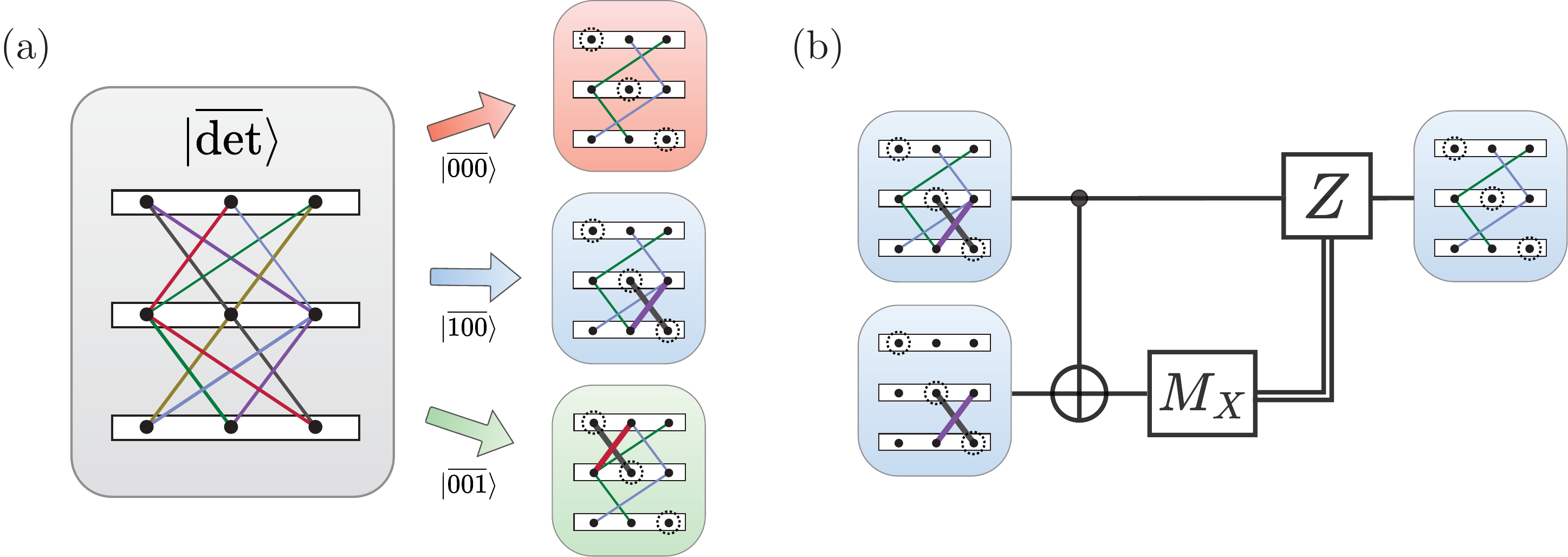}
    \caption{\textbf{Heralding disjoint \cczstate states from \detstate}. Each illustrated block is a hypergraph magic state, which is described by a diagonal circuit acting on logical qubits initialized in $\ket{\overline +}$, except for the circled logical qubits, which are in a computational basis state. Edges in the graph denote a CZ gate, whereas hyperedges denote CCZ gates.
    (a) Measuring $\bar q_{11}, \ \bar q_{22}$, and $\bar{q}_{33}$ eliminates four of the six CCZ gates from the \detstate,  and depending on the measurement outcome of $\ket{\bar{q}_{11} \bar{q}_{22} \bar{q}_{33}}$, annotated by the arrows, the resulting state is equivalent to disjoint \cczstate states up to some logical CZ gates. Each logical qubit measured in $\ket{\bar 1}$ contributes two CZ gates. (b) Any CZ-dressed resulting state can be heralded into disjoint $\cczstate$ states by teleporting in a correction.
    }
    \label{fig:gauging_and_fixups}
\end{figure*}

\noindent\textbf{A four-dimensional receiving code}\\ Three complementary features make the factory efficient. First, both code families are single-shot in the relevant bases, which keeps the protocol shallow. This generates the magic states rapidly and, just as importantly, confines postselection to a small space-time volume, making full error detection practical. Second, the dimension jump ends with a transversal measurement of the tricycle block, which exposes the full redundancy of the code and is effectively immune to measurement errors. In a generic mid-circuit postselection step, transversally measuring the data would collapse the computational state. Here the measurement is already the teleportation primitive that transfers the logical information into a separate quadcycle block, so its outcomes can be used for strong postselection without disturbing the escaped magic state. Third, the choice of a four-dimensional receiving code is essential. Two-dimensional codes can also be homomorphically related to the tricycle code, but they are not single-shot, so their preparation would require additional syndrome-extraction rounds and deepen the factory. Their smaller distances would further require subsequent code growth, likely with additional postselection, to preserve the fidelity obtained during the escape, and their smaller logical dimensions would prevent a single block from receiving all the magic states a tricycle code generates. A quadcycle block instead combines single-shot preparation with sufficient distance to preserve the postselected fidelity without further growth, and sufficient logical dimension to accept every generated magic state in one transfer.

\noindent\textbf{Space-time cost accounting}\\ 

The space-time cost figures of merit in \cref{tab:comparison} are computed using the formula

\begin{equation}
    \text{cost} = \frac{ \text{raw space-time volume}}{\text{acceptance fraction} \times \text{\# magic states}}
\end{equation}

For the tricycle-quadcycle factory, the circuit depth is the quadcycle syndrome extraction cycle depth plus four layers: one for the $\mathrm{CCZ}$ layer, one for the dimension-jump $\mathrm{CNOT}$, and two for the transversal $X$ measurement of the tricycle code, including the Hadamard layer (see Section~\ref{supp:cost-breakdown} of the Supplementary Information~\cite{SI}). The tricycle code and quadcycle code state-preparation can be parallelized. The number of qubits includes ancilla qubits to measure checks as well as the fan-out gadget ancillae to reduce the depth of the physical $\mathrm{CCZ}$ circuit. The acceptance fractions are computed as the cube of the single-block acceptance fractions reported in \cref{fig:dimjump}(c) -- this is conservative, as correlated detector events across all three code blocks could lead to slightly higher acceptance fractions at fixed logical error rates.

For the cultivation protocols, we use the results of Ref. \cite{mehta_practical_2026} which studies cultivation with mid-circuit postselection at the escape stage. This treats cultivation on a realistic, equal footing as only information available from the middle of the cultivation circuit is used without relying on information from an artificial transversal measurement in the escape stage of cultivation. The numbers in \cref{tab:comparison} result from optimizing the number of escape rounds and escape code distance for cultivation protocols to achieve a fixed target logical error rate set by our factories for an equal comparison. For cultivation, the circuit depths are then computed as the sum of $\text{\# QEC rounds} \times \text{cycle depth}$ for each of the different stages and codes for both the color code cultivation protocol of Ref. \cite{gidney_magic_2024} and the fold-transversal surface code protocol of Ref. \cite{sahay2026fold}. We refer the reader to Ref. \cite{mehta_practical_2026} for more details on the basis of this comparison to cultivation and the impact of practical postselection on its resource requirements. One detail to note, however, is that Ref. \cite{mehta_practical_2026} simulated color code cultivation with the escape stage proposed in \cite{gidney_magic_2024}, whereas Ref. \cite{sahay2026fold}'s estimates of the protocol use a more compact escape stage. Nonetheless, upon our simulation of this more compact escape stage following \cite{mehta_practical_2026}, we found the two to be comparable.

\noindent\textbf{Post-processing, injection and feed-forward gadgets}\\ 

\emph{Extracting disjoint \cczstate states.} The factory delivers the nine-qubit hypergraph state \detstate, which is not a primitive for most algorithms. Measuring one logical qubit of each quadcycle block in the $Z$ basis yields two disjoint \cczstate states, the optimal yield for this state \cite{menon_magic_2025}. We realize these selective logical measurements with an auxiliary quadcycle code, combining the transversal Hadamard--SWAP of the quadcycle code with different choices of constant-depth homomorphic CNOT circuits, extending beyond the dimension-jump CNOT used in the factory itself, so that different patterns of logical operators can be addressed. The addressability achievable this way is restricted, but because half of the quadcycle logical qubits remain in $\ket{\bar 0}$ they can be treated as logical ancillae, and the required measurements are then implementable in constant depth. The fault distance depends on the chosen homomorphic measurement gadget and can be smaller than the data-code distance (Section~\ref{sec:zero_out} of the Supplementary Information~\cite{SI}). In that case, repeated measurements may be needed to reach the full circuit distance, although error detection can reduce the average number required \cite{akahoshi2025runtime}. Depending on the outcomes, the resulting pair of \cczstate states is dressed by additional CZ gates, which are corrected in parallel with the feed-forward corrections applied after consumption. The same techniques allow denser storage: half of the logical qubits in \cref{fig:dimjump} are unused, and two quadcycle magic states can be consolidated onto a single block in a single-shot fashion. Full constructions are given in Section~\ref{sec:Extracting CCZs} of the Supplementary Information~\cite{SI}.

\emph{Injection by code surgery.} Injecting the magic state into another quadcycle block, or into a complementary code designed for general Clifford logic \cite{webster2026pinnacle, cain_shors_2026}, is typically accomplished by qLDPC code surgery. The escape codes are favorable in this setting: the logical $Z$ operators of the qubits carrying the magic state are disjoint, which simplifies parallel surgery \cite{cohen2022low, xu2025batched, zheng_high-rate_2025}. These operators can be measured fault-tolerantly, independently or in parallel, by gadgets with minimal qubit overhead whose regular structure is well matched to neutral-atom reconfigurability. The gadgets can then be bridged with surgery gadgets for the target code to give low-overhead teleportation \cite{swaroop_universal_2025, cross_improved_2025}. This step involves $d$ rounds of repeated error correction \cite{chamberland2022temporal}, though recent single-shot code surgery may reduce that overhead \cite{baspin2025fast, cowtan_fast_2025}.

\emph{Feed-forward Clifford corrections.} Consuming a magic state leaves an adaptively determined pattern of logical CZ gates. These are Clifford, so their implementation could be reserved for the broader architecture, for instance via code surgery \cite{chamberland2022temporal, cowtan_parallel_2026}. Alternatively, the factory itself can be modified to produce a resource state for whichever correction is required, which is then teleported in the same manner. Because no transversal non-Clifford gate is involved, this variant is simpler and decodes more effectively \cite{cain_correlated_2024}. Teleporting the correction introduces a further adaptive fix-up, but it is a Pauli operator and can be tracked classically. The cost of the feed-forward correction is therefore at most that of the factory itself.

\noindent\textbf{Neutral-atom compilation}\\
We assume static traps generated by a spatial light modulator together with movable two-dimensional grids generated by crossed acousto-optic deflectors (AODs). During an AOD move, selected atoms are transported by shifting rows and columns while preserving their relative ordering. Detailed hardware assumptions are given in Section~\ref{sec:appendix_compilation_setup} of the Supplementary Information~\cite{SI}.

\emph{Layout and routing primitive.} The product structure of the codes gives a regular toric layout for both data and check qubits. For $G=\mathbb{Z}_l\times\mathbb{Z}_m$ the array is an $l\times m$ grid of unit cells, each holding one qubit from every data and check sector, so a quadcycle cell contains six data qubits, four $X$-check ancillae, and four $Z$-check ancillae (Fig.~\ref{fig:compilation}a). Additional cyclic factors are flattened into the plane as described in Section~\ref{sec:appendix_compilation_misc} of the Supplementary Information~\cite{SI}. Each monomial in the defining polynomials acts as a uniform cyclic translation of an entire sector, implemented by parallel row and column movements. Where periodic wraparound makes trajectories cross, the move is partitioned among at most four AODs (Fig.~\ref{fig:compilation}b). The quasi-cyclic structure of the codes thus translates directly into regular, highly parallel atom motion.

\emph{Shared layout for a matched pair.} Because a tricycle code and its quadcycle partner share the factors $A$, $B$, and $C$, the tricycle unit cell is identified with the lower row of the quadcycle unit cell. Their common stabilizer terms are then measured with the same AOD routing patterns, so syndrome extraction on the two blocks can be coordinated and run in parallel. The same alignment makes the escape operation simple: the matched data sectors are overlaid and the homomorphic CNOT is one layer of site-matched physical CNOTs, with the remaining quadcycle sectors idle. The algebraic embedding underlying the dimension jump is therefore preserved directly in the array geometry.

\emph{Transversal $\mathrm{CCZ}$.} The three tricycle blocks are stacked vertically, and atoms from the upper and lower blocks are moved into interaction range of the middle block, which acts as a fixed anchor and never moves (Fig.~\ref{fig:compilation}c). Six monomial translations pair the sectors of the upper and lower blocks with those of the middle block. Without the fan-out gadget these six orderings split into two layers that reuse the same atoms and therefore require two Rydberg pulses. With the gadget's ancilla patch in place the two layers are realized simultaneously, so a single pulse applies every logical $\overline{\mathrm{CCZ}}$. Full routing schedules are given in Section~\ref{sec:appendix_compilation_ccz} of the Supplementary Information~\cite{SI}.

\emph{Duration.} Let $\tau$ be the time to traverse the code region diagonally and $w$ the total weight of the defining polynomials. Syndrome extraction costs at most $2w\tau$, since each monomial term is one movement of duration at most $\tau$ followed by a return to the standard layout. The six monomial translations of the CCZ cost at most $12\tau$ including returns, and the homomorphic CNOT at most one further $\tau$ to overlay the matched sectors. For the codes considered here, with $w=10$, the complete factory is bounded by $2w\tau+12\tau+\tau=33\tau$. Space--time tradeoffs in the number of available AODs are discussed in Sections~\ref{sec:appendix_compilation_syndrome} and~\ref{sec:appendix_compilation_ccz} of the Supplementary Information~\cite{SI}.

\clearpage
\onecolumngrid

\setcounter{section}{0}
\setcounter{subsection}{0}
\setcounter{equation}{0}
\setcounter{figure}{0}
\setcounter{table}{0}
\setcounter{theorem}{0}
\setcounter{definition}{0}

\renewcommand{\thefigure}{S\arabic{figure}}
\renewcommand{\thedefinition}{S\arabic{definition}}

\renewcommand{\theHsection}{supp.\arabic{section}}
\renewcommand{\theHsubsection}{supp.\arabic{section}.\arabic{subsection}}
\renewcommand{\theHequation}{supp.\arabic{equation}}
\renewcommand{\theHfigure}{supp.\arabic{figure}}
\renewcommand{\theHtable}{supp.\arabic{table}}

\begin{center}
    {\large\bfseries
    Supplementary Information for ``Single-shot magic state factories for quantum LDPC
    codes''\par}
\end{center}

\vspace{1.5em}

\begin{center}
    \textbf{\MakeUppercase{Contents}}
\end{center}

\vspace{0.5em}

\begingroup
\setlength{\parindent}{0pt}
\newcommand{\sitocline}[2]{%
    \makebox[3em][r]{\hyperref[#1]{\ref*{#1}.}}\hspace{0.6em}%
    \hyperref[#1]{#2}\nobreak\dotfill\nobreak\pageref{#1}\par
    \vspace{0.35em}%
}
\sitocline{app:tcpairs}{Tricycle--quadcycle pairs}
\sitocline{sec: Transversal gates}{Clustered logical bases and transversal gates}
\sitocline{sec:search}{Code search}
\sitocline{sec:DJ CNOT}{Dimension jumping between 3D and 4D}
\sitocline{app:concat}{Restoring circuit distance with the fan-out gadget}
\sitocline{sec:numerics-methods}{Further details of numerical simulations}

\sitocline{sec: Tricycle magic}{Review of tricycle magic}
\sitocline{sec:Extracting CCZs}{Extracting disjoint \cczstate{} states}
\sitocline{sec:zero_out}{Homomorphic measurement gadgets for \cczstate{} extraction}
\sitocline{supp:teleporting-cz}{Teleporting $\overline{CZ}$ corrections}
\sitocline{supp:cost-breakdown}{Cost breakdown of magic factories}
\sitocline{sec: 2D-3D comparison}{Comparison to 2D--3D code-switching protocols}
\sitocline{sec:code_growth}{Direct code-growth from 3D to 4D}
\sitocline{sec:appendix_compilation_setup}{Neutral atom array architecture overview}
\sitocline{sec:appendix_compilation_syndrome}{Implementation of syndrome extraction}
\sitocline{sec:appendix_compilation_ccz}{Implementation of transversal CCZ}
\sitocline{sec:appendix_compilation_misc}{Trade-off and generalization of implementation}
\endgroup

\clearpage

\section{Tricycle--quadcycle pairs}
\label{app:tcpairs}
Tricycle codes are the three-fold balanced product equivalent of quadcycle codes \cite{jacob2025singleshot, menon_magic_2025}. For a tricycle code $T_\bullet$, an algebraically related quadcycle code $Q_\bullet$ can be constructed as the asymmetric balanced product of $T_\bullet$ with the two-term complex $R \xrightarrow{d} R$. Equivalently, $Q_\bullet$ is the mapping cone of the chain map $\tilde d$ shown in the following diagram.
\begin{equation}\label{eq: mapping cone}
\begin{tikzcd}
    \underbrace{\F^{|G|}}_{\text{$Z$ metachecks}}
        \arrow[r, "M_Z^T"] \arrow[dr, "\mathbb{B}(d)"]
    & \F^{3|G|} \arrow[r, "H_Z^T"] \arrow[dr, "\mathbb{B}(d) \otimes \text{id}_3"]
    & \F^{3|G|} \arrow[r, "H_X"] \arrow[dr, "\mathbb{B}(d) \otimes \text{id}_3"]
    & \F^{|G|} \arrow[dr, "\mathbb{B}(d)"] \\
    & \underbrace{\F^{|G|}}_{\text{$Z$ checks}} \arrow[r, "M_Z^T"]
    & \underbrace{\F^{3|G|}}_{\text{qubits}} \arrow[r, "H_Z^T"]
    & \underbrace{\F^{3|G|}}_{\text{$X$ checks}} \arrow[r, "H_X"]
    & \underbrace{\F^{|G|}}_{\text{$X$ metachecks}}
\end{tikzcd}
\end{equation}
The quadcycle code can also be expressed as a mapping cone over any of the other three factors, with the remaining factors defining the corresponding tricycle code. In this way, there are really four tricycle codes that relate naturally to a given quadcycle code, though we focus here on a single pair. Note that four-dimensional hypergraph product and lifted product codes can also be related to three-dimensional subsystems in a similar manner. A similar correspondence was independently and concurrently proposed in Ref. \cite{zhang2026coupled}.

This presentation immediately upper bounds the number of logical qubits that the quadcycle code can support. From the snake lemma \cite{weibel1994introduction},
\begin{equation}\label{eq: LES}
    \begin{tikzcd}[column sep = 2 em]
        \cdots \arrow[r] & H_n(T_\bullet) \arrow[r, "d_*"] & H_n(T_\bullet) \arrow[r, "p"] & H_n(Q_\bullet) \arrow[r, "q"] & H_{n-1}(T_\bullet) \arrow[r, "d_*"] & H_{n-1}(T_\bullet) \arrow[r] & \cdots
    \end{tikzcd}
\end{equation}
From rank-nullity,
\begin{equation}
    \dim H_{n}(Q_\bullet) = \dim \text{coker}(d_*|_{H_n}) + \dim \ker (d_*|_{H_{n-1}}) \leq 2k
\end{equation}
where the inequality follows from the self-duality of $T_\bullet$ as a chain complex (not as a code). The quadcycle codes we derive from tricycle codes saturate this upper bound on logical qubits.

\section{Clustered logical bases and transversal gates}\label{sec: Transversal gates}
Quadcycle codes admit several code automorphisms and transversal Clifford gates. The CSS structure alone gives a transversal CNOT between two copies of the code, acting on all $k$ logical qubits at once. The circulant construction adds the group translations: multiplication by any $g\in G$ permutes the qubits of each block by $\mathbb{B}(g)$, which commutes with every check, so $G$ acts by depth-1 automorphisms that preserve $\mathcal{H}_X$ and $\mathcal{H}_Z$. The logical action of these gates depends on the logical basis. Many codes from our search admit a \emph{clustered} basis \cite{gu_qgpu_2026} on which the action is transparent: for example, a $[[456,24,16]]$ code over $G=\mathbb{Z}_4\times\mathbb{Z}_{19}$, a $[[300,12,18]]$ code over $G=\mathbb{Z}_2\times\mathbb{Z}_{25}$, as well as the $[[162, 6, 14]]$ quadcycle escape code in \cref{tab:codes}. A clustered basis partitions the physical qubits into equal clusters, supports every logical operator on whole clusters, and pairs each $\bar X_i$ with one anticommuting $\bar Z_i$ on the same cluster, disjoint from the rest. For the first two examples, the clusters are rows: arrange each block's $|G|$ qubits as a $\mathbb{Z}_a\times\mathbb{Z}_b$ grid, and every logical operator is a full row --- $\bar X_i$ the $X$-string and $\bar Z_i$ the $Z$-string along it, of odd length $b=|G|/a$ ($19$ and $25$). The $k=6a$ rows tile all $N$ qubits, $a$ to a block, so each logical qubit is labeled by a pair $(\beta,j)$: its block $\beta\in\{\{a,b\},\{a,c\},\dots\}$ and its row $j\in\mathbb{Z}_a$. In this basis, translation by $x \in \mathbb{Z}_a$ permutes logical qubits inside each block, $(\beta,j)\mapsto(\beta,\,j+1\bmod a)$ with $\beta$ the block and $j$ the operator within it. On the other hand, translation by $y \in \mathbb{Z}_b$ has a logically trivial action.

A transversal Hadamard on all $n$ physical qubits followed by an order-two permutation gives a Hadamard-SWAP, where the permutation swaps the column blocks in the first half of $\mathcal{H}_X/\mathcal{H}_Z$ with those in the second. Writing a logical $Z$ operator by its support on each half as a nontrivial cycle $\mathbb{B}(z_1, z_2) \in \ker \mathcal{H}_X$,
\begin{equation}
    \ker \mathcal{H}_X \setminus \im \mathcal{H}_Z^T \ni \mathbb{B}\begin{pmatrix}
        z_1 \\z_2
    \end{pmatrix} \xmapsto{H\text{-}\mathrm{SWAP}} \mathbb{B}\begin{pmatrix}
        (z_2^*)^T \\
        (z_1^*)^T
    \end{pmatrix} \in \ker \mathcal{H}_Z \setminus \im \mathcal{H}_X^T,
\end{equation}
the map carries low-weight $Z$ operators to low-weight $X$ operators, and therefore $d_X = d_Z$ \cite{bravyi_high-threshold_2024}. The same structure produces a transversal $\overline{CZ}$ gate between two distinct blocks. The duality also yields a fold-transversal $\overline{CZ}$--$\bar S$ gate \cite{breuckmann_fold-transversal_2024}: the $\overline{CZ}$ acts between interchanged pairs of logical qubits, and $\bar S$ on the fixed points of the permutation. Quadcycle codes, like bicycle codes, have no such fixed points, so the $\bar S$ component is absent. Recovering a $\bar S$ gate on bicycle codes requires imposing additional self-dualities~\cite{liang_self-dual_2026}, and we expect similar conditions to result in self-dual quadcycle codes.

\section{Code search}\label{sec:search}

A quadcycle code is specified by a finite abelian
group $G$ and four elements $a,b,c,d \in R=\FG$. It has blocklength
$N=6|G|$. For each choice of group and factor weights, we sample the four
factors at random and retain only codes satisfying prescribed target
parameters: a logical dimension threshold $K \ge K_{\min}$ and a distance
threshold $d \ge d_{\min}$. In the searches reported here,
\begin{equation}
    K_{\min} = 6\lfloor \lceil N/40\rceil/6\rfloor,
\end{equation}
which imposes a rate floor near $2.5\%$. The distance threshold
$d_{\min}$ is chosen separately for each search window.

The aim of the search procedure is not to exactly compute the full parameters
of every candidate code, which would be prohibitively expensive. Instead, the
procedure is organized as a sequence of increasingly costly rejection tests.
This ordering is chosen so that inexpensive algebraic checks eliminate as many
candidates as possible before any heuristic distance computation is attempted.
Early tests certify that a candidate fails one of the target conditions, while
distance estimation is reserved for the small fraction of candidates that
survive these algebraic filters. In this way, computational effort is
concentrated on candidates that have not already been ruled out by exact
invariants. This tiered structure allows us to sample a significantly larger
portion of the code ensemble for a fixed computational budget.

\begin{figure}[h]
\centering
\includegraphics[width=0.92\textwidth]{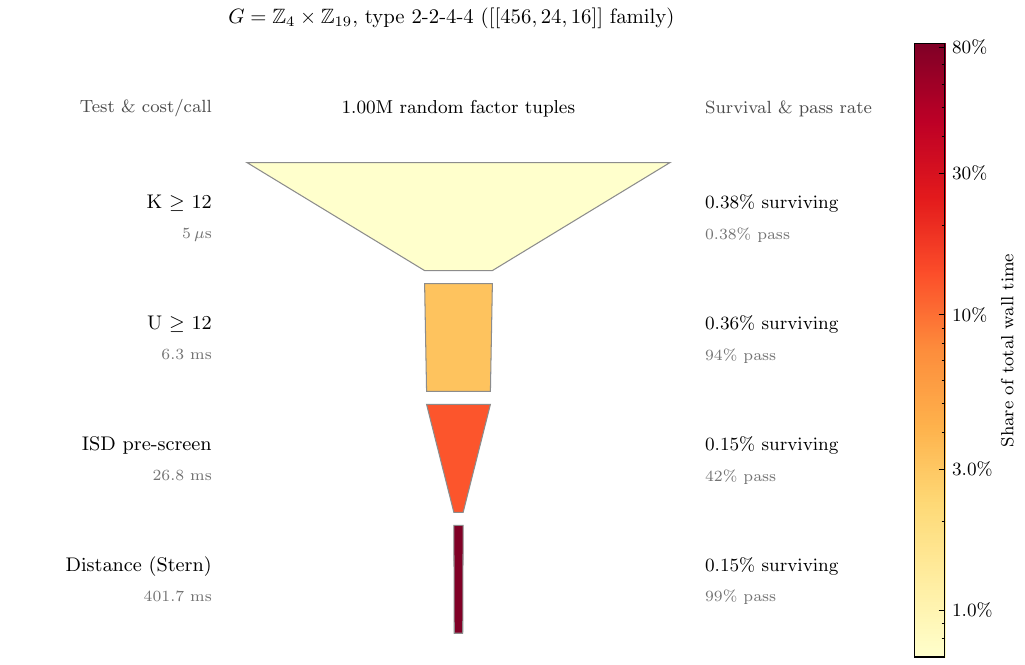}
\caption{Search funnel for the $[[456,24,16]]$ family
($G=\mathbb{Z}_4\times\mathbb{Z}_{19}$, type 2-2-4-4).
Band width is proportional to the number of surviving candidates on a
logarithmic scale, and shading shows the share of total wall time.
Per-candidate cost is shown at left and survivors at right. The
logical-count test is the dominant filter, while distance estimation is
reached by few candidates but dominates the runtime.}
\label{fig:search-funnel}
\end{figure}

The first test checks the logical dimension. We do not compute $K$ by
row-reducing the $N$-column parity-check matrices. Since $G$ is abelian, the
Chinese remainder theorem decomposes
\begin{equation}
\FG \cong \bigoplus_\lambda R_\lambda    
\end{equation}
into local rings indexed by Frobenius orbits $\lambda$ of characters. The code
complex is the Koszul complex on $(a,b,c,d)$ and splits over this
decomposition. Its middle cohomology gives the closed form
\begin{equation}\label{eq:search-logical-count}
K = 6 \sum_\lambda |\lambda|
\min\bigl(v_a^\lambda, v_b^\lambda, v_c^\lambda, v_d^\lambda\bigr),
\end{equation}
where $v_x^\lambda$ is the $t$-adic valuation of $x$ in $R_\lambda$. This
formula provides an exact logical-count test. If $K<K_{\min}$, the candidate
is discarded.

The second test searches for an explicit low-weight logical operator. The
block form of the parity checks exposes logical operators supported on a
single qubit block. Fix one of the six blocks $T=\{x,y\}$. A vector supported
on this block is a $Z$-cycle exactly when the two off-block factors annihilate
it, and it is nontrivial exactly when it is not in the local stabilizer image.
This gives the upper bound
\begin{equation}\label{eq:search-distance-bound}
d \le U \coloneqq
\min_{T=\{x,y\}}
\min_{v \in
\frac{\Ann_R(z)\cap\Ann_R(w)}{xR + yR}} |v|,
\qquad \{z,w\} = \{a,b,c,d\}\setminus T,
\end{equation}
where the outer minimum runs over the six qubit blocks. Any vector attaining
the minimum is an explicit logical operator. Hence $U<d_{\min}$ certifies that
the candidate fails the distance threshold, and the candidate is discarded. We also tighten this distance ceiling with bounded two- and three-block
witness searches and with a Kunneth representative obtained from the same
orbit data. These additional witnesses use the same principle: they exhibit
logical operators and therefore provide upper bounds on $d$. Computing the
ceiling is still about an order of magnitude cheaper than running the distance
estimator.

The final test estimates the distance with Stern's information-set-decoding
algorithm, a meet-in-the-middle refinement of Prange decoding. At the
distance-to-length ratios of the codes in our search, Stern trials have one to
two orders of magnitude higher success probability at comparable cost. Each
successful trial returns an explicit logical operator, so the running minimum
weight is an upper bound on $d$ and can only decrease as trials accumulate.
The estimator stops early once it finds an operator below $d_{\min}$, so
candidates that pass the algebraic ceiling but still have low distance are
usually rejected after only a few trials. Since quadcycle codes satisfy
$d_X=d_Z$ (Section~\ref{sec: Transversal gates}~\cite{SI}), we only need to estimate the distance in one basis.

\begin{figure}[h]
\centering
\includegraphics[width=\textwidth]{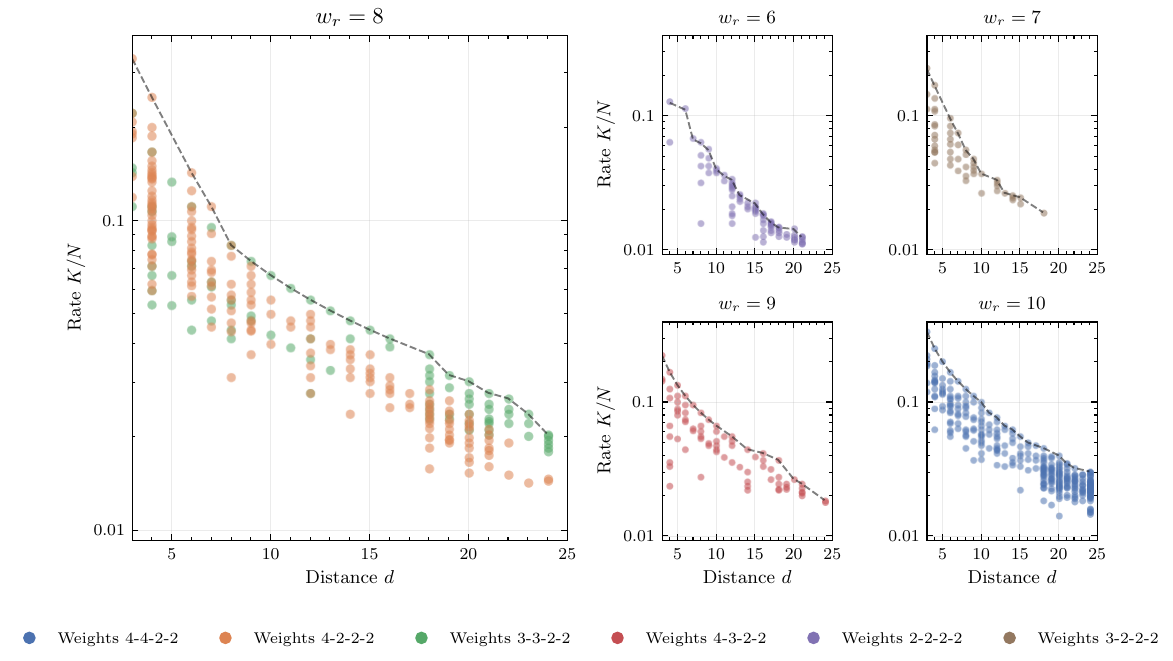}
\caption{Rate $K/N$ versus confirmed distance $d$, with one panel per
check weight $w_r$; the large panel shows $w_r=8$. Each point is a
code that satisfies the desiderata of the search, and the dashed line is the Pareto frontier.
Distances are upper bounds from Stern ISD.}
\label{fig:search-pareto}
\end{figure}

\Cref{fig:search-funnel} shows the resulting `funnel' for the search that produced
the $[[456,24,16]]$ code. The logical-count test removes $99.6\%$ of samples
at a cost of a few microseconds per sample. Roughly one in three hundred
candidates reaches the distance estimator. Although this final stage is invoked
on few candidates, it accounts for
more than $80\%$ of the wall time for the code search.

\section{Dimension jumping between 3D and 4D}\label{sec:DJ CNOT}
We describe the tricycle-to-quadcycle dimension jump in greater detail. It extends the dimension-jump method between 2D and 3D lifted product codes introduced in \cite{li_transversal_2025}. The key ingredient that facilitates the dimension jump is a homomorphic CNOT whose control acts on the tricycle codes and whose target acts on the quadcycle code. The chain map that enables the homomorphic CNOT is the one that identifies the tricycle code $T_\bullet$ with its embedding into the quadcycle code $Q_\bullet = \text{cone}(\tilde d)$.
\begin{equation}\label{eq: Diagram for dimension jump}
\begin{tikzcd}
    \F^{|G|} \arrow[r, "M_Z^T"] & 
    \F^{3|G|} \arrow[r, "H_Z^T"] & 
    \F^{3|G|} \arrow[r, "H_X"]  & 
    \F^{|G|}
    \\
    \F^{|G|} \arrow[r, "M_Z^T"] \arrow[dr, "\mathbb{B}(d)"] \arrow[u, "\text{id}_{|G|}"] &
    \F^{3|G|} \arrow[r, "H_Z^T"] \arrow[dr, "\mathbb{B}(d) \otimes \text{id}_3"] \arrow[u, "\text{id}_{3|G|}"] & 
    \F^{3|G|} \arrow[r, "H_X"] \arrow[dr, "\mathbb{B}(d) \otimes \text{id}_3"] \arrow[u, "\text{id}_{3|G|}"] & 
    \F^{|G|} \arrow[dr, "\mathbb{B}(d)"] \arrow[u, "\text{id}_{|G|}"] 
    \\
    & \F^{|G|} \arrow[r, "M_Z^T"] & 
    \F^{3|G|} \arrow[r, "H_Z^T"] & 
    \F^{3|G|} \arrow[r, "H_X"] & 
    \F^{|G|}
\end{tikzcd}
\end{equation}
As the chain map from the tricycle code to the quadcycle code is an embedding, the physical CNOT circuit that implements it has depth 1. In the main-text Section \hyperref[sec:implementation]{``Efficient implementation on neutral atom arrays''}, we also discuss how the CNOT is efficiently implementable on neutral atom hardware by virtue of its structure. Note that the direction of the identity map cannot be flipped unless $d=0$.

The logical action of the CNOT can be gleaned from \eqref{eq: LES}. Note that $q$ there describes the induced chain map on homology from the projection $Q_\bullet \rightarrow T_\bullet$, i.e., from the quadcycle to the tricycle code. At $n=1$, in the relevant case where the quadcycle code has twice as many logical qubits as the tricycle code, $d_* = 0$, so from exactness, $\im q = \ker d_*$, which is then all of $H_0(T_\bullet)$. The induced map on logical $Z$ operators is surjective. We further confirmed this numerically for all tricycle-quadcycle pairs studied in this work. This has the advantage that all the logical qubits of the tricycle code can be entangled with logical qubits of the quadcycle code, ensuring that all magic produced on the tricycle codes can be transversally teleported to the quadcycle code.

One can also change the connection of the auxiliary $T_\bullet$ complex.
\begin{equation}\label{eq: Related diagram}
\begin{tikzcd}
    \F^{|G|} \arrow[r, "M_Z^T"] \arrow[dr, "\mathbb{B}(d)"] &
    \F^{3|G|} \arrow[r, "H_Z^T"] \arrow[dr, "\mathbb{B}(d) \otimes \text{id}_3"] & 
    \F^{3|G|} \arrow[r, "H_X"] \arrow[dr, "\mathbb{B}(d) \otimes \text{id}_3"] & 
    \F^{|G|} \arrow[dr, "\mathbb{B}(d)"]
    \\
    & \F^{|G|} \arrow[r, "M_Z^T"] & 
    \F^{3|G|} \arrow[r, "H_Z^T"] & 
    \F^{3|G|} \arrow[r, "H_X"] & 
    \F^{|G|}
    \\
    & \F^{|G|} \arrow[r, "M_Z^T"] \arrow[u, "\text{id}_{|G|}"] & 
    \F^{3|G|} \arrow[r, "H_Z^T"] \arrow[u, "\text{id}_{3|G|}"] & 
    \F^{3|G|} \arrow[r, "H_X"] \arrow[u, "\text{id}_{3|G|}"] & 
    \F^{|G|} \arrow[u, "\text{id}_{|G|}"]
\end{tikzcd}
\end{equation}
The interpretation of this diagram is less clear since the auxiliary tricycle complex has been shifted by $+1$ degree. However, we can express the duality of the tricycle codes to make it more straightforward. First, tricycle codes generically satisfy $H_X = M_Z^T$. Thus,
\begin{equation}
\begin{tikzcd}
    T^*_\bullet = \F^{|G|} \arrow[r, "H_X^T"] & 
    \F^{3|G|} \arrow[r, "H_Z"] & 
    \F^{3|G|} \arrow[r, "H_X"] & 
    \F^{|G|}
\end{tikzcd}
\end{equation}
is the dual complex of $T_\bullet$. Let $\mathbb P$ denote the permutation matrix such that $\mathbb P \mathbb{B}(r) = \mathbb B(r^\dag)$ for any $r \in R$. Then, there is an explicit self-duality in the sense of there being an isomorphism of chain complexes $\tilde{\mathbb P}: T_\bullet^* \rightarrow T_\bullet$.
\begin{equation}
    \begin{tikzcd}
        \F^{|G|} \arrow[r, "H_X^T"] \arrow[d, "\text{id}_{|G|}"] & 
    \F^{3|G|} \arrow[r, "H_Z"] \arrow[d, "\mathbb{P} \otimes \text{id}_{3}"] & 
    \F^{3|G|} \arrow[r, "H_X"] \arrow[d, "\mathbb{P} \otimes \text{id}_{3}"] & 
    \arrow[d, "\text{id}_{|G|}"] \F^{|G|} 
    \\
    \F^{|G|} \arrow[r, "H_X^T"] & 
    \F^{3|G|} \arrow[r, "H_Z^T"] & 
    \F^{3|G|} \arrow[r, "H_X"] & 
    \F^{|G|}
    \end{tikzcd}
\end{equation}
Then, we can replace \eqref{eq: Related diagram} instead with
\begin{equation}
    \begin{tikzcd}
    \F^{|G|} \arrow[r, "M_Z^T"] \arrow[dr, "\mathbb{B}(d)"] &
    \F^{3|G|} \arrow[r, "H_Z^T"] \arrow[dr, "\mathbb{B}(d) \otimes \text{id}_3"] & 
    \F^{3|G|} \arrow[r, "H_X"] \arrow[dr, "\mathbb{B}(d) \otimes \text{id}_3"] & 
    \F^{|G|} \arrow[dr, "\mathbb{B}(d)"]
    \\
    & \F^{|G|} \arrow[r, "M_Z^T"] & 
    \F^{3|G|} \arrow[r, "H_Z^T"] & 
    \F^{3|G|} \arrow[r, "H_X"] & 
    \F^{|G|}
    \\
    & \F^{|G|} \arrow[r, "H_X^T"] \arrow[u, "\text{id}_{|G|}"] & 
    \F^{3|G|} \arrow[r, "H_Z"] \arrow[u, "\mathbb{P} \otimes \text{id}_{3}"] & 
    \F^{3|G|} \arrow[r, "M_Z"] \arrow[u, "\mathbb P \otimes \text{id}_{3}"] & 
    \F^{|G|} \arrow[u, "\text{id}_{|G|}"]
    \end{tikzcd}
\end{equation}
The bottom complex can be understood as a Hadamard-conjugated tricycle code: the metachecks now lie in the $X$ basis, and the resulting code has a larger $Z$ distance than $X$ distance. As $\mathbb P$ is a permutation matrix, the transversal-CNOT described by $\mathbb P \otimes \text{id}_3$ also has depth 1.

\section{Restoring circuit distance with the fan-out gadget}\label{app:concat}
In this section we describe the fan-out gadget that reduces the transversal $\overline{\mathrm{CCZ}}$ from depth two to depth one, and show that it restores the circuit-level distance of the factory. To see why this is necessary, we illustrate how the depth $2$ transversal CCZ can lead to a distance reduction. When performing the first CCZ, consider a single gate fault producing the error $Z_iX_j X_k$, where $i, j, k \in [n]$ are tricycle qubits in the three different code blocks. In a $Z$-biased noise model, these errors occur less frequently, but persist at $O(p_{\mathrm{CCZ}})$ nonetheless. Then, in the second CCZ layer, qubit $j$ may interact with another qubit $i'$ in the first block, while qubit $k$ may interact with an entirely different qubit $i''$ in the first block. The $X$ errors on these two qubits propagate to a phase error consisting of CZ and $Z$ terms supported on qubits $i, i'$, and $i''$ simultaneously. In the transversal $X$ measurement of the tricycle code, this correlated phase error stochastically collapses to a pattern of $Z$ errors. One such pattern, which occurs with constant probability, consists of $Z_i Z_{i'} Z_{i''}$, while also having no $Z$ support on any qubits of the second and third code blocks. Although the $X$ errors on those blocks remain, they are undetected by our postselection strategy as we measure transversally in the $X$ basis. Thus, a single gate fault can produce a weight-$3$ error on the first tricycle block. We verified numerically that such an error exists in the smallest instance of our factory---occurring on the $[[27, 3, 3]]$ code following the first CCZ layer---leading to a logical $Z$ error. This reduces the circuit-level distance of a factory based on the depth $2$ CCZ of the $[[27, 3, 3]]$ code to $d = 1$.

\begin{figure}[t!]
    \centering
    \includegraphics[width=0.6\linewidth]{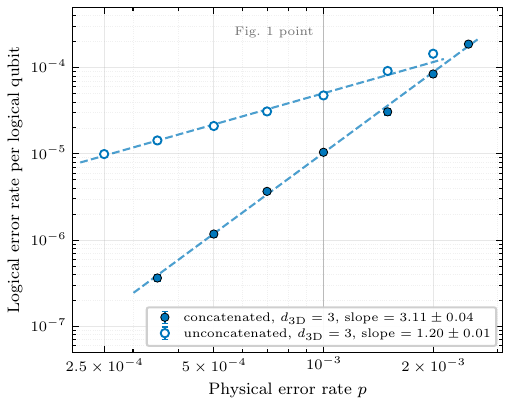}
    \caption{\textbf{Recovering $p^{d}$ scaling through the fan-out gadget.}
    Logical error rate per logical qubit versus physical error rate $p$, with all
    noise rates scaled together and $p_{\mathrm{CCZ}} = 2p$. The fitted slope is the
    effective error exponent and should equal the 3D code distance $d_{\textnormal{3D}}$.
    The factory with the gadget, $2\times[[27,3,3]]\rightarrow[[54,6,6]]$, reaches it. The factory without the gadget,
    $[[27,3,3]]\rightarrow[[54,6,6]]$, does not: its depth-2 $\overline{\mathrm{CCZ}}$
    is not fault tolerant, since a single error after the first CCZ layer spreads to three
    errors on the top simulated block, aligned along a logical operator. Dashed lines are
    weighted power-law fits. Error bars are Poisson on the logical failure count. The
    vertical line marks the operating point of Fig.~\ref{fig:dimjump}.}
    \label{fig:scaling}
\label{fig:fanout-scaling}
\end{figure}

To recover the full distance of the factory, we curb this error propagation between consecutive CCZ layers with the fan-out gadget: each tricycle block is paired with an ancilla patch of the same size, prepared in $\ket{\overline{0}}$ and entangled with it by one transversal layer of CNOTs. To analyze the gadget it is convenient to describe a block together with its ancilla patch as a single code, namely the tricycle code concatenated with an inner $d=2$ repetition code, with parity-check matrices
\begin{equation}
    \tilde H_X = \begin{pmatrix}
        H_X & H_X
    \end{pmatrix} \qquad \tilde H_Z = \begin{pmatrix}
        H_Z & 0 \\
        I & I
    \end{pmatrix},
\end{equation}
where $H_X$ and $H_Z$ are the parity-check matrices of the original tricycle code. The concatenated code has parameters $[[2n_{\textnormal{3D}}, k_{\textnormal{3D}} = 3, (d^X = 2d^X_{\textnormal{3D}}, d^Z = d^Z_{\textnormal{3D}})]]$. The $\ket{\overline +}^{\otimes 3}$ state of this concatenated code is also single-shot state preparable, as shown in Fig.~\ref{fig:dimjump}, since a $d=2$ repetition code can be encoded unitarily. The protocol consists of preparing a tricycle block in the usual way, followed by initializing $n_{\textnormal{3D}}$ additional qubits in $\ket{\bar 0}$ and performing a transversal CNOT. It can alternatively be viewed as preparing logical Bell pairs in the tricycle code \cite{bonilla_ataides_constant-overhead_2025}. These two stages can be performed in either order, but by preparing the repetition code second, we can re-use the $n_{\textnormal{3D}}$ qubits that serve as ancillae during the $Z$ basis syndrome extraction as the other half of the concatenated block. Despite the concatenation, the peak number of qubits needed to prepare $\detstate$ on the tricycle blocks remains $2n_{\textnormal{3D}}$ per block. Preparing this concatenated code also remains compact, increasing the two-qubit gate depth only incrementally from $4$ to $5$. Note that the weights of the $X$ checks in the resulting code are doubled to $12$ with this construction; however, these checks are never measured directly, as demonstrated in Fig.~\ref{fig:dimjump}, and are instead inferred classically from the transversal measurement of the tricycle blocks.

We now demonstrate that this concatenation scheme can indeed be used to reduce the CCZ depth of the tricycle code from $2$ to $1$. The approach is based on the ability to fan out controlled-$U$ gates, which reduces the circuit depth at the cost of ancilla qubits \cite{moore1998parallel}, and can be extended to the setting of logical qubits. The codespace of a tricycle code can be described by $\mathcal C \subset \mathcal H^{\otimes n}$, a subspace of the $n$-qubit Hilbert space. Any state $\ket{\overline \psi_i} \in \mathcal C$ in the codespace can be expanded in the computational basis as $\ket{\overline \psi_i} = \sum_{z} c_z^{(i)} \ket z $
where $z \in \F^n$. Define $\vec z := (z_1,z_2,z_3) \in \F^{3n}$ and $c_{\vec z} = c_{z_1}^{(1)} c_{z_2}^{(2)} c_{z_3}^{(3)}$, so that $\ket{\overline{\psi_1 \psi_2 \psi_3}} = \sum_{\vec z} c_{\vec z} \ket{\vec z}$. The CCZ circuit $U_{\mathrm{CCZ}}: \mathcal C^{\otimes 3} \rightarrow \mathcal C^{\otimes 3}$ is described by a diagonal $2^{3n} \times 2^{3n}$ matrix in the computational basis with entries $\pm 1$ because it comprises physical CCZ gates; and since it is diagonal, $U_{\mathrm{CCZ}} \ket{\vec z} = (-1)^{f(\vec z)}\ket{\vec z}$ for some $f: \F^{3n} \rightarrow \F$. Let $V, W$ be matrices of the same form, with $V\ket{\vec z} = (-1)^{g(\vec z)}\ket{\vec z}$ and $W\ket{\vec z} = (-1)^{h(\vec z)}\ket{\vec z}$. Suppose $U_{\mathrm{CCZ}} = VW$: then $f(\vec z) = g(\vec z) + h(\vec z)$. Using this fact, consider when $\ket{\overline \psi_i}$ is concatenated into a repetition code. The encoded state is $\ket{\overline \psi_i'} = \sum_z c_z^{(i)} \ket{z} \otimes \ket{z}$. It is easy to see that $U_{\mathrm{CCZ}} \otimes I$ persists as a transversal CCZ circuit in the concatenated code since it commutes with the CNOT that (un)-encodes the repetition code.
\begin{align}
    \ket{\overline{\psi_1' \psi_2' \psi_3'}}
    \xRightarrow{\text{Un-encode}} \sum_{\vec z} c_{\vec z} \ket{\vec z} \ket{0}^{\otimes 3n}
    \xRightarrow{U_{\mathrm{CCZ}}} \sum_{\vec z} c_{\vec z} (-1)^{f(\vec z)}\ket{\vec z} \ket{0}^{\otimes 3n} \xRightarrow{\text{Encode}} \sum_{\vec z} c_{\vec z} (-1)^{f(\vec z)}\ket{\vec z} \otimes \ket{\vec z}
\end{align}
which equals $(U_{\mathrm{CCZ}} \otimes I)\ket{\overline{\psi_1' \psi_2' \psi_3'}}$.
If we instead act on the concatenated state with $V \otimes W$, 
\begin{equation}
    (V \otimes W)  \ket{\overline{\psi_1' \psi_2' \psi_3'}} = \sum_{\vec z} c_{\vec z} V\ket{\vec z} \otimes W \ket{\vec z} = \sum_{\vec z} c_{\vec z} (-1)^{g(\vec z) + h(\vec z)}\ket{\vec z} \otimes \ket{\vec z} = (U_{\mathrm{CCZ}} \otimes I)\ket{\overline{\psi_1' \psi_2' \psi_3'}},
\end{equation}
so it is logically equivalent to applying $U_{\mathrm{CCZ}}$, followed by encoding. Choosing $V$ and $W$ to correspond to the first and second layers of transversal CCZ gates in $U_{\mathrm{CCZ}}$, we obtain a depth $1$ realization.

There is also a homomorphic CNOT between the concatenated tricycle code and the 4D code. Let us denote the concatenated tricycle code by $D_\bullet$,
\begin{equation}
    D_\bullet = 
    \begin{tikzcd}
\F^{|G|} \arrow[r, "(M_Z \ 0)^T "] &
    \F^{6|G|} \arrow[r, "\tilde H_Z^T"] & 
    \F^{6|G|} \arrow[r, "\tilde H_X"] & 
    \F^{|G|}
    \end{tikzcd}
\end{equation}
Then, we have the following chain map $\gamma: T_\bullet \rightarrow D_\bullet$, which gives rise to a depth-1 CNOT between a concatenated and unconcatenated tricycle code. 
\begin{equation}
    \begin{tikzcd}
\F^{|G|} \arrow[r, "(M_Z \ 0)^T "] & 
    \F^{6|G|} \arrow[r, "\tilde H_Z^T"] & 
    \F^{6|G|} \arrow[r, "\tilde H_X"] & 
    \F^{|G|} \\
\F^{|G|} \arrow[r, "M_Z^T"] \arrow[u, "\text{id}"] & 
    \F^{3|G|} \arrow[r, "H_Z^T"] \arrow[u, "\binom{I}{0}"] & 
    \F^{3|G|} \arrow[r, "H_X"] \arrow[u, "\binom{I}{0}"] & 
    \F^{|G|} \arrow[u, "\text{id}"] 
    \end{tikzcd}
\end{equation}
One can verify that $\gamma$ acts as a logical isomorphism. 
When we then compose this with the chain map for the dimension jump CNOT $\pi: Q_\bullet \rightarrow T_\bullet$ from \eqref{eq: Diagram for dimension jump}, we obtain the depth-1 transversal CNOT $\gamma \circ \pi: Q_\bullet \rightarrow D_\bullet$.

With the transversal CCZ circuit now being reduced to depth 1, we no longer suffer from the problematic error propagation of the depth 2 CCZ circuit. Crucially, the same number of CCZ gates is applied in aggregate, and therefore the number of fault locations introduced by the gates remains unchanged. The benefit of this concatenation is that distributing the circuit over more qubits controls the error propagation. We confirmed this numerically, with the results being shown in Fig.~\ref{fig:scaling}. As we vary the physical error rate $p \rightarrow 0$, for a factory instance with distance $d$ operating in full error detection, we expect the logical error rate to scale as $P_L \sim Ap^{d}$, or $\log P_L \sim d \log p + C$. For the unconcatenated tricycle code with a depth-2 CCZ, we observe a fitted exponent of $\hat d = 1.20 \pm 0.01$, in agreement with our analysis that the circuit distance of this factory drops well below $3$. As $p \rightarrow 0$, the suppression of the factory is thought to be distance-limited, so the fitted exponent is expected to approach $d=1$. At finite $p$, other factors contribute to the error suppression \cite{gu2026scalable} that should only maintain or increase the suppression, making $\hat d$ an effective upper bound on the effective distance. With concatenation and a depth-1 CCZ, we observe a slope of $\hat d = 3.11 \pm 0.04$, in agreement with $d = 3 \leq \hat d$.

\section{Further details of numerical simulations}\label{sec:numerics-methods}

We employed the standard circuit-level depolarizing error model (SD6) in our numerical simulations (which are performed in Stim \cite{gidney_stim_2021}), where idling and single-qubit gates are followed by single-qubit depolarization channels, two-qubit gates are followed by two-qubit depolarizing channels, and resets/measurements are followed by a single-qubit flip channel, all at strength $p$. This is the same error model employed in Refs. \cite{gidney_magic_2024, sahay2026fold} on magic state cultivation, offering a fair comparison. For the CCZ gates, we derived an effective channel arising from a three-qubit depolarization channel (in the uniform nonidentity-Pauli convention) with strength $p_{\mathrm{CCZ}} = 2p$. Since we simulate only a single tricycle-quadcycle pair, we traced out the other tricycle code qubits involved in the CCZ error channel, resulting in a marginal noise channel equivalent to single-qubit depolarization at a strength of $16/21\times p_{\mathrm{CCZ}}$, or $32/21 \times p$. After estimating the acceptance fractions $\alpha$ for the individual tricycle-quadcycle pairs we simulate (Fig.~\ref{fig:dimjump}(c)), we calculate the success rate of the entire factory as $\alpha^3$, treating the three pairs as independent, which is what we use in the expected space-time volume estimates of Table~\ref{tab:comparison}. From both a logical error rate and acceptance fraction perspective, we expect this to be a conservative means of modeling the full factory since it erases the cross-block correlations in the CCZ error channel.

For the quadcycle memory simulations, we optimize over interleaved, depth-optimal syndrome extraction macroschedules \cite{tan26syndrome}, picking for each code a schedule that attains the maximum single-round circuit-level distance. Similar to the situation for bivariate bicycle codes \cite{bravyi_high-threshold_2024, strikis2026highperformancesyndromeextractioncircuits}, these circuits typically have a circuit distance lower than their nominal code-capacity distance. For the factory simulations, in which the state preparation round of syndrome extraction needs only to act in a single basis, we optimize the tricycle $Z$-basis syndrome-extraction and quadcycle $X$-basis syndrome-extraction circuits over schedules obtained from edge coloration, again selecting the schedule with the highest circuit distance.

\begin{figure}[ht!]
      \centering
      \includegraphics[width=0.6\linewidth]{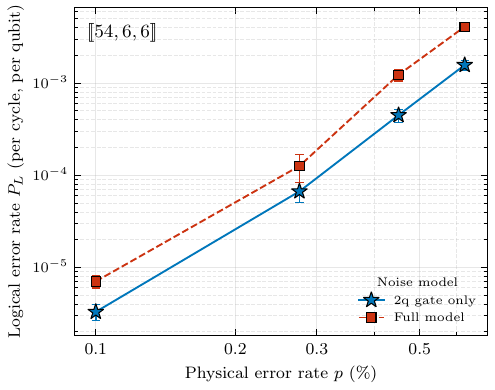}
      \caption{\textbf{Two-qubit gate errors dominate the logical failure rate.}
      Logical error rate per
      cycle per logical qubit for the $[[54,6,6]]$ quadcycle code under single-shot
      windowed most-likely-error decoding. Blue: noise on the two-qubit gates only.
      Red: the full circuit-level model, adding reset, measurement, and idle errors
      at the same rate $p$. The extra channels raise the logical error rate by only
      about a factor of two, with no change in slope. Error bars are binomial.}
  \label{fig:noise_model_54}
  \end{figure}

\section{Review of tricycle magic}\label{sec: Tricycle magic}
 Tricycle codes do not generally produce hypergraph magic states with disjoint \cczstate states, but rather some more complex logical hypergraph magic state on the $3k$ logical qubits. The exact structure of the hypergraph magic state is governed by the details of the cup product and generally also depends on the choice of logical basis for the tricycle codes. Prior work formulated the question of determining how many disjoint magic states can be extracted from a given hypergraph magic state by an optimal choice of logical basis and resetting certain logical qubits to $\ket 0$ prior to the CCZ in terms of the maximum subrank of a rank-3 binary tensor. Unfortunately, this problem is NP-hard in general, and the authors could only compute it exactly with $k \sim 4$. 

The $\detstate$ magic state produced by the 2-2-2 tricycle codes that we focus on in this work is a special case wherein the hypergraph magic state is independent of the choice of logical basis. We express it in terms of a connected hypergraph with six overlapping hyperedges as shown in Fig.~\ref{fig:dimjump}. The CCZs can be described as a set of triples $(i, j, k)$ denoting a CCZ between logical qubit $i$ in the first block, $j$ in the second block, and $k$ in the third block. The set of these triples for the $\mathrm{CCZ}$ on the 2-2-2 tricycle codes is described by $\overline{\mathrm{CCZ}}(\sigma)$ where $\sigma \in S_3$ is the one-line notation for the permutation. The $\overline{\text{det}}$ gate implements the determinant (or equivalently over $\F$, the permanent) of a $3 \times 3$ binary matrix. It is possible that this gate could be a useful arithmetic primitive in certain algorithms, the exploration of which we leave to future work. Instead, we focus on extracting \cczstate states from \detstate.

\section{\texorpdfstring{Extracting disjoint \cczstate states}{Extracting disjoint CCZ states}}\label{sec:Extracting CCZs}
We briefly discuss the protocol to extract two disjoint \cczstate states from \detstate. We use the notation $\bar q_{i\alpha}$ to denote the $i$th logical qubit on the $\alpha$th block, where $i, \alpha \in \{1, 2, 3\}$. Writing each permutation in $S_3$ as a triple $(i,j,k)$, recall that
\begin{equation}
    \detstate = \Biggl[\prod_{(i,j,k)\in S_3} \overline{\mathrm{CCZ}}(\bar q_{i1}, \bar q_{j2}, \bar q_{k3})\Biggr] \ket{\overline +}^{\otimes 9}.
\end{equation}
Consider instead the magic state where the unitary acts on a product state such that $\ket{\bar q_{\ell, \alpha}} = \ket{\bar 0}$ if $\ell = \alpha$; otherwise, $\ket{\bar q_{\ell, \alpha}} = \ket{\overline +}$. More generally, of the three logical qubits per tricycle code block, resetting a different one for each of the three blocks to $\ket{\bar 0}$ results, after an application of $\overline{\text{det}}$, in the following state containing two disjoint \cczstate states.
\begin{equation}
    \begin{aligned}
    &\Biggl[\prod_{(i,j,k)\in S_3} \overline{\mathrm{CCZ}}(\bar q_{i1}, \bar q_{j2}, \bar q_{k3})\Biggr]
    \Biggl[\bigotimes_{\alpha=1}^{3}\Biggl(\ket{\bar 0_{\alpha\alpha}}\otimes\bigotimes_{\substack{i=1\\i\ne\alpha}}^{3}\ket{\overline{+}_{i\alpha}}\Biggr)\Biggr]\\
    &\qquad= \ket{\overline{\mathrm{CCZ}}(\bar q_{12}, \bar q_{23}, \bar q_{31})}\otimes\ket{\overline{\mathrm{CCZ}}(\bar q_{13}, \bar q_{21}, \bar q_{32})}\otimes\bigotimes_{\alpha=1}^{3}\ket{\bar 0_{\alpha\alpha}}.
    \end{aligned}
\end{equation}
We refer to this subset of logical qubits as gauge qubits \cite{menon_magic_2025}. However, in the protocol we present in the main text, $\ket{\bar q_{\ell, \alpha}}$ is measured in $Z$ after the action of the CCZ circuit. One reason for this is that the aforementioned, selective reset protocol requires fault-tolerant gadget operations on the tricycle codes prior to the transversal CCZ. However, this contributes additional error locations and therefore reduces the acceptance fraction of error detection, so we instead elect to conjugate this measurement through the transversal CCZ and dimension jump steps. The result is that these logical qubits are instead measured afterwards, with a conditional fixup needed depending on the measurement outcome. Let $\mu(1) \subseteq \{1, 2, 3\}$ record the logical measurements: $\bar q_{\alpha \alpha}$ was measured to be in $\ket{\bar 1}$ iff $\alpha \in \mu(1)$. Then, the resulting state is
\begin{equation}\label{eq:Collapsed magic state}
    \begin{aligned}
    &\Biggl[\prod_{\substack{\alpha \in \mu(1)\\
    \beta,\gamma\in\{1,2,3\}\setminus\{\alpha\}\\
    \beta < \gamma }}
    \overline{\mathrm{CZ}}(\bar q_{\beta\beta}, \bar q_{\gamma\gamma})\overline{\mathrm{CZ}}(\bar q_{\beta\gamma}, \bar q_{\gamma\beta})\Biggr]\\
    &\quad\times\Biggl[\ket{\overline{\mathrm{CCZ}}(\bar q_{12}, \bar q_{23}, \bar q_{31})}\otimes\ket{\overline{\mathrm{CCZ}}(\bar q_{13}, \bar q_{21}, \bar q_{32})}
    \otimes\bigotimes_{\alpha\in\mu(1)}\ket{\bar 1_{\alpha\alpha}}\otimes\bigotimes_{\delta\in\{1,2,3\}\setminus\mu(1)}\ket{\bar 0_{\delta\delta}}\Biggr].
    \end{aligned}
\end{equation}
Aside from these measured logical qubits lying in the computational basis, this resulting logical state is a hypergraph magic state consisting of two disjoint $\cczstate$ states dressed with further CZ gates, determined by the (random) measurement outcomes (see Fig. \ref{fig:gauging_and_fixups}(a)). Note that some of the logical CZ gates in \eqref{eq:Collapsed magic state} act on two logical qubits, both of which are measured to be in computational basis states. Thus, these CZ gates reduce simply to a global phase and can be ignored. The remaining CZ gates must be removed by applying them a second time. 

In Section~\ref{sec:zero_out}~\cite{SI} we present a simple method with structured ancilla codes to perform the $Z$ logical measurement of the gauge qubits using the framework of homomorphic logical measurements \cite{huang_homomorphic_2023}. Alternatively, for codes admitting clustered logical bases (Section~\ref{sec: Transversal gates}~\cite{SI}), the measurements can be implemented using graph-based generalized lattice surgery techniques \cite{cohen2022low, williamson2026low, swaroop_universal_2025}. More generally, recent high-rate approaches to qLDPC lattice surgery \cite{zheng_high-rate_2025} may offer alternative approaches to implementing this surgery in parallel with low overhead, and recently developed single-shot methods may also reduce the temporal overhead to $O(1)$ \cite{baspin2025fast, cowtan_fast_2025} -- we leave an exploration of these directions to future work. The CZ gates dressing the magic state can be fixed in parallel with the feed-forward correction \cite{chamberland2022temporal} -- see Section~\ref{supp:teleporting-cz}~\cite{SI}.

\section{\texorpdfstring{Homomorphic measurement gadgets for \cczstate extraction}{Homomorphic measurement gadgets for CCZ extraction}}\label{sec:zero_out}

After teleportation, the $\detstate$ state is distributed across three quadcycle blocks, occupying three of the six logical qubits in each block. The remaining three logical qubits are unused and held in $|\bar0\rangle$. Without loss of generality, we label the qubits with the $\detstate$ as $1,2,3$ and the remaining logical qubits as $4,5,6$. The extraction requires measuring $\bZ$ on compute qubit $1$ of the first block, on compute qubit $2$ of the second block, and on compute qubit $3$ of the third block. Let $x_{ri}\in\F$ denote the logical bit of compute qubit $i$ on block $r$. A change of logical basis by an invertible matrix $B\in\mathrm{GL}(3,\F)$ applied to the compute qubits of all three blocks multiplies the determinant of $(x_{ri})$ by $\det B=1$ and hence preserves the $\detstate$ state. It therefore suffices to measure on block $r$ a logical $Z$ operator $\bZ^{v_r}=\prod_i\bZ_i^{(v_r)_i}$ of the compute qubits for any linearly independent $v_1,v_2,v_3\in\F^3$, since $B$ can be chosen to map $v_1,v_2,v_3$ to the three unit vectors. The action of the measured operators on the unused qubits can be arbitrary. Here we construct fault-tolerant gadgets for these logical measurements and determine their fault distances.

Recall that a quadcycle code with factors $a,b,c,d\in\R=\F[G]$ has six qubit sectors labeled by pairs of factors, $ab,ac,ad,bc,bd,cd$, each of size $n=|G|$ and indexed by group elements. Its $X$ checks form four sectors labeled by triples. The $X$ checks of sector $abc$ act on the qubit sectors $ab,ac,bc$ through $\mathbb{B}(c),\mathbb{B}(b),\mathbb{B}(a)$, the multiplier being the factor of the triple absent from the pair, and likewise for $abd,acd,bcd$. Its $Z$ checks form four sectors labeled by single factors. The $Z$ checks of sector $a$ act on $ab,ac,ad$ through $\mathbb{B}(b),\mathbb{B}(c),\mathbb{B}(d)$, and likewise for $b,c,d$. All factors of the codes considered have even weight.

We measure the logical operators with an ancillary quadcycle code block, defined by factors $a',b',c',d'\in\R$. In the circuits constructed here, three of these factors coincide with the factors of the data code and the fourth is a new element of $\R$, related to the data factors by the $\R$-linear relation~\eqref{eq:cprime} below. We describe the case in which $c'$ is the new factor, so that $a'=a$, $b'=b$ and $d'=d$. The other cases follow by renaming the factors. The construction is based on \textit{homomorphic logical measurement}, using a constant-depth $\mathrm{CNOT}$ circuit that couples target logicals of the data block to the ancilla block, which is then measured transversally -- we refer the reader to Ref.~\cite{huang_homomorphic_2023} for more details on homomorphic measurement. We first describe the circuit between the data block and the ancilla block, then show that it preserves the code space of both blocks, and finally determine the logical operators it measures. The basic construction is determined by chosen group elements $g$ and $t$ and consists of three layers of CNOT gates with data qubits as controls and ancilla qubits as targets.
\begin{enumerate}[nosep,leftmargin=*]
\item For every sector and position $h$, CNOT from the data qubit at $h$ to the ancilla qubit at $h$ of the same sector.
\item For the sectors $ac,bc,cd$ and every $h$, CNOT from the data qubit at $h$ to the ancilla qubit at $hg$ of the same sector.
\item For every $h$, CNOT from the data qubit at $h$ of sector $bc$ to the ancilla qubit at $ht$ of sector $ab$, and from the data qubit at $h$ of sector $cd$ to the ancilla qubit at $ht$ of sector $ad$.
\end{enumerate}
For certain codes, we may only need $2$ layers, while for the $[[162,6,14]]$ code, certain measurements require a depth $4$ circuit (see \cref{tab:hom_meas}) which we elaborate on below.

The layers contain $6n$, $3n$ and $2n$ gates, and the circuit has depth three. We encode it in a $6\times6$ matrix $F$ over $\R$, indexed by ancilla sectors (rows) and data sectors (columns), whose entry $F_{P'P}$ is the set of shifts $s$ such that the circuit contains CNOTs from position $h$ of data sector $P$ to position $hs$ of ancilla sector $P'$. Thus $F_{PP}=1$ for $P\in\{ab,ad,bd\}$, $F_{PP}=1+g$ for $P\in\{ac,bc,cd\}$, $F_{ab,bc}=F_{ad,cd}=t$, and all other entries vanish. Under conjugation by the circuit, $X$ on a set of data qubits acquires $X$ on the ancilla qubits reached through $F$, and $Z$ on a set of ancilla qubits acquires $Z$ on the data qubits reached through $F^\T$.

The circuit and the ancilla factors are generated by a single $4\times4$ matrix $A$ over $\R$ relating the two factor lists $x=(a,b,c,d)$ and $x'=(a',b',c',d')$,
\begin{equation}
x_j=\sum_{i\in\{a,b,c,d\}}A_{ij}\,x'_i .
\label{eq:app-sub}
\end{equation}
Here $A$ is the identity except for $A_{cc}=1+g$ and $A_{ac}=t$, so that $a'=a$, $b'=b$, $d'=d$ and
\begin{equation}
(1+g)\,c'=c+t\,a .
\label{eq:cprime}
\end{equation}
A $k\times k$ minor of $A$ is the determinant of the submatrix obtained by keeping $k$ rows and $k$ columns, computed in $\R$, where it coincides with the permanent since $\R$ has characteristic two. For index sets $S'$ (rows) and $S$ (columns) we write $A[S',S]$ for this minor. Since rows and columns of $A$ are indexed by factors, a $2\times2$ minor is indexed by an ancilla qubit sector and a data qubit sector, a $3\times3$ minor by an ancilla and a data $X$-check sector, and the $4\times4$ minor is $\det A$. The entry $F_{P'P}$ of the circuit matrix is the minor $A[P',P]$, which reproduces the entries listed above.

\begin{prop}\label{prop:chain}
Let $A\in\R^{4\times4}$ satisfy~\eqref{eq:app-sub} and let $F$ be the matrix of its $2\times2$ minors. Under conjugation by the CNOT circuit defined by $F$,
\begin{enumerate}[nosep,leftmargin=*,label=(\roman*)]
\item the $X$ checks of data sector $T$ acquire the combination $\sum_{T'}A[T',T]$ of the $X$ checks of ancilla sectors $T'$,
\item the $Z$ checks of ancilla sector $i'$ acquire the combination $\sum_{i}A_{i'i}$ of the $Z$ checks of data sectors $i$.
\end{enumerate}
Consequently, both blocks remain in their code spaces, and every logical $Z$ operator of the ancilla block acquires a logical $Z$ operator or a stabilizer of the data code.
\end{prop}

\begin{proof}
Write $X_T$ for the $X$ checks of data sector $T$ as a vector over the six qubit sectors, with entry $x_{T\setminus P}$ on the sector $P\subset T$ and $0$ elsewhere, and $Z_i$ for the $Z$ checks of sector $i$, with entry $x_{P\setminus i}$ on $P\ni i$. Primed symbols refer to the ancilla block. Statement (i) is the identity $FX_T=\sum_{T'}A[T',T]\,X'_{T'}$. Its left side has entry $\sum_{P\subset T}A[P',P]\,x_{T\setminus P}$ on ancilla sector $P'$. Substituting $x_{T\setminus P}=\sum_lA_{l,T\setminus P}\,x'_l$ from~\eqref{eq:app-sub} and expanding the $3\times3$ minor $A[P'\cup\{l\},T]$ along its row $l$ gives $\sum_{l\notin P'}A[P'\cup\{l\},T]\,x'_l$, which is the entry of the right side on $P'$. Terms with $l\in P'$ vanish because the corresponding determinant has a repeated row. Statement (ii) is the identity $Z'_{i'}F=\sum_iA_{i'i}Z_i$. Its left side has entry $\sum_{l}x'_l\,A[\{i',l\},\{j,k\}]$ on data sector $\{j,k\}$, and expanding the $2\times2$ minors along the row $i'$ and using~\eqref{eq:app-sub} gives $A_{i'j}x_k+A_{i'k}x_j$, which is the entry of the right side. Both identities hold over $\R$ and hence, after applying $\mathbb{B}$, as binary matrix identities. For the last claim, let $z'$ be the support of an ancilla logical $Z$ operator, so that every ancilla $X$ check commutes with it. By (i), every data $X$ check acquires ancilla $X$ checks only, and therefore the data operator with support $F^\T z'$ commutes with every data $X$ check.
\end{proof}

For $A$ as in~\eqref{eq:cprime} the identities read $FX_{abd}=X'_{abd}$, $FX_{abc}=(1+g)X'_{abc}$, $FX_{acd}=(1+g)X'_{acd}$, $FX_{bcd}=(1+g)X'_{bcd}+tX'_{abd}$ and $Z'_aF=Z_a+tZ_c$, $Z'_bF=Z_b$, $Z'_cF=(1+g)Z_c$, $Z'_dF=Z_d$. Together with the analogous identities for the metachecks, Proposition~\ref{prop:chain} states that the matrices of $k\times k$ minors of $A$ form a chain map from the data complex to the ancilla complex, where $C_2=\R^6$ is the qubit space, $C_3=C_1=\R^4$ are the $X$- and $Z$-check spaces, $C_4=C_0=\R$ are the metacheck spaces, and the horizontal maps are the check and metacheck matrices.
\[
\begin{tikzcd}[column sep=large, row sep=large]
\text{\small $X$ metacheck} & \text{\small $X$ checks} & \text{\small qubits} & \text{\small $Z$ checks} & \text{\small $Z$ metacheck}\\[-18pt]
C_4 \arrow[r, "\partial_4"] \arrow[d, "\det A"'] & C_3 \arrow[r, "\partial_3"] \arrow[d, "3\times3\text{ minors}"'] & C_2 \arrow[r, "\partial_2"] \arrow[d, "F"'] & C_1 \arrow[r, "\partial_1"] \arrow[d, "A"'] & C_0 \arrow[d, "1"'] \\
C'_4 \arrow[r, "\partial'_4"'] & C'_3 \arrow[r, "\partial'_3"'] & C'_2 \arrow[r, "\partial'_2"'] & C'_1 \arrow[r, "\partial'_1"'] & C'_0
\end{tikzcd}
\]
The outer squares give compatibility with the metachecks and hence single-shot decoding of the ancilla block. Measuring the logical $Z$ operators of the ancilla block after the circuit measures the data operators they acquire, and we now identify these.

\begin{prop}\label{prop:measured}
Let $|G|$ be odd. The $Z$ operators supported on a full qubit sector are the logical $Z$ operators of the quadcycle code, with compute operators $\bZ_{bc},\bZ_{ac},\bZ_{ab}$ and unused operators $\bZ_{ad},\bZ_{bd},\bZ_{cd}$. The ancilla sector operator on $P'$ acquires the data operator $\prod_P\bZ_P^{\,\wt(F_{P'P})}$, where $\wt(r)$ is the weight of $r\in\R$. For $A$ as in~\eqref{eq:cprime} the ancilla sector operators on $ab$, $ad$, $bd$ acquire $\bZ_{ab}\bZ_{bc}$, $\bZ_{ad}\bZ_{cd}$, $\bZ_{bd}$, and those on $ac$, $bc$, $cd$ acquire the identity. The circuit therefore measures exactly one compute operator, $\bZ_{ab}\bZ_{bc}$.
\end{prop}

\begin{proof}
For $r\in\R$ the matrix $\mathbb{B}(r)^\T$ maps the all-ones vector to $\wt(r)$ copies of itself, since every group element is reached exactly once from each element of the support of $r$. Every factor has even weight, so a $Z$ operator on a full sector commutes with every $X$ check, and for odd $|G|$ these six operators represent the six logical $Z$ operators. The same holds for the ancilla block. The ancilla sector operator on $P'$ has support $\mathbf 1_{P'}$ and acquires $Z$ on $F^\T\mathbf 1_{P'}=\sum_P\wt(F_{P'P})\mathbf 1_P$. Weight parity is a ring homomorphism $\R\to\F$, so $\wt(A[P',P])$ equals the minor of the binary matrix $\bar A$ obtained from $A$ by reducing entries modulo weight parity. For $A$ as in~\eqref{eq:cprime}, $\bar A$ is the identity except for $\bar A_{cc}=0$ and $\bar A_{ac}=1$, and its $2\times2$ minors are $\bar A[P,P]=1$ for $P\in\{ab,ad,bd\}$, $\bar A[ab,bc]=\bar A[ad,cd]=1$, and $0$ otherwise.
\end{proof}

The even weight of $1+g$ removes the sectors containing $c$ from the measurement, the entry $t$ multiplies $\bZ_{ab}$ by $\bZ_{bc}$, and the choice of $c'$ affects only the ancilla code. For $\code{72,6,8}$ and $\code{228,6,17}$ the order of $G$ is even and the sector operators are not logical. There we embed the three $X$ logical operators of the tricycle code on the sectors $ab,ac,bc$, extend them to a basis $\bX_1,\dots,\bX_6$ of the $X$ logical operators, and let $\bZ_1,\dots,\bZ_6$ be the dual basis. With this choice, we aim to measure three linearly independent operators formed from products of $\bZ_1,\bZ_2,\bZ_3$. Table~\ref{tab:hom_meas} lists the measured operator for every circuit. On each block the logical basis is chosen so that this operator is the $Z$ of the first compute qubit.

The remaining circuits in Table~\ref{tab:hom_meas} differ only in the nonzero entries of $A$. With $A_{ac}=0$ the third layer is absent and $\bZ_{ab}$ is measured. An additional entry $A_{bc}=t'$ adds a fourth layer from sector $ac$ to $ab$ and from $cd$ to $bd$ with shift $t'$, changes~\eqref{eq:cprime} to $(1+g)c'=c+ta+t'b$, and measures $\bZ_{ab}\bZ_{ac}\bZ_{bc}$. The entry $A_{cd}=t$ in place of $A_{ac}$, used once on $\code{162,6,14}$, moves the third layer to run from $ad$ to $ac$ and from $bd$ to $bc$, sets $(1+g)c'=c$ and $d'=d+tc'$, and measures $\bZ_{ab}$.

\begin{lemma}\label{lem:solve}
Let $g\in G$ have order $m$, let $N_g=\sum_{k<m}g^k$, and let $u\in\R$. The equation $(1+g)c'=u$ has a solution $c'\in\R$ if and only if $N_gu=0$, and its solutions form a coset of $N_g\R$.
\end{lemma}

\begin{proof}
Multiplication by $1+g$ and by $N_g$ are $\F$-linear maps of $\R$, and $(1+g)N_g=1+g^m=0$, so the image of multiplication by $1+g$ lies in the kernel of multiplication by $N_g$. An element lies in the kernel of multiplication by $1+g$ if and only if it is invariant under multiplication by $g$, that is, if and only if it is a sum of cosets of the subgroup $\langle g\rangle$. These elements form $N_g\R$, of dimension $|G|/m$, so the image of multiplication by $1+g$ has dimension $|G|-|G|/m$. Since $N_gh$ is the coset of $h$ for every $h\in G$, an element $u$ lies in the kernel of multiplication by $N_g$ if and only if every coset of $\langle g\rangle$ contains an even number of elements of the support of $u$. This kernel also has dimension $|G|-|G|/m$, so it equals the image, which proves the solvability criterion. The solutions of $(1+g)c'=u$ form a coset of the kernel of multiplication by $1+g$, which is $N_g\R$.
\end{proof}

Equation~\eqref{eq:cprime} therefore has a solution exactly when every coset of $\langle g\rangle$ contains an even number of elements of $c+ta$, and the solution is unique only up to $N_g\R$. Every solution yields a valid circuit by Proposition~\ref{prop:chain}, but $c'$ must have even weight for the ancilla block to possess the logical qubits that carry the measurement, and the solution determines the ancilla code. For odd $|G|$ a weight-2 factor $c=h_1(1+h_2)$ admits the weight-2 solution $c'=h_1(1+g)$ with $g^2=h_2$ at $t=0$. For $\code{72,6,8}$ and $\code{228,6,17}$ no even-weight solution exists at $t=0$, which necessitates the third layer. The ancilla block coincides with the data code when $c'=c$ solves~\eqref{eq:cprime}, as for $\code{54,6,6}$, whose diagonal entries $h(1+g)$ differ from $1+g$ by a relabeling of the ancilla qubits. A solution of weight $2$ preserves the check weights of the data code, and a solution of weight $4$ raises half of the ancilla check weights from $8$ to $10$.

The procedure is given in Algorithm~\ref{alg:meas}. The ancilla check round of step~5 and the data syndrome round of step~8 provide the information used in the fault analysis below.

\begin{algorithm}[H]
\caption{Homomorphic measurement of gauge logical qubits}
\label{alg:meas}
\begin{algorithmic}[1]
\Require data block with factors $a,b,c,d$. Group elements $g,t\in G$ and an even-weight factor $c'\in\R$ satisfying $(1+g)c'=c+ta$.
\Ensure eigenvalue of the measured compute operator.
\State Prepare the ancilla block, with factors $a,b,c',d$, in $|0\rangle^{\otimes N}$ and measure its $X$ checks.
\State Apply CNOT from data qubit $(P,h)$ to ancilla qubit $(P,h)$ for all sectors $P$ and positions $h$.
\State Apply CNOT from $(P,h)$ to $(P,hg)$ for $P\in\{ac,bc,cd\}$ and all $h$.
\State Apply CNOT from $(bc,h)$ to $(ab,ht)$ and from $(cd,h)$ to $(ad,ht)$ for all $h$.
\State Measure the ancilla $X$ checks.
\State Decode the ancilla $Z$ error $e$ from steps~1 and~5 and record $F^\T e$ in the data Pauli frame.
\State Measure every ancilla qubit in the $Z$ basis.
\State Perform the next syndrome round of the data block.
\State Decode the ancilla outcomes jointly with the data syndrome and correct the outcomes.
\State \Return the parity of the corrected outcomes on the support of the ancilla logical $Z$ that acquires the target operator.
\end{algorithmic}
\end{algorithm}

A fault is a single-qubit Pauli error on a data or ancilla qubit at any point of Algorithm~\ref{alg:meas}. Check measurements are treated as noiseless, their errors being handled by the metachecks of the two codes. A set of faults is undetected if all check outcomes of steps~1, 5 and~8 agree with the fault-free case, and harmful if it changes the measurement result or leaves a logical error on the data block. The fault distance of the circuit is the minimum size of an undetected harmful set of faults. Let $H_X,H_Z$ and $H'_X,H'_Z$ be the binary check matrices of the data and ancilla blocks, let $z'\in\F^N$ be the support of the ancilla logical $Z$ operator that acquires the target operator, and let $W$ be the span of the data operators acquired by the ancilla logical $Z$ operators. Define
\begin{align}
d_X&=\min\bigl\{\,2|e_d|+|e_a|\;:\;e_d,e_a\in\F^N,\ H'_Z(Fe_d+e_a)=0,\ z'\cdot(Fe_d+e_a)=1\,\bigr\},\label{eq:dX}\\
d_Z&=\min\bigl\{\,|\delta_b|+|\delta_a|+|e|\;:\;\delta_b,\delta_a,e\in\F^N,\ H'_X(\delta_b+\delta_a)=0,\nonumber\\
&\hphantom{=\min\bigl\{\,|\delta_b|+|\delta_a|+|e|\;:\;}F^\T\delta_b+e\in\ker H_X\setminus(\operatorname{rowspace}H_Z+W)\,\bigr\}.\label{eq:dZ}
\end{align}

\begin{prop}\label{prop:fd}
The fault distance of the circuit is $\min(d,d_X,d_Z)$.
\end{prop}

\begin{proof}
Both codes are CSS and the CNOTs map $X$ to $X$ and $Z$ to $Z$, so the $X$-type and $Z$-type parts of a set of faults are detected by disjoint sets of checks and act on the result independently. It suffices to treat the two types separately.

Let $e_d$, $e'_d$ and $e_a$ be the $X$ faults on the data block before the CNOT layers, on the data block after the layers, and on the ancilla block at any time. The data syndrome of step~8 is $H_Z(e_d+e'_d)$ and the readout of step~7 carries the error $Fe_d+e_a$, since $X$ faults on the ancilla block and $X$ faults on the data block after the layers are not copied. The set is undetected if $e_d+e'_d\in\ker H_Z$ and $H'_Z(Fe_d+e_a)=0$, and harmful if $z'\cdot(Fe_d+e_a)=1$ or $e_d+e'_d$ is a nontrivial logical operator. In the latter case the size is at least $d$. Otherwise $e_d+e'_d$ is a data $X$ stabilizer $s$, and $Fs$ is an ancilla $X$ stabilizer by Proposition~\ref{prop:chain}, so replacing $e_d$ by $e'_d$ changes $Fe_d+e_a$ by $Fs$ and preserves both conditions. Hence $(e'_d,e'_d,e_a)$ is undetected and harmful as well, and $|e_d|+|e'_d|+|e_a|\ge2\min(|e_d|,|e'_d|)+|e_a|\ge d_X$. Conversely a minimizer of~\eqref{eq:dX}, placed before and after the layers, is undetected and harmful, so the $X$-type fault distance is $\min(d,d_X)$.

Let $\delta_b$ and $\delta_a$ be the $Z$ faults on the ancilla block before and after the layers but before step~5, and $e$ the $Z$ faults on the data block. $Z$ faults on the ancilla block after step~5 do not affect the readout and are not copied. Step~5 measures $H'_X(\delta_b+\delta_a)$, and the data block carries the error $F^\T\delta_b+e$, since $\delta_a$ is not copied. The set is undetected if $H'_X(\delta_b+\delta_a)=0$ and $F^\T\delta_b+e\in\ker H_X$, and harmful if $F^\T\delta_b+e$ is not in $\operatorname{rowspace}H_Z+W$, because operators in this space act trivially on the post-measurement state. This is the definition of $d_Z$, which is at most $d$ since $\delta_b=\delta_a=0$ and $e$ a minimum-weight logical operator is admissible.
\end{proof}

The minima in~\eqref{eq:dX} and~\eqref{eq:dZ} were estimated with the randomized information-set decoder used for the code distances, with $10^5$ and $3\cdot10^4$ trials respectively, and are upper bounds. For every circuit in Table~\ref{tab:hom_meas} the estimate of $d_Z$ equals $d$, so the fault distance is governed by $d_X$. Table~\ref{tab:hom_meas} lists, for each code, a set of three circuits whose ancilla block has the check weights of the data code and, for $\code{162,6,14}$ and $\code{228,6,17}$, a set with maximal ancilla check weight $10$. For $\code{72,6,8}$ the higher check weight ancilla block gives no improvement.

Since the gadgets only involve an ancillary quadcycle code, the ancilla block can be implemented efficiently using the same atom-movement scheme as in Section~\ref{sec:appendix_compilation_setup}~\cite{SI}. Since the homomorphic $\mathrm{CNOT}$ circuits are determined by $R$-linear maps, we expect that the atom-movements for these circuits are highly regular and can be done with minimal overhead. We note that the logical measurement gadgets we present here are not optimized, and in practice one may choose to use lattice-surgery based methods \cite{cohen2022low, zheng_high-rate_2025, baspin2025fast, cowtan_fast_2025} to implement them more efficiently.

\begin{table}[t]\centering\footnotesize\setlength{\tabcolsep}{3.5pt}
\resizebox{\textwidth}{!}{%
\begin{tabular}{@{}l c l l l c c@{}}\toprule
code & check weight & replaced factor & entries of $A$ & measured & $\mathrm{CNOT}$ depth & fault distance \\\midrule
$\code{54,6,6}$ & 6 & none & $A_{aa}=x+x^2$ & $\bZ_{bc}$ & 2 & 6\\
 & & none & $A_{bb}=y+y^2$ & $\bZ_{ac}$ & 2 & 6\\
 & & none & $A_{cc}=xy^2+x^2y$ & $\bZ_{ab}$ & 2 & 6\\\midrule
$\code{72,6,8}$ & 8 & $c'=xy+x^2y$ & $A_{cc}=1+x$, $A_{bc}=x$ & $\bZ_2\bZ_3$ & 3 & 8\\
 & & $a'=1+x$ & $A_{aa}=1+x$, $A_{ca}=x$ & $\bZ_1$ & 3 & 8\\
 & & $a'=y^3+x^2y^2$ & $A_{aa}=1+x^2y^3$, $A_{ba}=y^3$ & $\bZ_1\bZ_2$ & 3 & 8\\\midrule
$\code{162,6,14}$ & 8 & $c'=y^2+xy$, $d'=1+y^6+x+xy^2$ & $A_{cc}=1+xy^8$, $A_{cd}=y$ & $\bZ_{ab}$ & 3 & 14\\
 & & $a'=1+xy^5$ & $A_{aa}=1+y^3$, $A_{ba}=y^8$, $A_{ca}=xy^3$ & $\bZ_{bc}\bZ_{ac}\bZ_{ab}$ & 4 & 14\\
 & & $b'=1+x$ & $A_{bb}=1+x^2$, $A_{cb}=1$ & $\bZ_{ac}\bZ_{ab}$ & 3 & 13\\
 & 10 & $c'=1+y+y^2+x^2y$ & $A_{cc}=1+y^8$, $A_{bc}=x^2y^8$ & $\bZ_{ac}\bZ_{ab}$ & 3 & 14\\
 & & $a'=y^8+xy^2+x^2y^5+x^2y^7$ & $A_{aa}=1+xy^3$, $A_{ca}=y^7$ & $\bZ_{bc}\bZ_{ab}$ & 3 & 14\\
 & & $c'=y^2+y^4+y^7+x^2y^3$ & $A_{cc}=1+y^6$, $A_{ac}=y^7$, $A_{bc}=x^2y$ & $\bZ_{bc}\bZ_{ac}\bZ_{ab}$ & 4 & 14\\\midrule
$\code{228,6,17}$ & 8 & $c'=y^4+x$ & $A_{cc}=1+y$, $A_{bc}=y$ & $\bZ_2$ & 3 & 17\\
 & & $a'=y^{13}+x$ & $A_{aa}=1+y^8$, $A_{ca}=y^8$ & $\bZ_1$ & 3 & 15\\
 & & $c'=y^5+xy^{13}$ & $A_{cc}=1+y^6$, $A_{ac}=xy^{11}$ & $\bZ_1\bZ_3$ & 3 & 14\\
 & 10 & $c'=y^5+y^9+y^{13}+xy^{15}$ & $A_{cc}=1+y^4$, $A_{ac}=y^{15}$ & $\bZ_1$ & 3 & 17\\
 & & $a'=xy^7+xy^{10}+xy^{13}+xy^{16}$ & $A_{aa}=1+y^3$, $A_{ca}=xy^2$ & $\bZ_1\bZ_3$ & 3 & 17\\
 & & $a'=y^7+x+xy^9+xy^{14}$ & $A_{aa}=1+y^{14}$, $A_{ba}=y^4$ & $\bZ_1\bZ_2$ & 3 & 17\\
\bottomrule\end{tabular}}
\caption{Homomorphic measurement gadgets for $\cczstate$ extraction. Each group of three rows is a set of logical $\bZ$ measurements, each acting on a separate code block. We construct gadgets whose ancilla blocks have either the same maximum check weight as the data block or a maximum check weight larger by $2$; the latter typically yields larger fault distances. Ancilla group-algebra factors not listed equal the data factors. The 'measured' column identifies which three linearly independent $\bZ$ operators are measured: changing logical bases of the logical qubits supporting the \detstate state to these operators turns the state into two disjoint $\cczstate$ states after the logical measurements (up to $\overline{CZ}$ fixups). }

\label{tab:hom_meas}
\end{table}

\section{\texorpdfstring{Teleporting $\overline{CZ}$ corrections}{Teleporting CZ corrections}}\label{supp:teleporting-cz}
\begin{figure}
    \centering
    \includegraphics[width=\linewidth]{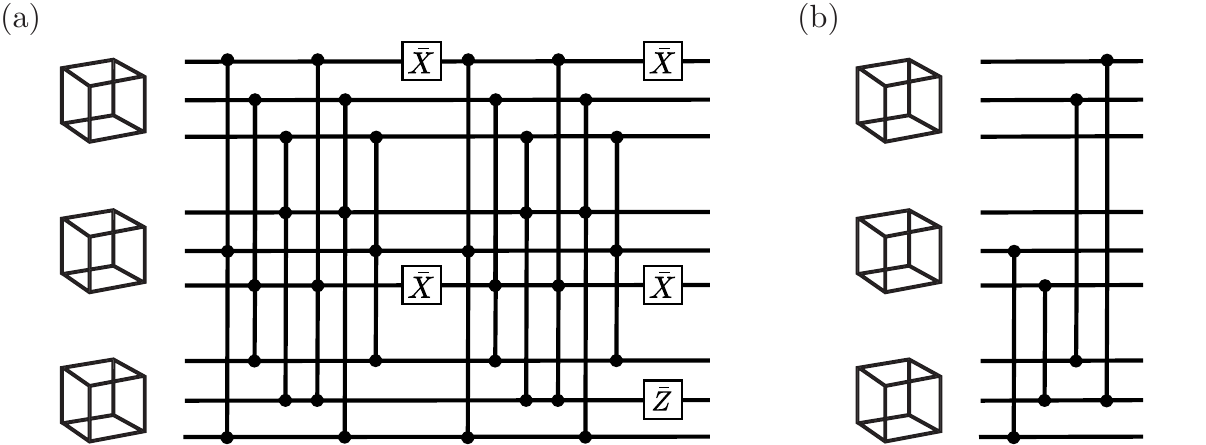}
    \caption{\textbf{Transversal, addressable CZ gates on tricycle codes}. The $k_{\textnormal{3D}} = 3$ logical qubits for each tricycle block are drawn with individual wires. (a) Conjugating logical $\bar X$ operators, e.g., on $\bar q_{11}$ and $\bar{q}_{23}$ with the transversal CCZ circuit results in a product of logical $X$, CZ, and $Z$ gates. Composing these with the same logical $X$ gates and the corresponding $Z$ gates yields (b) a transversal CZ circuit that implements $\overline{\text{CZ}}(\bar q_{22}, \bar q_{33})\overline{\text{CZ}}(\bar q_{23}, \bar q_{32})\overline{\text{CZ}}(\bar q_{12}, \bar q_{31})\overline{\text{CZ}}(\bar q_{11}, \bar q_{32})$. Applying the same method on the physical circuit, we find this transversal CZ is enacted by a depth-2 transversal CZ circuit, with the single-site $Z$ rotations tracked classically. }
    \label{fig:transversal_czs}
\end{figure}

We now demonstrate that the resource states to implement feed-forward Clifford corrections after injecting our magic states can be prepared transversally with our factory. This is also true for any CZ corrections needed after measuring the gauge qubits to extract $\cczstate$ states (Section~\ref{sec:Extracting CCZs}~\cite{SI}). Whether from measuring gauge qubits or in the teleportation stage, CZ corrections arise from a logical qubit being measured in $\ket{\bar 1}$ rather than $\ket{\bar 0}$, where no correction is necessary. Importantly, if a logical qubit is measured in $\ket{\bar 1}$, both of the incident CCZ gates simultaneously contribute a CZ correction, determined by propagating an $X$ logical operator through the CCZ circuit, as shown in Fig.~\ref{fig:transversal_czs}.

Thus, even though the \detstate contains six overlapping CCZs, there are fewer than $\binom{3}{2}^6 = 3^6 = 729$ possible patterns of CZ corrections. Rather, each of the nine logical qubits can be measured in $\ket{\bar 1}$, in which case it contributes a pattern of two CZ gates, with a total of $2^9 =512$ possible corrections. These can be derived from variants of Fig.~\ref{fig:transversal_czs}, in which there are the same $2^9 =512$ ways to apply logical $\bar X$ operators. As both the physical and logical CCZ circuits have depth 2, the physical and logical depths of the resulting circuit can be bounded by at most 4, ensuring that every such CZ circuit is indeed transversal. With a judicious choice of representative $X$ operators the physical depth could be even smaller. For instance, we numerically found that on the $[[81, 3, 5]]$ tricycle codes, any of these transversal CZ circuits can be performed in depth $3$.

\section{Cost breakdown of magic factories}\label{supp:cost-breakdown}
We compute circuit volumes in a model where single-qubit gates, multi-qubit entangling gates (including CNOT, CZ, CCZ, and controlled-$S$), resets to $\ket 0$, and measurements in the $Z$ basis each occupy one layer of circuit depth. Resetting a qubit to $\ket +$ and measuring in the $X$ basis are therefore depth-2 operations, requiring a $Z$ reset followed by a Hadamard gate or a $Z$ measurement preceded by a Hadamard gate. This maintains consistency with circuit depth accounting in Ref.~\cite{sahay2026fold}.

\begin{table}[ht]
    \centering
    \begin{tabular}{lccc}
    \toprule
        Step & \makecell{Peak qubit\\count} & \makecell{Maximum\\depth} & \makecell{Space-time\\volume} \\
    \midrule
    Tricycle $\ket{\overline +}$ preparation $\times 3$ & 162 & 7 & 972 \\
    Quadcycle $\ket{\bar 0}$ preparation $\times 3$ & 270 & 10 & 1890 \\
    Tricycle fan-out gadget $\times 3$ & 162 & 2 & 243 \\
    Tricycle transversal CCZ & 162 & 1 & 162 \\
    Homomorphic CNOT & 243 & 1 & 243 \\
    Transversal tricycle $X$ measurement & 162 & 2 & 324 \\
    \midrule
    Total & & & 3834 \\
    \bottomrule
    \end{tabular}
    \caption{Space-time volume breakdown for the $[[27, 3, 3]] \rightarrow[[54, 6 ,6]]$ factory (small), including all three block pairs. The first two numerical columns give the peak qubit count and maximum depth of each stage considered separately. Their product is not the space-time volume in the final column, given that the peak count does not account for reductions in the number of qubits in use during a stage, and the overall depth of the circuit does not compose linearly from the depth of each stage due to parallelization.}
    \label{tab:breakdown}
\end{table}

Accounting for the fan-out, we include initialization of the additional tricycle qubits and the transversal encoding CNOT. Both copies participate in the transversal CCZ layer and are included in the final tricycle $X$-measurement row. The homomorphic CNOT acts on only half of the quadcycle data qubits, but we include the idling cost from the other half of the qubits. The space-time volume, estimated in terms of the circuit volume, for each stage is listed in Table~\ref{tab:breakdown}. Note that summing the volumes from each component separately results in a total volume estimate that is in fact accurate for the entire factory overall, as we verified these stages can be arranged such that there is no ``inactive'' circuit volume \cite{litinski_active_2022}. 

The space-time volumes of our larger factories scale proportionately with the code sizes, while the circuit depths remain fixed. The only exception to this is the quadcycle preparation circuit, where the circuit depth of measuring the $X$ stabilizers increases from $10$ to $12$, due to the fact that the $[[54, 6, 6]]$ is a 2-2-2-2 quadcycle code with check weight $6$, whereas the other quadcycle codes in Fig.~\ref{fig:dimjump} are all 4-2-2-2 codes with maximum check weight $8$. To obtain the expected space-time volume estimates per \detstate, which are halved to obtain the results in Table~\ref{tab:comparison}, we divide the total space-time volume for the factory by $\alpha^3$, where $\alpha$ is the acceptance fraction of a single block.  This estimate could be even a bit conservative due to error correlations between blocks, e.g., correlations from the three-qubit depolarizing channel following the transversal CCZ gate.

\section{Comparison to 2D-3D code-switching protocols}\label{sec: 2D-3D comparison} 

The technique of dimension jumping between codes of different algebraic or geometric dimensions has been explored in past works. Ref. \cite{raussendorf2005long} first introduced the notion of collapsing a 3D cluster state into a 2D cluster state. 
Later works extended this method to contract 3D topological codes to 2D topological codes and introduced ``expansion'' mechanisms based on gauge-fixing to grow from 2D codes into 3D codes \cite{paetznick2013universal, anderson2014fault, kubica2015universal, bombin2016dimensional, watson2015qudit, butt2024fault}. An alternative approach to dimension jumping based on transversal teleportation emerged as a compelling candidate since it allowed the state preparation of the target code to be performed offline. In retrospect, this transversal dimension jump can be understood in terms of a homomorphic CNOT \cite{huang_homomorphic_2023} that arises naturally between a three-dimensional topological code and a two-dimensional projection of it \cite{heussen2024efficient}. However, owing to the low encoding rates of topological codes, Ref. \cite{beverland2021cost} reported that code-switching between 2D and 3D topological codes does not offer substantial savings over complementary approaches to generating magic fault-tolerantly, calling into question the benefits of dimension jumping.

This line of work was in part reinvigorated by the advent of high-rate qLDPC codes. Ref. \cite{hong2024single} uses a generalization of dimensional contraction to prepare 2D hypergraph product (HGP) codes in constant depth by first initializing to a 3D HGP code. Ref. \cite{bergamaschi2025constant} applies the same mechanism to arbitrary CSS codes. Ref. \cite{golowich2025constant} considered switching between $(r-1)$- and $r$-dimensional expander HGP codes for any $r \geq 3$ by using contraction and transversal teleportation to respectively decrease and increase dimension, enabling single-shot, addressable Clifford gates. In a similar vein, Ref. \cite{tan2025single} achieves single-shot, addressable universality. There, the authors use transversal teleportation for their code-switching protocols in conjunction with contraction \cite{bergamaschi2025constant} to prepare their 2D codes in constant depth. Finally, Ref. \cite{li_transversal_2025} extends the homomorphic CNOT relating 2D/3D topological and hypergraph product codes to lifted product codes, including the subfamily of balanced product codes we consider in this work. We remark that most of these qLDPC constructions focus on asymptotic overheads and do not obviously yield advantageous protocols at finite sizes.

An exception is Ref. \cite{li_transversal_2025}, which considers magic-state generation with the same family of tricycle codes we consider in this work. However, the authors there benchmark---just as in Ref. \cite{menon_magic_2025}---immediately measuring the tricycle code transversally and applying full postselection. Of course, such a protocol destroys the magic state in the process and therefore is unrealistic. They also consider teleporting the tricycle magic state to bicycle codes, yielding a 3D-to-2D switching protocol. We remark on a few key distinctions from our protocol that arise as a consequence of that choice. First, the bicycle codes they pair with these tricycle codes have $k_{\textnormal{2D}} = 2 < k_{\textnormal{3D}} = 3$ (Ref. \cite{li_transversal_2025}, Table I). As a result, the chain map that defines the appropriate homomorphic CNOT is not logically surjective. Consequently, one tricycle code's worth of logical qubits must be teleported across two or more bicycle codes. Moreover, as the rank of the logical map is uniformly $2$ (or 0), there are unavoidably some logical qubits of the tricycle code that are distributed across multiple bicycle code blocks. The resulting magic state becomes \detstate dressed with CNOTs.

Second, unlike our quadcycle codes, bicycle codes are not single-shot state preparable. In order to make the transversal teleportation fault-tolerant, typically $d$ rounds of error correction are needed on the bicycle code. This does not contradict our earlier assumptions. There, the fact that the teleportation protocol remains fault tolerant, even though the quadcycle code can be prepared in a non-fault tolerant way, is a consequence of our logical chain map being injective, which makes it compatible with algorithmic fault-tolerance \cite{zhou_low-overhead_2025}. Thus, use of an auxiliary bicycle code initialized in $\ket{\bar 0}$ requires $O(d)$ rounds of error correction, significantly inflating the space-time volume of the protocol. One potential workaround for this issue is if the bicycle code block is not freshly prepared; rather, it stores the logical qubits on which we want to execute the magic gate. However, in this setting, we can no longer postselect on the transversal measurement of the tricycle code.

Regardless of whether the bicycle code is a logical ancilla block or not, the most formidable limitation of the tricycle-to-bicycle teleportation protocols is that the bicycle codes tend to have distances comparable to that of the tricycle codes. With the magic state unable to rapidly escape to a higher-distance code, even if postselection on the transversal measurement of the tricycle codes is performed, its benefits are nullified by the low distance of the bicycle code. As such, the postselection advantages reported in Refs. \cite{li_transversal_2025, menon_magic_2025} and studied in this work are lost in that setting. 

Finally, we remark that the ability to explore code-switching between 3D and 4D codes at practical blocklengths is due largely to the balanced product construction. While ideas about switching between 2D and 3D codes were generally expected to carry over to the 3D–4D case, the practicality of 4D codes based on topological and HGP constructions has remained in question. For instance, for HGP codes constructed as the four-fold product of one seed $[n, k]$ classical code of rate $R = k/n$, the blocklength scales rapidly as $\sim n^4$, and the rate drops off quickly as $\sim R^4$. For 4D balanced product codes, in contrast, the blocklength scales as $\sim 4n$, and the rates are observed to be more favorable.

\section{Direct code-growth from 3D to 4D}\label{sec:code_growth}

In the main text we presented a code-switching protocol between the tricycle and quadcycle codes in the form of a transversal homomorphic $\mathrm{CNOT}$ between appropriately chosen code-pairs, enabling an efficient state-teleportation mechanism. It is instead possible to directly grow the tricycle codes into the paired quadcycle code by simply measuring the stabilizers of the quadcycle code. This is possible because the same algebraic structure that gives rise to the homomorphic $\mathrm{CNOT}$ also can be used to describe a code growth procedure. As discussed in Section~\ref{sec:DJ CNOT}~\cite{SI}, this homomorphic CNOT is described by a chain inclusion $\iota: T_\bullet \rightarrow Q_\bullet$ from the tricycle code chain complex $T_\bullet$ to the corresponding quadcycle code chain complex $Q_\bullet$. By initializing the $n_{\textnormal{3D}}$ additional physical qubits of the quadcycle code in $\ket{+}$, which correspond to the standard basis of $Q_1/\im \iota_1$, any tricycle code $X$ stabilizer extends into a quadcycle code stabilizer with, in the absence of noise, the same measurement eigenvalue, which allows us to form $X$ detectors. A basis of linearly independent $Z$ stabilizers of the quadcycle code that has support on these additional qubits initialized in $\ket +$ will have locally randomized measurement outcomes. However, any product of $Z$ checks whose support vanishes on these added qubits should, barring noise, be a $Z$ stabilizer of the tricycle code and should yield a consistent measurement outcome. This enables $Z$ detectors. In fact, directly measuring the quadcycle stabilizers on top of the tricycle code after introducing additional qubits is equivalent to logical teleportation with the homomorphic CNOT up to circuit rewrite identities in the language of ZX diagrams \cite{coecke2011interacting}. In our simulations, we found direct stabilizer growth to have just a slightly lower acceptance fraction and higher logical error rate than transversal teleportation, due to the presence of some additional idling locations where an idling error could occur. The entire protocol, together with the fan-out gadgets for reducing the depth of the transversal $\mathrm{CCZ}$ circuit to $1$ (see Section~\ref{app:concat}~\cite{SI}), is summarized in \cref{fig:code-growth}.

\begin{figure}
    \centering
    \includegraphics[width=1.0\linewidth]{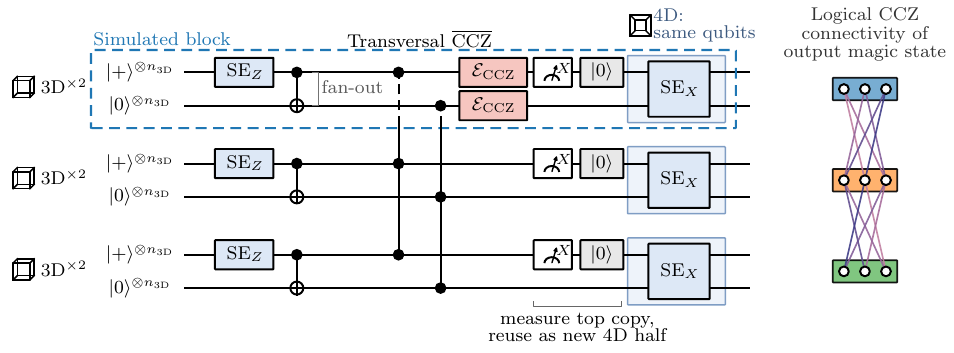}
    \caption{Direct code-growth protocol for tricycle to quadcycle escape. Three sets of two tricycle code blocks admit a depth-$1$ transversal $\mathrm{CCZ}$ circuit, using the fan-out gadget described in the main text. Additional measurements of $X$ stabilizers of the quadcycle code prepare the code-space of the quadcycle code in constant depth while preserving logical information.}
    \label{fig:code-growth}
\end{figure}

\section{Neutral atom array architecture overview}\label{sec:appendix_compilation_setup}

In this work we assume a single-zone architecture where Rydberg interactions can be performed throughout the zone. The SLM traps are initially defined as a list of coordinates relative to the bottom left corner of the zone. Each trap holds at most one atom at a time. The traps are positioned to allow Rydberg interactions between the atoms involved in an entangling gate at each time step, while keeping all other atoms far enough away to prevent unwanted interactions. The SLM array is assumed to be constant throughout each experiment. Moreover, we assume that the SLM trap lattice can have an arbitrary arrangement in the plane (e.g., rectangular, triangular, or a combination of the two), which is achievable in current hardware.

AODs create a grid of moving traps that pick up atoms from the SLM traps, transport them, and deposit them in their destination SLM traps. The AOD traps themselves are located at the intersection points of the column and row coordinates that are turned on. During one AOD movement, the row and column tones can be moved independently with the condition that the relative order of the tones does not change.

Formally, we model an AOD lattice as a series of $x,y$ coordinates (tones) $x_1 < x_2 < ... < x_p$ and $y_1 < y_2 < ... < y_q$ which creates a total of $pq$ traps located at coordinates $\{(x_i,y_j) \mid i \in [p], \ j \in [q] \}$. During a move, the tones are moved to the final positions $\overline{x_1} < \overline{x_2} < ... < \overline{x_p}$ and $\overline{y_1} < \overline{y_2} < ... < \overline{y_q}$, where each tone has the same relative position to the other tones, meaning that they transform according to $x_i \rightarrow \overline{x_i}$ and $y_j \rightarrow \overline{y_j}$. As a result, a trap that is initially at $(x_i,y_j)$ will end at $(\overline{x_i}, \overline{y_j})$ which allows us to define a displacement vector $\vec{v}_{i,j} = (\overline{x_i} - x_i, \overline{y_j} - y_j)$.
We use this vector notation to describe atom movements in the following sections.

\textbf{4D patches.}
Each of the three $[[54,6,6]]$ patches is a triangular grid of $21 \times 6$ doublets: for column $k \in \{0,\dots,20\}$ and row $d \in \{0,\dots,5\}$ the doublet center of copy $m$ is
\begin{equation}
  X = x_0 + p_x\!\left(k + \tfrac{1}{2}(d \bmod 2)\right),
  \qquad
  Y = Y_{4}^{(m)} + p_4\, d,
  \qquad
  p_4 = \sqrt{84} \approx 9.165 .
  \label{eq:traps-4d}
\end{equation}
Odd rows are shifted right by $p_x/2$; the row pitch $p_4$ is fixed by the safety rule $\sqrt{p_4^2 + (p_x/2 - 2)^2} = 10$, so that even with both rails of a doublet occupied during a gate, diagonal neighbors stay $10\,\mathrm{\mu m}$ apart. The unit cell $(i,j) \leftrightarrow \alpha = x^i y^j$ of the code occupies columns $k = 7i + c$, $c \in \{0,\dots,6\}$, and rows $d = 2j$ (lower) and $2j{+}1$ (upper), with the sector assignment
\begin{align*}
  \text{lower row } (d = 2j) &:\quad X_3\; Q_1\; Q_2\; Z_1\; X_1\; X_2\; Q_3\,,\\
  \text{upper row } (d = 2j{+}1) &:\quad Q_6\; Z_2\; Z_3\; X_4\; Q_4\; Q_5\; Z_4\,.
\end{align*}

\textbf{3D patches.}
Each $[[27,3,3]]$ patch is a rectangular grid of $18 \times 3$ doublets:
\begin{equation}
  X = x_0 + p_x\!\left(k + \tfrac{1}{2}\right),\quad k \in \{0,\dots,17\},
  \qquad
  Y = Y_{3}^{(m)} + p_3^{(m)}\, j,\quad j \in \{0,1,2\}.
  \label{eq:traps-3d}
\end{equation}
The constant $\tfrac{1}{2}$ column offset aligns every 3D doublet with the upper (odd-$d$) rows of the 4D patches, which is what the transversal 4D--3D CNOT relies on. The row pitch is $p_3^{(m)} = 10$ for the outer copies and $14$ for the middle copy ($m = 2$), which needs the extra room for CCZ blockade isolation. Cell $i$ occupies columns $k = 6i + c$ with the sector assignment $Q_3\; Z_1\; Z_2\; Q_1\; Q_2\; Z_3$ for $c = 0,\dots,5$ (the $X$ checks are dropped entirely).

\textbf{CCZ apex traps.}
Every doublet of the middle 3D copy carries a third, bare trap at
\begin{equation}
  \left(X,\; Y - \sqrt{3}\,\right),
  \label{eq:traps-apex}
\end{equation}
completing an equilateral triangle of side $2\,\mathrm{\mu m}$ with the doublet's two rails --- the assembly site of one CCZ triple (data doublets host layer~1, twin doublets layer~2; cf.\ Section~\ref{sec:appendix_compilation_ccz}~\cite{SI}).

\begin{table}[t]
  \centering
  \begin{tabular}{lccc}
    \toprule
    band & base ordinate $Y$ & row pitch & rows\\
    \midrule
    \multicolumn{4}{l}{\emph{off-array (below)}}\\
    measurement zone, copy-$3$ sub-band & $-114$ & $10$ & $3$\\
    measurement zone, copy-$2$ sub-band & $-82$ & $10$ & $3$\\
    measurement zone, copy-$1$ sub-band & $-50$ & $10$ & $3$\\
    \midrule
    \multicolumn{4}{l}{\emph{on-array}}\\
    3D copy $3$ & $10$ & $10$ & $3$\\
    3D copy $2$ & $42$ & $14$ & $3$\\
    3D copy $1$ & $82$ & $10$ & $3$\\
    4D copy $3$ & $118$ & $\sqrt{84}$ & $6$\\
    4D copy $2$ & $130 + 5\sqrt{84} \approx 175.83$ & $\sqrt{84}$ & $6$\\
    4D copy $1$ & $142 + 10\sqrt{84} \approx 233.65$ & $\sqrt{84}$ & $6$\\
    \midrule
    \multicolumn{4}{l}{\emph{off-array (above)}}\\
    fresh-qubit reservoir & \multicolumn{3}{c}{the $108$ $Z$-check doublets, shifted up by $190$}\\
    \bottomrule
  \end{tabular}
  \caption{Band ordinates of the trap layout (all in $\mathrm{\mu m}$; bands are stacked bottom-to-top as listed). The bases realize a $12\,\mathrm{\mu m}$ gap between adjacent patches, $16\,\mathrm{\mu m}$ between the 3D and the 4D group, $40\,\mathrm{\mu m}$ between the topmost zone row and the bottom 3D row, and $12\,\mathrm{\mu m}$ between zone sub-bands.}
  \label{tab:trap-bands}
\end{table}

\textbf{Initial occupancy and totals.}
In the initial layout the 3D patches are fully occupied, while the 4D atoms sit on the \emph{compact} sub-grid of columns $k = 6i + c'$ with cell assignment $X_3\,Q_1\,Q_2\,X_1\,X_2\,Q_3$ (lower row) and $Q_6\,Z_2\,X_4\,Q_4\,Q_5\,Z_4$ (upper row), $c' = 0,\dots,5$, with every $Z$-check position on the array empty. The compact columns are a subset of the grid of Eq.~\eqref{eq:traps-4d}; the remaining columns are the spare traps into which the patches re-expand to the full $2 \times 7$ cell in the last stage. The $108$ atoms that become those $Z$ checks are not absent, only \emph{off-array}: they wait in a reservoir of $108$ doublets obtained by translating the $Z$-check doublets up by $190\,\mathrm{\mu m}$, and descend into place by that single uniform translation while the patches expand beneath them.
The on-array layout therefore contains $756$ 4D traps, $324$ 3D traps and $54$ apex traps --- $\mathbf{1134}$ distinct traps inside $[9, 257] \times [10, 279.5]\,\mathrm{\mu m}$, i.e.\ $248 \times 269.5\,\mathrm{\mu m}$. The two off-array staging regions add $162$ readout-zone traps below and $216$ reservoir traps above, for $1512$ in total. The formulas above specify the complete trap layout.

\section{Implementation of syndrome extraction}\label{sec:appendix_compilation_syndrome}

The quasi-cyclic and product structure of lifted product codes enables parallel atom moves that naturally facilitate efficient implementation of these codes on neutral atom arrays.

Both the tricycle and quadcycle codes can be understood as generalizations of 2D Quasi-Abelian Lifted Product codes which, as shown by the authors in a previous work \cite{tan26syndrome}, can be scheduled for full syndrome extraction with near-optimal depth per round.
In this work, the magic state factory protocol in Fig.~\ref{fig:dimjump} requires only partial syndrome extraction: the $X$ checks for the quadcycle code and the $Z$ checks for the tricycle codes.
As such, we can simplify the circuit by performing the CNOTs between all qubit pairs specified by each term of the component circulant matrices $A, B, C, D$.
The circuit is shown below. Its entangling-gate depth is minimal and equals the maximum weight of the checks being measured.

\begin{figure}[ht!]
    \centering
    \includegraphics[width=1.0\textwidth]{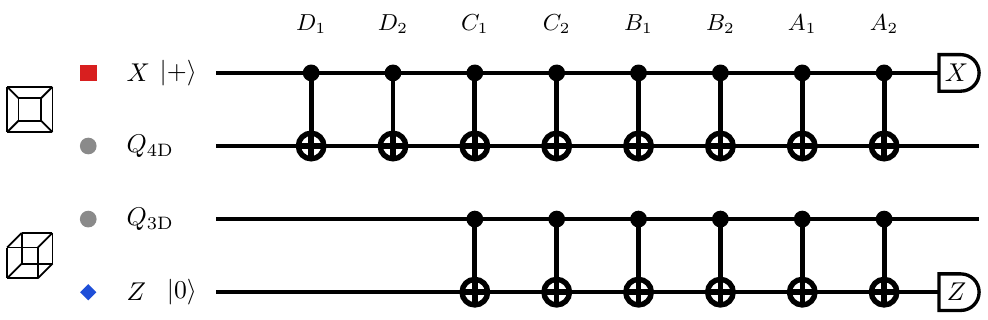}
    \caption{Schematic circuit diagram for the partial syndrome extraction of the tricycle and quadcycle codes from \cref{fig:dimjump}. The qubit lines represent multiple qubits and the CNOT gates represent parallel CNOTs between all qubits that are paired by that respective polynomial term. Note that the tricycle code idles when the $D$ terms are being performed on the quadcycle code.}
\end{figure}

The above circuit is naturally suited for neutral atoms. 
At each stage the check qubits are moved by some term $P_i$ where $P \in \{A, B, C, D \}$ and $i \in [W_P]$ where $W_P$ is the number of terms in the generating polynomial for circulant matrix $P$.
We can model the move of some qubit from sector $\mathcal{S}$ of unit cell $\mathcal{C}$ as an addition between two vectors:
\begin{equation}
    \vec{v}_\mathrm{cyclic} (\mathcal{C}, P_i) + \vec{v}_\mathrm{in-cell} (\mathcal{S}, P)
\end{equation}

Note here the dependencies of the two vectors. 
The cyclic-shift vector depends only on the starting unit cell $\mathcal{C}$ and the monomial $P_i$; it points from one unit cell to another.
The in-cell shift only depends on what circulant matrix $P$ is currently being implemented and which sector $\mathcal{S}$ the qubit belongs to: this is decided by the construction of the full parity-check matrices in Eq.~\eqref{eq: quadcycle checks}.
For the identity polynomial term, $1$, there is no cyclic shift: $\vec{v}_\mathrm{cyclic} = \vec{0}$. All moves then occur within a cell, between sectors paired by the circulant matrix $P$.

We define a set of qubits that move with the same $\vec{v}_\mathrm{cyclic}$ as a \textit{toric wrap class}. For the identity term, there is one toric wrap class, with $\vec{v}_\mathrm{cyclic} = \vec{0}$. For univariate terms, which are powers of a single variable (e.g., $x$, $y^2$), there are two classes: one that moves in the indicated direction, and another that wraps around one of the cyclic boundaries. Finally, for bivariate terms (e.g., $xy$, $x^2 y$), there are four classes: no wrap, wrap around $x$, wrap around $y$, wrap around both $x$ and $y$.

The split between $\vec{v}_\mathrm{cyclic}$ and $\vec{v}_\mathrm{in-cell}$ allows us to reason about which moves can be parallelized using AODs.
The unit cell of the 4D code has been designed to allow for every circulant matrix to move the paired qubits with a single AOD, as shown in \cref{fig: unit cell pairings}.
\begin{figure}[ht!]

    \centering
    \includegraphics[width=1.0\textwidth]{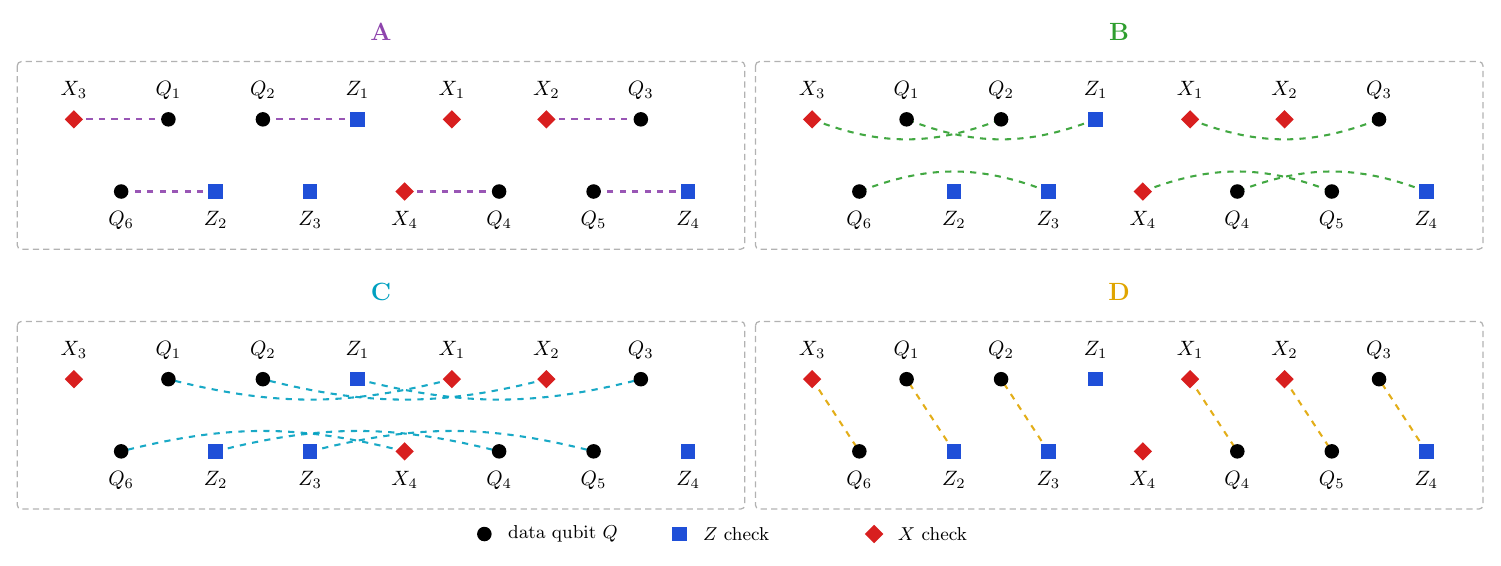}
    \caption{Atom pairings within a quadcycle unit cell for the four component circulant matrices. All pairings for $A$, $B$, and $C$ lie within a single row; $D$, which extends the construction to 4D, is the only matrix that pairs qubits in different rows. The pairings shown correspond to the identity polynomial term, $1$, with no cyclic shift.
    \label{fig: unit cell pairings}}
\end{figure}
One AOD can therefore move the paired atoms within every cell, allowing an identity term to be implemented with a single AOD move for the entire block.
Terms with two or four toric wrap classes require two or four AODs, respectively, to move all classes simultaneously. These cases are illustrated in \cref{fig:compilation} and discussed in Section~\ref{sec:appendix_compilation_misc}~\cite{SI}.

\begin{figure}[ht!]

    \centering
    \includegraphics[width=1.0\textwidth]{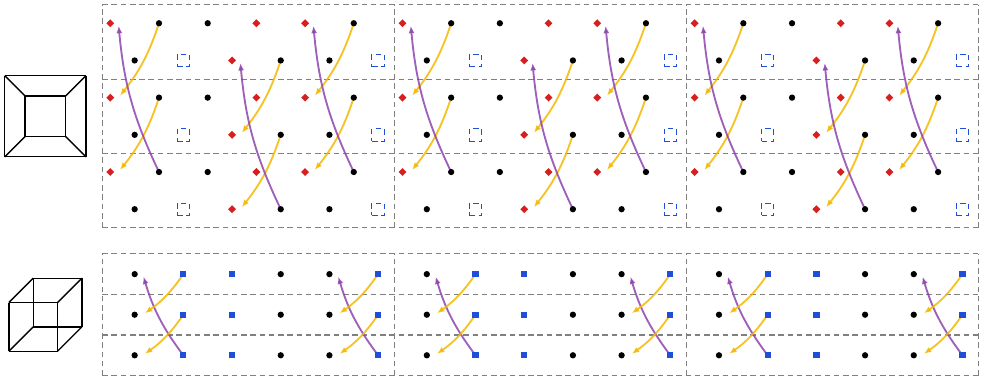}
    \caption{
    A quadcycle and a tricycle code patch performing the $A_2 = y$ term with two shared AODs. The first AOD (yellow) performs a cyclic shift of $y$ while the second AOD (purple) performs the wrap-around move. Note that the unit cells of both code patches are compressed because only one type of check is measured. The tricycle code's unit cell is reduced to 6 qubits without the $X$ checks, while the quadcycle has been compressed to a $2\times 6$ cell by consolidating the $Z_1,Z_3$ checks; moreover, the other $Z$ checks are also missing and are represented by the hollow blue squares. For visualization purposes, only 2 unit cells per code are shown to be moving; in reality, the same AODs would pick up the appropriate qubits from every unit cell of the diagram.
    \label{fig:coupled_AOD}
    }
\end{figure}

Since the quadcycle and tricycle codes' unit cells share a common row, the syndrome extraction for both types of codes can be done in parallel without extra AOD moves.
By placing the tricycle code's unit cells on the same rows as the quadcycle code's second row we can guarantee that the moves corresponding to the $A, B, C$ terms will have the same $\vec{v}_{polynomial}$ while the $\vec{v}_{cyclic}$ of the tricycle code will be analogous to that of the quadcycle; as such, the moves on these three polynomials could be done in parallel by sharing the AODs between the two codes.
Because $X$ and $Z$ checks are extracted on the quadcycle and tricycle codes, respectively, and $H_X$ uses transposed circulant matrices whereas $H_Z$ does not, sharing AODs requires moving different types of qubits in the two codes: one moves the check qubits, while the other moves the data qubits.
Since the tricycle codes have lower distance, and since moving atoms can lead to more errors, we chose to move the data qubits of the quadcycle code and the check qubits of the tricycle code.

\section{Implementation of transversal CCZ}\label{sec:appendix_compilation_ccz}

The unit cell layout of \cref{fig:coupled_AOD} (from left to right: $Q_3\,Z_1\,Z_2\,Q_1\,Q_2\,Z_3$) allows for highly parallel implementation of the transversal CCZ. Note that the $X$-checks are dropped since the magic state preparation measures only the $Z$ checks. Throughout this section we write $Q_s^{(m)}$ for the data qubits in sector $Q_s$ of copy $m$, where $s,m\in\{1,2,3\}$.
Further, to identify the representative qubit of some unit cell $(i,j)$, we write $Q_s^{(m)}(\alpha)$ where $\alpha = x^{i} y^{j} \in G = \mathbb{Z}_3 \times \mathbb{Z}_3$.
Similarly, we use $Z_s^{(m)}(\alpha)$ for check qubits.
For the realization of the CCZ, only the data sectors and their repetition-code twins (introduced below) move; copy~$2$ never moves.

The three copies are stacked in three horizontal bands, column-aligned, with copy~$1$ on top, copy~$2$ in the middle, and copy~$3$ on the bottom. 
Group elements are $\alpha=x^{i}y^{j}$; the cell of $\alpha$ sits at horizontal index $i$. 
 The scheme has depth two \cite{jacob2025singleshot}: every physical qubit participates in exactly two of the $54$ CCZ gates, so they split into two layers of $27$ disjoint gates.
For each layer, we list the three orderings of sectors that participate in CCZ gates, with tuple entries corresponding to $(\text{copy }1,\text{copy }2,\text{copy }3)$:
\begin{align*}
  \text{layer 1} &: (Q_1Q_2Q_3),\ (Q_2Q_3Q_1),\ (Q_3Q_1Q_2)
                    && \text{(the even permutation orderings)},\\
  \text{layer 2} &: (Q_1Q_3Q_2),\ (Q_2Q_1Q_3),\ (Q_3Q_2Q_1)
                    && \text{(the odd permutation orderings)}.
\end{align*}
Realized as two sequential stages, the two layers would cost two Rydberg flashes, since they reuse the same $81$ atoms. 

However, the concatenation from Section~\ref{app:concat}~\cite{SI} allows the implementation of both layers in a single flash.
Moreover, the concatenation itself does not require any extra qubits as we can recycle the $Z$ ancillae. Alternatively, we can replace them with fresh atoms in order to minimize loss errors due to heating.
After the $Z$ checks are read out (the ancillae fly to a measurement zone below the array for the readout), they are replaced with fresh atoms (in state $|0\rangle$) that have been assumed to be idling in the measurement zone.
Each is then paired with the nearest data qubit of its own patch --- one lattice site over, $Z_{s+1}$ next to $Q_s$ (indices mod $3$) --- and a transversal CNOT turns each pair into a distance-$2$ repetition encoding of the data qubit. Every data qubit $Q_s^{(m)}(\alpha)$ now owns a \emph{twin} $\tilde{Q}_s^{(m)}(\alpha)$ carrying the same information, and since $\mathrm{CCZ}$ is diagonal, a gate may act on either member of a pair. We let layer~1 act on the original data qubits and layer~2 on the twins: the two layers then occupy disjoint atoms. The whole concatenation costs three transport jobs (each a single $81$-atom AOD shot: to the zone, to the CNOT partners, back to the ancilla traps), one readout pulse and one CNOT flash.
In both layers, copy~$2$ is a fixed anchor: it supplies the middle qubit of every gate and never moves --- layer~1 anchors on its data doublets, layer~2 on its twin doublets.
To assemble the qubits for one CCZ gate we bring the other two copies to the anchor, so that the qubit triplets participating in a physical CCZ gate are placed in equilateral triangles with a separation of $2\,\mathrm{\mu m}$.
Thus, with a single Rydberg pulse the physical CCZ is implemented on all $54$ triplets of qubits at once.
This protocol uses at its peak a total of $\mathbf{432}$ atoms: $270$ for the quadcycle patches (unit cell of $10$ qubits) and $162$ for the tricycle patches (unit cells of $6$) in accordance with the missing check qubits from \cref{fig:coupled_AOD}.

We define an ordering by a permutation $\pi=(\pi_1,\pi_2,\pi_3)$ of $(1,2,3)$, so that copy~$m$ contributes its sector $Q_{\pi_m}$.
The gate family over $\alpha\in G$ is
\begin{equation}
  \mathrm{CCZ}\!\left(
    Q_{\pi_1}^{(1)}(\alpha),\;
    Q_{\pi_2}^{(2)}\!\big(s_2^{\pi}\,\alpha\big),\;
    Q_{\pi_3}^{(3)}\!\big(s_3^{\pi}\,\alpha\big)
  \right),
  \label{eq:ccz-family}
\end{equation}
where $s_2^\pi, s_3^\pi$ are cyclic shifts over $G$. For the layer-2 orderings the gate acts on the twins, $\tilde{Q}$ in place of $Q$ throughout; a twin sits one lattice site from its partner, so the cell-level translations below are unchanged. 
The cyclic shifts $s_2,s_3$ can be determined from the cup product on the balanced product code whose construction is $G$-invariant, and therefore can be represented by quasi-cyclic sparse matrices~\cite{jacob2025singleshot, menon_magic_2025, breuckmann_cups_2025}.

Because copy~$2$ holds the target cell $s_2\alpha$, the two moving copies must translate onto it:
\begin{itemize}
  \item copy~$1$ carries $Q_{\pi_1}^{(1)}$ from cell $\alpha$ to cell
        $s_2\alpha$, so its translation is exactly $s_2$;
  \item copy~$3$ carries $Q_{\pi_3}^{(3)}$ from cell $s_3\alpha$ to cell
        $s_2\alpha$, so its translation is $s_2 s_3^{-1}$.
\end{itemize}
Two structural facts make these moves highly parallel:
\begin{enumerate}
  \item \emph{The product over $\alpha$ is free.} All nine group elements of a data sector share one cyclic translation, so a single AOD move shifts an entire nine-atom sector in one go.
  \item \emph{The flash is maximally parallel.} One Rydberg pulse realizes all $54$ $\mathrm{CCZ}$s at once ($6$ orderings $\times\ 9$ group elements --- both layers, accounting for the twins).
\end{enumerate}
The only operations that cost AOD moves are therefore the sector translations
$s_2$ (copy~$1$) and $s_2 s_3^{-1}$ (copy~$3$), collected per ordering in
Table~\ref{tab:ccz-shifts}.

\begin{table}[t]
  \centering
  \begin{tabular}{llcc}
    \toprule
    layer & ordering $(\pi_1\pi_2\pi_3)$ & copy-$1$ shift $s_2$ & copy-$3$ shift $s_2 s_3^{-1}$\\
    \midrule
    \multirow{3}{*}{1}
      & $(Q_1Q_2Q_3)$ & $y$      & $x^{2}$    \\
      & $(Q_2Q_3Q_1)$ & $x$      & $x^{2}y$   \\
      & $(Q_3Q_1Q_2)$ & $xy^{2}$ & $y^{2}$    \\
    \midrule
    \multirow{3}{*}{2}
      & $(Q_1Q_3Q_2)$ & $x$          & $\mathbf{1}$ \\
      & $(Q_2Q_1Q_3)$ & $x$          & $x^{2}$      \\
      & $(Q_3Q_2Q_1)$ & $\mathbf{1}$ & $x^{2}$      \\
    \bottomrule
  \end{tabular}
  \caption{The two sector translations for each ordering. Copy~$1$ shifts by $s_2$;
  copy~$3$ shifts by $s_2 s_3^{-1}$. A shift moving one coordinate of
  $\mathbb{Z}_3\times\mathbb{Z}_3$ costs two toric wrap pieces, one moving both
  coordinates costs four, and the identity $\textbf{1}$ costs one.}
  \label{tab:ccz-shifts}
\end{table}

Similar to the syndrome extraction moves, the move count is decided by the number of toric wrap classes. Assigning one class to each identity term, two to each univariate term, and four to each bivariate term in Table~\ref{tab:ccz-shifts} gives the naive count of single AOD translations per layer:

\begin{itemize}
  \item \textbf{layer 1}: $2+2+4=8$ for copy~1, $2+4+2=8$ for copy~3, total $\mathbf{16}$;
  \item \textbf{layer 2}: $2+2+1=5$ for copy~1, $1+2+2=5$ for copy~3, total $\mathbf{10}$.
\end{itemize}
The $26$ single-AOD translations would place all moving atoms from both layers before the single Rydberg pulse. We can further merge these rigid translations. The AOD model of Section~\ref{sec:appendix_compilation_setup}~\cite{SI} allows a single translation to apply a non-uniform but order-preserving separable map (a per-column $\Delta x$ and a per-row $\Delta y$), not merely a rigid translation.
Pieces from different blocks of the same copy and also from different layers may therefore share one shot whenever their combined map still preserves the $x$- and $y$-tone orders and picks up no stray
atom.
Greedily fusing compatible pieces against the dense unit-cell layout reduces the $26$ rigid translations to $20$ placement jobs; the remaining fusions are forbidden only because their tone grids would sweep up movers that have not departed yet. 
We therefore batch the placement in stages. The target fusions are computed as if every pending mover had already left and the batches are ordered so that the blocking movers depart first, re-fusing the remainder against the live layout after every batch. The placement then reaches the same tone-order floor for the return moves: $\mathbf{16}$ jobs each way, for a total of $\mathbf{32}$ AOD jobs.

After merging the moves, we batch them onto the available AODs, one job per AOD per batch.
With 2 AODs a batch runs two jobs at once (typically copy-$1$ movers alongside copy-$3$ movers), giving $16$ move batches; with 4 AODs a batch additionally overlaps both wrap fragments of a shift, halving this to $8$ move batches. The merged job count is the same in both cases---merging is a property of the AOD tone map, not of how many AODs are available; the AOD count
only sets how many jobs run in the same instant. Table~\ref{tab:ccz-budget}
collects the totals.

\begin{table}[t]
  \centering
  \begin{tabular}{lc}
    \toprule
     & count\\
    \midrule
    naive rigid translations (layer $1$ + layer $2$) & $16+10=26$\\
    \midrule
    merged placement jobs (way in)     & $16$\\
    merged return jobs (way out)       & $16$\\
    AOD jobs (place $+$ return)        & $\mathbf{32}$\\
    \midrule
    move batches (2 AODs)              & $8+8=\mathbf{16}$\\
    move batches (4 AODs)              & $4+4=\mathbf{8}$\\
    \bottomrule
  \end{tabular}
  \caption{Move budget for the transversal $\overline{\mathrm{CCZ}}$ on three
  $[[27,3,3]]$ copies, after the distance-$2$ concatenation. The fused jobs include both layers, which are placed together ahead of the single flash, thanks to the concatenation. The only
  unavoidable serialization is given by
  place $\to$ flash $\to$ return.}
  \label{tab:ccz-budget}
\end{table}
Hence, the transversal $\mathrm{CCZ}$, or equivalently the six logical
$\overline{\mathrm{CCZ}}$s given by
$$\prod_{\sigma\in S_3}\overline{\mathrm{CCZ}}
(\text{copy}_1{:}\,\sigma(1),\,\text{copy}_2{:}\,\sigma(2),\,\text{copy}_3{:}\,\sigma(3))$$
are implemented with a single Rydberg flash and $32$ AOD transport jobs,
which are grouped into $16$ move batches with $2$ AODs or $8$ move
batches with $4$ AODs, on top of the concatenation's three $81$-atom transport
jobs, one readout pulse and one CNOT flash.

\section{Trade-off and generalization of implementation}\label{sec:appendix_compilation_misc}
As discussed in the preceding sections, the choice of code affects the cost of atom placement and routing.
If the generating polynomials have fewer bivariate monomials or equivalently more identity and univariate monomials, then the number of toric wrap classes decreases, yielding faster syndrome extraction circuit implementations and ultimately faster cycle times.

Formally, let the polynomial in question be $A = \sum_i^{T} x^{p_i} y^{q_i}$ where $T$ is the number of monomials in the polynomial, each term being a product of the two bases $x,y$ raised to powers $p_i, q_i$ respectively. 
Let the lift of the code be $L = \{\ell_x, \ell_y\}$ so that the unit cells are arranged onto a grid of size $\ell_x \times \ell_y$. We note that $x^{\ell_x} = y^{\ell_y} = 1$.
Then, we can write a cost function $\Theta$ for each of the monomials and expand it to the polynomial after summation:
\begin{align}
    \Theta(A, L) &= \sum_i^T \Theta(x^{p_i} y^{q_i}, L)
    \label{eq:poly_cost_function}\\
    \Theta(x^{p_i} y^{q_i}, L) &= \label{eq:term_cost_function} \begin{cases}
        \tau_0 &\textrm{ if } p=q=0 \\
        \tau_0+ \max\left( \Delta(0,q), \Delta(0,\ell_y-q) \right) &\textrm{ if } p = 0, q\ne 0 \\
        \tau_0+\max\left(\Delta(p,0), \Delta(\ell_x-p,0) \right) &\textrm{ if } p \ne 0, q=0 \\
        \tau_0+\max\left(\Delta(p,q),\Delta(p,\ell_y-q),\Delta(\ell_x-p,q),\Delta(\ell_x-p,\ell_y-q) \right) &\textrm{ if } p\ne 0, q\ne 0
    \end{cases}\\
    \Delta(p,q) &= \sqrt{\frac{1}{k} \sqrt{(u_x p)^2 + (u_y q)^2}}
\end{align}

The function $\Delta(p,q)$ estimates the time it takes to implement the move given by the monomial $x^p y^q$. The grid step lengths $u_x, u_y$ account for anisotropic spacings arising from the atom spacing and the dimensions of the unit cell. For example, the unit cell in Fig.~\ref{fig: unit cell pairings} is wider in the $x$ direction than in the $y$ direction. The constant $k$ relates the distance traveled to the time it takes the AOD to perform the move; according to the experimental implementation \cite{zhao2026towards} the shuttling time with AODs has a square-root dependence on the distance with a proportionality constant $k = 1375 \ {\rm m/s}^2$.

Assuming that monomials are implemented one at a time using four independent AODs, the movement cost for a monomial $x^{p_i} y^{q_i}$ is given by the toric wrap class that is slowest to implement. Depending on how many of $p,q$ are zero, we can have one, two or four toric wrap classes. For the identity monomial $p=q=0$ there is no dependence on distance since the atoms need to move only within their respective unit cells for a small distance.
We also add a constant $\tau_0$ to account for intra-cell movement $\vec{v}_{\mathrm{in-cell}}$, as well as the ramping time of the AOD, i.e., the transfer between AOD and SLM. Note that for small distances (e.g., when $\vec{v}_{\mathrm{cyclic}} = \vec{0}$), the movement time weakly depends on the distance and it is dominated by the transfer time from AOD to SLM.

\begin{figure}
    \centering
    \includegraphics[width=0.8\linewidth]{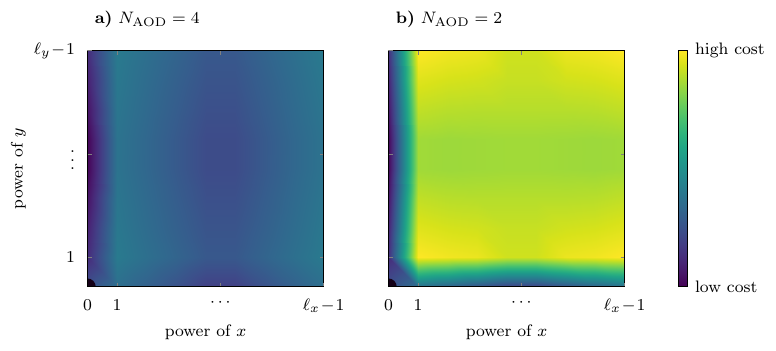}
    \caption{Heatmaps of the cost associated with implementing the different monomials using both 4 AODs (panel \textbf{a}) and 2 AODs (panel \textbf{b}). Both figures are drawn on the same scale. The results are illustrative: the quantitative costs depend on the hardware capabilities and the group sizes of the quantum error-correcting codes.}
    \label{fig:heatmap_comparisons}
\end{figure}

Equation~\eqref{eq:term_cost_function} presumes that all toric wrap classes of a monomial are displaced simultaneously, which for a generic monomial ($p \ne 0$, $q \ne 0$) requires four independent AODs. We now relax this assumption to the experimentally more common case of two AODs, where the four classes must be implemented in two successive batches.

It is convenient to name the wrap classes. For $p \ne 0$, $q \ne 0$ we collect their individual move times in the multiset
\begin{equation}
    \mathcal{D}(p,q,L) = \Big\{\, \Delta(p,q),\; \Delta(p,\ell_y-q),\;
        \Delta(\ell_x-p,q),\; \Delta(\ell_x-p,\ell_y-q) \,\Big\},
    \label{eq:wrap_class_set}
\end{equation}
and write its elements in non-increasing order as
\begin{equation}
    \Delta_{(1)} \;\ge\; \Delta_{(2)} \;\ge\; \Delta_{(3)} \;\ge\; \Delta_{(4)},
    \label{eq:order_statistics}
\end{equation}
suppressing the arguments $(p,q,L)$ for readability; the last line of Eq.~\eqref{eq:term_cost_function} is then simply $\tau_0 + \Delta_{(1)}$.

Every class within a batch is displaced in parallel, so a batch takes as long as its slowest member. Each batch is moreover a separate shuttling cycle and therefore carries its own overhead $\tau_0$: the atoms must be transferred from the SLM to the AOD, moved, and handed back before the following batch can be picked up. Partitioning the four classes into two pairs $\{\mathcal{P}_1,\mathcal{P}_2\}$ consequently costs $2\tau_0 + \max\mathcal{P}_1 + \max\mathcal{P}_2$. Minimizing over the three distinct pairings is immediate: the overhead is pairing-independent, while $\Delta_{(1)}$ is the maximum of whichever pair contains it and therefore contributes to the total regardless, leaving the remaining pair as cheap as possible once the second-slowest class $\Delta_{(2)}$ has been absorbed alongside $\Delta_{(1)}$. The optimal two-AOD schedule thus moves the two slowest classes together and the two fastest classes together, at a cost of $2\tau_0 + \Delta_{(1)} + \Delta_{(3)}$ -- twice the overhead, plus the slowest and the third-slowest wrap class. Monomials with a vanishing exponent carry only one or two wrap classes, which two AODs still discharge in a single batch; those terms pay a single $\tau_0$ and are unchanged. Collecting the cases,
\begin{align}
    \Theta^{(2)}(A, L) &= \sum_i^{T} \Theta^{(2)}\!\left(x^{p_i} y^{q_i}, L\right)
    \label{eq:poly_cost_2aod}\\
    \Theta^{(2)}\!\left(x^{p} y^{q}, L\right) &= \begin{cases}
        \tau_0 &\textrm{ if } p=q=0 \\
        \tau_0 + \max\left( \Delta(0,q), \Delta(0,\ell_y-q) \right)
            &\textrm{ if } p = 0,\ q \ne 0 \\
        \tau_0 + \max\left( \Delta(p,0), \Delta(\ell_x-p,0) \right)
            &\textrm{ if } p \ne 0,\ q = 0 \\
        2\tau_0 + \Delta_{(1)} + \Delta_{(3)}
            &\textrm{ if } p \ne 0,\ q \ne 0
    \end{cases}
    \label{eq:term_cost_2aod}
\end{align}

Both cost functions are instances of a single expression. Let $N$ be the number of available AODs and let $W$ be the number of wrap classes of the monomial: $W = 4$ generically, $W = 2$ when exactly one exponent vanishes, and $W = 1$ for the identity monomial, whose single class is trivial, $\Delta_{(1)} = 0$. The classes are discharged in $B = \lceil W/N \rceil$ batches of at most $N$, each batch greedily filled with the slowest classes not yet moved and each paying its own overhead, so that the $r$-th batch costs $\tau_0 + \Delta_{(rN+1)}$ and
\begin{equation}
    \Theta^{(N)}\!\left(x^{p} y^{q}, L\right)
        = B\,\tau_0 + \sum_{r=0}^{B-1} \Delta_{(rN+1)},
        \qquad B = \left\lceil \frac{W}{N} \right\rceil .
    \label{eq:term_cost_general}
\end{equation}
Taking $N \ge 4$ gives $B = 1$ and leaves $\tau_0 + \Delta_{(1)} = \tau_0 + \max \mathcal{D}$, recovering Eq.~\eqref{eq:term_cost_function}; $N = 2$ recovers Eq.~\eqref{eq:term_cost_2aod}; and $N = 1$ gives the fully serial cost $4\tau_0 + \sum_{r} \Delta_{(r)}$ for a monomial with four wrap classes.

The two-AOD cost therefore exceeds the four-AOD cost by $\tau_0 + \Delta_{(3)}$ throughout the interior of the exponent grid -- a strictly positive penalty, since all four wrap classes of a generic monomial have non-vanishing move times -- with equality exactly on the axes $p = 0$ or $q = 0$, where halving the AOD budget costs nothing. Fig.~\ref{fig:heatmap_comparisons} contrasts the two cost fields on a common color scale: for the $\ell_x = \ell_y = 9$ lift at grid-step ratio $u_x{:}u_y = 2{:}1$ the mean per-monomial cost rises from $0.372\,\mathrm{ms}$ to $0.608\,\mathrm{ms}$, a factor of $1.63$, and the worst-case monomial from $0.425\,\mathrm{ms}$ to $0.723\,\mathrm{ms}$. The penalty is concentrated in the interior, where all four wrap classes are present, while the axes -- the monomials that shift along a single lattice direction -- remain entirely unaffected.

The cost function above neglects batching or merging of moves for monomials implemented before the same pulse, such as those in Table~\ref{tab:ccz-budget}. Although these estimates are not tight, they provide an upper bound on the movement time that can guide the initial co-design of an error-correcting code for a given application.

\end{document}